%% file: coveringME.tex
\documentclass[11pt]{article}
\usepackage{prettyref,cleveref}
\usepackage{tabularx,amssymb,nicefrac,amsmath,multirow}
\usepackage{amsthm,amsfonts} 
\usepackage{epsfig} 
\usepackage{latexsym,nicefrac,bbm}
\usepackage{xspace}
\usepackage{color,fancybox,graphicx,url,subfigure,fullpage}
\usepackage{tabularx} 
\usepackage[dvipsnames,usenames,table]{xcolor}
\usepackage[backref,colorlinks=true,urlcolor=blue,linkcolor=RoyalBlue,citecolor=OliveGreen]{hyperref}
\usepackage{algorithmicx}
\usepackage{algorithm}
\usepackage{algpseudocode}
\usepackage{verbatim} 
\usepackage{times}
\usepackage{enumitem}
\usepackage{booktabs}
\usepackage[font=small,labelfont=bf]{caption}
\usepackage[margin=1in]{geometry}
\usepackage[numbers,sort&compress]{natbib} 
\usepackage{thmtools}
\usepackage{thm-restate}
\declaretheorem[name=Theorem]{theorem}
\declaretheorem[name=Lemma]{lemma}

\theoremstyle{definition}\newtheorem{definition}{Definition}

\newtheorem{claim}[theorem]{Claim}
\theoremstyle{definition}\newtheorem{remark}{Remark}

\newcommand\M{\rule[-1.0ex]{0pt}{0pt}} 

\newcommand{\suppress}[1]{}

\input{vijay_edited}

\begin{document}

\title{A New Class of Combinatorial Markets with Covering Constraints:
Algorithms and Applications}

\author{Nikhil R. Devanur\thanks{Microsoft Research. Part of the work was done when author was visiting Simons Institute, UC Berkeley.
		nikdev@microsoft.com} 
	\and Jugal Garg\thanks{University of Illinois at Urbana-Champaign. jugal@illinois.edu}
	\and  Ruta Mehta\thanks{University of Illinois at Urbana-Champaign. Part of the work was done when author was visiting Simons
	Institute, UC Berkeley. rutameht@illinois.edu} 
	\and Vijay V. Vazirani\thanks{College of Computing, Georgia Institute of Technology, Atlanta. vazirani@cc.gatech.edu}
	\and  Sadra Yazdanbod\thanks{College of Computing, Georgia Institute of Technology, Atlanta. syazdanb@cc.gatech.edu}}

\date{}

\maketitle

\begin{abstract}

We introduce a new class of combinatorial markets in which agents have covering constraints over resources required and are interested
in delay minimization. Our market model  is applicable to several settings including scheduling, cloud computing, and communicating over
a network.
This  model is quite different from the traditional models, to the extent that neither do the classical equilibrium existence
results seem to apply to it nor do any of the {efficient} algorithmic techniques developed to compute equilibria 
seem to {apply directly}. 
We give a proof of existence of equilibrium and a polynomial time algorithm for finding one,
drawing heavily on techniques from LP duality and submodular minimization.
We observe that in our market model, the set of equilibrium prices could be a connected, non-convex set (see figure below).
To the best of our knowledge, this is the first natural example of the phenomenon where the set of solutions 
could have such complicated structure, 
yet there is a combinatorial polynomial time algorithm to find one. 
Finally, we show that our model inherits many of the fairness properties of traditional equilibrium models. 

\end{abstract}

\maketitle
\begin{center}
\begin{figure}[!h]
	\vskip -1.7cm
\hspace{2.9cm}
	\includegraphics[width=0.75\textwidth]{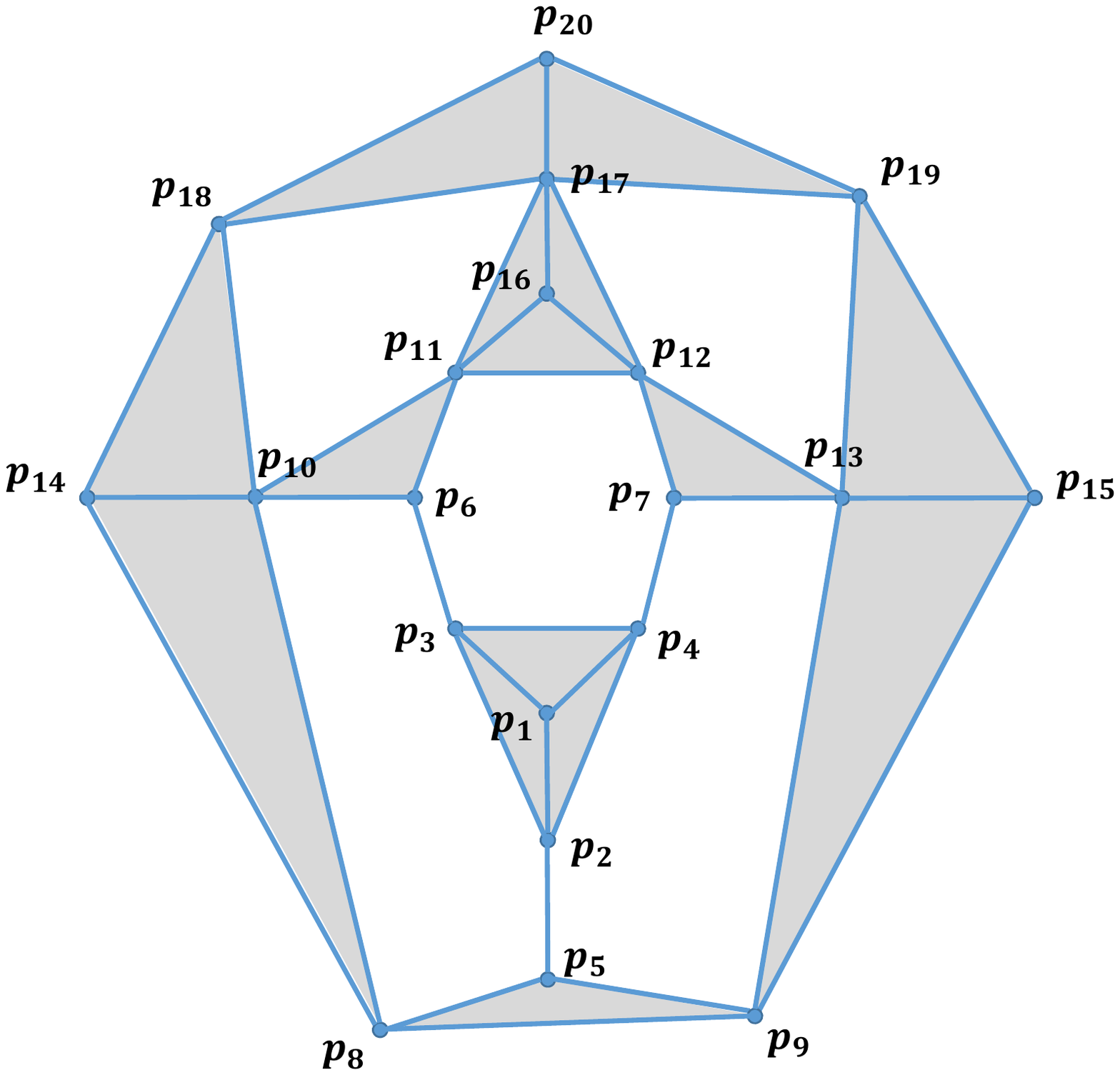}
	\vskip -1.2cm
	\caption{An example of a non-convex set of equilibria. See Table \ref{fig:hole5} for an explanation.}
	\vskip -2cm
\end{figure}
\end{center}
\newpage
\setcounter{page}{1}

\section{Introduction}
\label{sec.intro}

In a free market economy, prices naturally tend to find an ``equilibrium'' under which there is parity between
supply and demand. The power of this pricing mechanism is well explored and understood in economics: It allocates resources efficiently since prices send strong
signals about what is wanted and what is not, and it prevents artificial scarcity of goods while at the same time ensuring that goods that are 
truly scarce are conserved \cite{MasColell-book}. In addition, it ensures that 
the allocation of goods is Pareto optimal. Hence it is beneficial to both consumers and producers. 
Furthermore, equilibrium-based mechanisms have been
designed even for certain applications which do not involve any exchange of money but require fairness properties such as envy-freeness and the sharing incentive property; 
a popular one being CEEI\footnote{Competitive Equilibrium with Equal Incomes} \cite{Moulin-book}.
Today, with the surge of markets on the Internet, in which pricing and allocation are done in a centralized manner via 
computation, an obvious question arises: can we apply insights gained from traditional markets to these new markets to accrue similar benefits? 
This (in addition to other motivations) has led to a long line of work in the \rh{TCS} 
community on the computation of economic equilibria; see Appendix \ref{sec:marketapps} for more details.

In this paper, we define a broad class of market models that are appropriate for modeling several new markets, including scheduling, cloud computing, and bandwidth allocation \rh{in networks}.
A common feature of our markets is that these are resource allocation markets in which each agent desires a {\em specific amount of resources} to complete a task, 
i.e., each agent has a {\em covering constraint}.
If the agent does not get all the resources requested, then she will not be able to complete the task and hence has no value for
this partial allocation.
With several agents vying for the same set of resources, a new parameter that becomes crucially important is the {\em delay} experienced by
agents. This naturally leads to a definition of supply and demand, as well as pricing and allocation, based on {\em temporal considerations}.

We define an equilibrium-based model for pricing and allocation in these markets. 
Our model is fundamentally different from traditional market models: Each agent needs only a bounded amount of resources to finish her tasks and has no use for more,
and her utility, which corresponds to the delay she experiences, also has a finite maximum value, i.e., her ``utility function'' satiates. On the other hand, traditional models satisfy non-satiation, i.e.,
no matter what bundle of goods an agent gets, there is a way of giving her additional goods so her utility strictly increases. Non-satiation turns out to 
be a key assumption in the Arrow-Debreu Theorem, which established existence of equilibrium in traditional markets. Despite this, we manage to give an existence proof for our model. Additionally, we prove that all the above-stated benefits of equilibria, including the fairness properties of CEEI,
continue to hold for our model. 

We next address the issue of computing equilibria in our model. Rubinstein \cite{Rubinstein} recently showed that
computing an equilibrium in our general model is PPAD-hard. For this reason, we defined a sub-model for which we 
seek an efficient algorithm. This sub-model is of interest in applications, including the three mentioned above. However,
it turns out that equilibria of this sub-model have a different structure than those of models
for which polynomial time algorithms have been designed. For instance, we give examples in which
the set of equilibrium prices is non-convex.
Hence techniques used for designing
polynomial time algorithms for traditional models, such as the primal-dual method and convex programming, are not applicable. 
Our algorithms are based on new ideas: we make heavy use of LP duality and the way optimal solutions to LPs change with changes in certain parameters. 
Submodular minimization, combined with binary search, is used as a subroutine in this process. 

In summary, for the main applications stated above, our market-based model admits equilibrium prices, which can moreover be computed in polynomial time and
we show that our algorithm is incentive compatible.
In addition, the equilibrium allocations satisfy a range of fairness properties. Considering the many favorable properties of market-based models and the availability of massive computing power for computing equilibria, we believe they will play an important role in markets on the Internet.

\paragraph{Organization.} 
We define the market model and state our main results in Section \ref{sec:model}. In this section, we  define the notions of \emph{strong feasibility}, under which we establish
existence of equilibrium, and \emph{extensibility,} which gives the sub-model for which we give a polynomial time algorithm.
We also discuss properties of fairness and incentive compatibility of our solution. 
In Section \ref{sec:cloud} we describe the algorithm for a special case in a \emph{scheduling} setting, in order to convey the main ideas. The algorithm in its full generality, and an overview of the analysis are in Section \ref{sec:general}. 
Section \ref{sec:examples} contains all the different examples referred to, and also a description of a run of the algorithm for some examples. 
\rh{The appendices contain more details on related work (\ref{sec:marketapps}), special cases of our model (\ref{app:specialcases}), existence of equilibria (\ref{sec:existenceproof}), connection of our algorithm to Myerson's ironing in the {\em scheduling} case (\ref{app:myerson}), equilibrium characterization for the general model (\ref{sec:eqchar}), proofs missing from the main paper (\ref{sec:fullproofs}), and fairness and incentive compatibility properties (\ref{app:prop}).}


\section{Model and Main Results}\label{sec:model}
We introduce a combinatorial version of the well studied Fisher market model \citep{DPSV,BSAD}. 
\rh{In market $\CM$,} let $A$ be a set of $n$ agents, indexed by $i$, and $G$ be a set of $m$ divisible goods, indexed by $j$. 
We represent an allocation of goods to agents using the variables $\xij \in \Rplus, i\in A, j \in G$. 
Each agent $i\in A$ wants to procure goods that satisfy a set of {\em covering constraints}, $C$, where $C$ is a set indexing the constraints ($C$ is the same for all agents for ease of notation). 
\begin{equation} \tag{CC$(i)$}\label{eq:CC} \textstyle \forall~k\in C, \sum_{j\in G} \aijk \xij \geq \rik, \ \ \mbox{ and }\ \ \forall~j \in G, \xij \geq 0\end{equation}

The objective of each agent is to minimize the ``delay'' she experiences, 
while meeting these constraints. 
%
%
We refer to the term $\dij$ as the delay faced by agent $i$ on using good $j$, and the terms $\rik$s as the ``requirements''; $\dij$s and $\rik$s are assumed to be 
non-negative.
Agent $i$ wants an allocation that optimizes the following LP. 

\begin{equation} \tag{Delay LP$(i)$}\label{eq:coveringLP} \textstyle \min \sum_{j\in G} \dij \xij \text{ s.t.} \end{equation}
\[ \textstyle\forall~k\in C, \sum_{j\in G} \aijk \xij \geq \rik.  \]
\[\textstyle \forall~j \in G, \xij \geq 0 .\]

We use the notation $\d_i := (\dij)_{j\in G}, \r_i := (\rik)_{k\in C}$,  $A_i := (\aijk)_{j\in G, k\in C}$, 
$\x_i := (\xij)_{j \in G}$, and $X := (\x_i)_{i \in A}$. 
Although our results hold for any LP, the most interesting cases are when the constraints are covering constraints, i.e., 
the matrix $A_i$ has only non-negative entries. 

We will use a market mechanism to allocate resources. Let $p_{j}\in \Rplus$ denote the price per unit amount of good $j$,
and assume agent $i$ has a total budget of $m_i\in \Rplus$.  Then, as is standard in
Fisher markets, the bundle $\x_i$ that the agent may purchase is restricted
by,
\begin{equation}\tag{Budget constraint$(i)$}\label{eq:budget}
\textstyle \sum_{j\in G}  p_{j}\xij   \leq m_i. 
\end{equation}
Allocation $\x_i$ is an {\em optimal allocation (bundle)} of agent $i$
relative to prices $\pp := (p_j)_{j \in G}$, 
if it optimizes  LP (\ref{eq:coveringLP}) with an additional budget constraint (\ref{eq:budget}). 
%
%
Each good has a given supply which, after normalization, may be assumed to be equal to $1$. 
The allocation needs to be \emph{\supplyrespecting}, that is, it has to satisfy the supply constraints: 
\begin{equation}\tag{Supply constraints}\label{eq:supply}
\textstyle \forall~j \in G, \sum_{i\in A}\xij \leq 1.
\end{equation} 
Finally, a \supplyrespecting allocation $X$ and prices $\pp$ are a {\em market equilibrium} \rh{of $\CM$} iff 
\vspace{-0.2cm}
\begin{enumerate}
	\item Each agent gets an optimal allocation relative to prices $\pp$.
	\item If some good $j\in G$ is not fully allocated, i.e., $\sum_{i \in A} \xij < 1$, then $p_{j} =	0$.
\end{enumerate}
The equilibrium condition requires that each agent does the best for herself,
regardless of what the other agents do or even what the supply constraints are. From the perspective of the goods, the aim is market
clearing (rather than, say, profit maximization)\footnote{However, our algorithm will find an equilibrium where every agent spends all of her budget, and thereby it maximizes the profit automatically.}. Some goods may not have sufficient demand and therefore we may not be able to clear
them. This is handled by requiring these goods to be priced at zero. 

\rh{In Theorem \ref{thm:lpl2eq} we obtain characterization of equilibrium in this general model in terms of solutions of a parameterized linear program that has one parameter per agent.}

\subsection{Existence of equilibria.} 
We  show how the above model is a special case of the classic Arrow-Debreu
market model with quasi-concave utility functions in Appendix
\ref{sec:quasiconcave}. Unfortunately, these utility functions do not satisfy
the ``non-satiation'' condition required by the Arrow-Debreu theorem for the
existence of an equilibrium: utility does not increase beyond a point even if
additional goods are allocated. In fact, equilibrium doesn't always exist
for all covering LPs, \rh{as shown via a simple example in Section \ref{sec:examples}, Figure \ref{fig:non-exist}.} And therefore next we identify conditions under which it
does exist; the example in Figure \ref{fig:non-exist} shows how this condition
is necessary.  


The equilibrium condition requires at a minimum that there exists a \supplyrespecting allocation that also satisfies \ref{eq:CC} of all the agents.
In fact, it is easy to see that a somewhat stronger feasibility condition is necessary: suppose that a subset of agents all have high
budgets while the remaining agents have budgets that are close to 0. 
Then at an equilibrium, agents in the former set get their ``best'' goods, which means that 
whatever supply remains must be sufficient to allocate a feasible bundle for the remaining agents. 

We require a similar condition for all {\em minimally feasible} allocations, {\em i.e.,} 
an allocation $\x_i$ such that reducing amount of any good would make \ref{eq:CC} infeasible.
We call this condition strong feasibility. 


\begin{definition}[Strong feasibility] \label{def:strongfeasibility}
Market $\CM$ satisfies {\em strong feasibility} if any minimally feasible and \supplyrespecting solution to a subset of agents can be extended to a feasible and \supplyrespecting allocation to the entire set. 
%
%
	Formally, $\forall~S \subset A$, and $\forall~(\x_i)_{i\in S}$  that 
 are \emph{minimally} feasible for $\eqref{eq:CC}_{i\in S}$ and
are \supplyrespecting  (with $\xij = 0~\forall~i\in S^c$), 
	$\exists$ solutions $(\x_i)_{i \in S^c}$ that  are feasible for $\eqref{eq:CC}_{i\in S^c}$ and $(\x_i)_{i \in (S \cup S^c)}$ is \supplyrespecting. 
\end{definition}
\begin{restatable}{theorem}{existence} [Strong feasibility implies existence of an equilibrium] \label{thm:existence}
If \eqref{eq:CC}$_{i \in A}$ of market $\CM$ satisfies strong feasibility, then $\exists$ an allocation $\xx$ and prices $\pp$ that constitute a market equilibrium of $\CM$. 
\end{restatable}
The proof of this theorem is in Appendix \ref{sec:existenceproof}. 
Strong feasibility is quite general in the following sense: it is satisfied if there is a ``default'' good that has a large enough capacity and may have a  large delay but occurs in every constraint with a positive coefficient. 
In other words, any agent's covering constraints may all be met by  allocating sufficient quantity of the default good. 

\subsection{Efficient computation.}
Ideally we would want to design an efficient algorithm for markets with {\em Strong feasibility} condition, however this problem turns out to be PPAD-hard \cite{Rubinstein}. 
%
%
In order to circumvent this hardness, we define a stronger condition called \emph{extensibility}, and design a polynomial time algorithm to compute a market equilibrium under it. Extensibility requires that any ``optimal allocation'' to a subset of agents can be ``extended'' to an ``optimal allocation'' for a set that includes one extra agent.
Hence this is a matroid like condition. For this we first formally define ``optimal allocation'' for a subset of agents. 

\begin{definition}\label{def:bestfor} 
For any subset of agents $S \subseteq A$, we say that an allocation $X$ is {\em \bestfor} $S$ if (i) it satisfies $\eqref{eq:CC}_{i\in
S}$, (ii) it is \supplyrespecting, and (iii) it  minimizes $\sum_{i \in S} \d_i
\cdot \x_i $. (Observe that $X$ may not be optimal for individual agents in $S$.)
\end{definition}

\begin{definition}[Extensibility]\label{def:extensibility}
	Market $\CM$ satisfies  \emph{extensibility} if 
	$~\forall~S \subset A$,  given an allocation $X $ that is \bestfor   $S$, 
	the following holds:
	for any $i \in S^c$, $\exists$ an allocation $X' $ that is \bestfor $S'=S\cup \{i\}$, while not 
	changing the delay of the agents in $S$, i.e., 
	$\d_i\cdot \x'_i = \d_i\cdot \x_i,\ \forall i \in S$. In other words, total delay cost of agents in $S'$ can be minimized without
	changing the delay cost of agents in $S$. 
\end{definition}
Extensibility seems somewhat stronger than strong feasibility, but the two conditions are formally incomparable; \rh{see example in Section \ref{sec:examples}, Figure \ref{fig:non-exist}}. 
In Section \ref{sec:specialcases} we show that extensibility condition captures many interesting problems as special cases. 
Another mild condition we need is that there is enough demand for goods from each agent, otherwise an agent with very little
	requirement but huge amount of money may drive everyone else out of the market. 
\begin{definition}[Sufficient Demand]\label{def:ED}
Market $\CM$ satisfies sufficient demand if under zero prices, the optimal bundle of each agent contains some good that is demanded more than its supply, 
i.e., for an optimal solution $\x_i$ to (\ref{eq:coveringLP}), there exists a good $j$ such that $x_{ij}>1$. 
\end{definition}
\rh{Even for very simple markets, e.g., Tables \ref{fig:hole1} and \ref{fig:hole5} in Section \ref{sec:examples}, the set of equilibria may turn out to be highly non-convex. Therefore techniques used to obtain polynomial time algorithms for traditional models are not applicable. In Section \ref{sec:general} we design a polynomial time algorithm by making a heavy use of parameterized LP, duality and submodular minimization, and obtain the following result.}
\begin{restatable}{theorem}{algo} [Extensibility and sufficient demand implies polynomial time algorithm] \label{thm:algorithm}
	There is a polynomial time algorithm that computes a market equilibrium allocation $X$ and prices $\pp$ for any market
	$\CM$ that satisfies extensibility and sufficient demand.
\end{restatable}
\rh{Since the algorithm is quite involved, we first convey the main ideas through a special case of {\em scheduling} in Section \ref{sec:cloud}. We show the run of the main algorithm on an example in Section \ref{sec:examples}, Figure \ref{fig:sp}. }

%
%

\subsection{Applications} \label{sec:specialcases}
As a consequence of the above theorem we get polynomial time algorithms for the following special cases. (The proofs that these satisfy extensibility are in Appendix \ref{app:specialcases}.)

\paragraph{Scheduling.}
The agents are jobs that require $d$ different types of machines, and the set of machines of type $k$ is $M_k$; the machines are the goods in the market. 
Each agent needs $r_{ik}\in \Rplus$ units of machines in all of type $k$, which is captured by the covering constraint $\sum_{j \in M_k} x_{ijk} \ge r_{ik},\ \forall k\in [d]$. All agents experience the same delay $d_{jk}$ from machine $j$ in type $k$.
Assume that the number of machines in each type $k$ is greater than the total requirements of the agents, $\sum_{i \in A} r_{ik}$. In reality different machines in this model may represent actually 
different machines, or the same machine at different times.  The main motivation for this problem is scheduling in the cloud
computing context, but it also captures other client-server scenarios such as crowdsourcing.
An earlier version of this paper designed an algorithm only for this case \cite{devanur2015market}. 

Even for this simple case with only one type, we observe that the set of equilibria may form a connected non-convex set. 
The non-convexity example shown in Section \ref{sec:examples}, Tables \ref{fig:hole1} and \ref{fig:hole5}, are instances of this setting.

\paragraph{Restricted assignment with laminar families -- Different arrival times.} The above basic scheduling setting can be generalized to the following restricted assignment case, where job $i$ is allowed to be processed only on a subset of all the machines  $S_{ik} \subseteq M_k$ for type $k$.  We need the $S_{ik}$s to form a laminar\footnote{A family of subsets is said to be \emph{laminar} if any two sets $S$ and $T$ in the family are either disjoint, $S\cap T = \emptyset$, or contained in one another, $S\subseteq T$ or $T\subseteq S$. } family within each type, and in addition, we require that the machines in a larger subset have lower delays. That is, if for some two agents $i,i' \in A$,  $S_{i'k} \subset S_{ik}$ then $\max_{ j \in S_{ik} \setminus S_{i'k}}d_{jk} \leq \min_{j'\in S_{i'k}} d_{j'k}$ for each type $k$. This helps to model the important condition in the cloud computing context that jobs may arrive at different times.


\paragraph{Network Flows.}
The goods are edges in a network, where each edge $e$ has a certain (fixed) delay $d_e$. 
Each agent $i$ wants to send $r_i$ units of flow from a source $s_i$ to a sink $t_i$, 
and minimize her own delay (which is a min-cost flow problem). 
We show that if the network is series-parallel and the source-sink pair is common to all agents, then the instance satisfies extensibility.
This is similar to the basic scheduling example in that there is a sequence of paths of increasing delay, 
but the difference here is that we need to price edges and not paths.  The difficulty is that paths share edges and hence the edge
prices should be co-ordinated in such a way that the path prices are as desired. 
There are networks that are not series-parallel but still the instance satisfies extensibility.
We show one such network (and also how our algorithm runs on it) in \prettyref{fig:sp}. 
For a general network, an equilibrium may not exist; we give such an example in Figure \ref{fig:non-exist}. 

A generalization of all these special cases that still satisfies extensibility is as follows: take any number of independent copies of 
any of these special cases above. 
E.g., each agent might want some machines for job processing, as well as send some flows through a network or a set of networks, but 
have a common budget for both together. Our algorithm works for all such cases. 

\subsection{Properties of equilibria}
\paragraph{Fairness:}
We first discuss an application of our market model to fair division of goods, where there are no monetary transfers involved. This captures scenarios where the goods to be shared are commonly owned, such as the computing infrastructure of a large company to be shared among its users. 
A standard fair division mechanism is competitive equilibrium from equal incomes (CEEI) \cite{Moulin-book}. 
This mechanism uses an equilibrium allocation corresponding to an instance of the market where all the agents have the same budget. 
This can be generalized to a weighted version, where different agents are assigned different budgets based on their importance. 

The fairness of such an allocation mechanism follows from the following properties of equilibria \rh{shown in Appendix \ref{app:fair} for the general model}. 1. The equilibrium allocation is {\em Pareto optimal}; this an analog of the first welfare theorem for our model. 2. The allocation is envy-free; since each agent gets the optimal bundle given the prices and the budget, he doesn't envy the allocation of any other agent. 3. Each agent gets a ``fair share'': the equilibrium allocation Pareto-dominates an ``equal share" allocation, where each agent gets an equal amount of each resource. This property is also known as \emph{sharing incentive} in the scheduling literature \cite{Ghodsi2011}. 4. Incentive compatibility (IC): the equilibrium allocation is incentive compatible ``in the large'', where no single agent is large enough to significantly affect the equilibrium prices. In this case, the agents are essentially price takers, and hence the allocation is IC. 
We also show a version of IC when the market is not large. We discuss this in more detail below. 

\paragraph{Incentive Compatibility:} 
In the quasi-linear utility model,  
an agent maximizes the valuation of the goods she gets minus the payment. 
In the presence of budget constraints, \citet{dobzinski2012multi} show that no anonymous\footnote{Anonymity is a very mild restriction, which disallows favoring any agent based on the identity.} IC mechanism can also be Pareto optimal, even when there are just two different goods. 
In the context of our model, a quasi-linear utility function 
specifies an ``exchange rate'' between delay and payments, and wants to minimize a linear combination of the two. 
We show in Appendix \ref{app:QL} that the impossibility extends to our model via an easy reduction to the case of \citet{dobzinski2012multi}.

In the face of this impossibility, we show the following second best guarantee in Appendix \ref{app:IC}. 
For the scheduling application mentioned above, we show that our algorithm as a market based mechanism is IC in the following sense: 
non-truthful reporting of $m_i$ and $r_{ik}$s can never result in an allocation with a lower delay.   
A small modification to the payments, keeping the allocation the same, makes the entire mechanism incentive compatible for the model in which agents want to first minimize their delay and subject to that, minimize their payments. 

The first incentive compatibility assumes that utility of the agents is only the delay, and does not depend on the money spent (or saved). 
Such utility functions have been considered in the context of online advertising \citep{borgs2007dynamics,feldman2007budget,muthukrishnan2010stochastic}. 
It is a reflection of the fact that companies often have a given budget for procuring compute resources, 
and the agents acting on their behalf really have no incentive to save any part of this budget. 
In the fair allocation context (CEEI), this gives a truly IC mechanism, since the $m_i$s are determined exogenously, and hence are not private information.

The second incentive compatibility does take payments into account, but gives a strict preference to delay over payments. Such preferences are also seen in the online advertising world, where advertisers want as many clicks as possible, and only then want to minimize payments. 
The modifications required for this are minimal, and essentially change the payment from a ``first price'' to a
``second price'' wherever required. 


\section{Scheduling on a single machine}\label{sec:cloud}
Our  algorithm for the general setting is quite involved, therefore we first present it for a very special case in a scheduling setting mentioned in Section \ref{sec:specialcases}. The basic building blocks and the structure of the algorithm and the analysis are reflected in this case. 
In Section \ref{sec:examples}, we describe the run of this algorithm on the example in Table \ref{fig:hole1}.
\rh{We note that the formal proofs are given only for the general case and not for this section.}

Suppose that there is just one machine and a good is this machine at a certain time $t \in \integers_+$, 
which we refer to as slot $t$. 
The set of goods is therefore $G = \integers_+$ and we index the goods by $t$ instead of $j$ as before.
Further, assume that the delay of slot $t$ is just $t$, i.e., $\forall~i\in A, d_{it} = t$.  
Each agent $i$ requires a certain number of slots to be allocated to her, as captured by the covering  constraint $\sum_{t\in \integers_+} x_{it} \geq r_i$, for some $r_i \in \integers_+$. 
We denote the sum of the requirements over a subset $S\subseteq A$ of agents as $r(S) := \sum_{i \in S} r_i$. Recall that the budget of agent $i$ is $m_i$, and similarly $m(S) := \sum_{i \in S} m_i$. 
We will show that equilibrium prices are characterized by the following conditions.\footnote{For how this equilibrium characterization leads to an analogy with Myerson's ironing 
	for a special case of this setting, with $r_i=1$ for all $i\in A$, see Appendix \ref{app:myerson}. } 
\begin{enumerate}
	\item The prices form a piecewise linear convex decreasing curve. Let the linear pieces {(segments)} of this curve be numbered $1,2,\ldots,k,\ldots$, from \emph{right to left}. \label{enum:convexprices} 
	\item There is a partitioning of the agents into sets $S^1, S^2, \ldots, S^k, \ldots, $ where the number of slots in $k$th segment is $r(S^k)$. Note that since $r_i$s are integers so are the $r(S^k)$s. \label{enum:slots}
 	\item The sum of the prices of slots in $k$th segment equals $m(S^k)$. \label{enum:equals}
	\item For any $S\subset S^k$, the total price of the first $r(S)$ slots of the segment $\ge m(S)$, since otherwise these slots would be over demanded. This is equivalent to saying that the total price of the {last} $r(S)$ slots in this segment $ \leq m(S)$. \label{enum:less}
\end{enumerate}
The above only characterizes equilibrium \emph{prices}. We will show that Conditions \ref{enum:equals} and \ref{enum:less} imply that there exists an \emph{allocation} of the slots in segment $k$  to the agents in $S^k$ such that both their requirements and budget constraints are satisfied. Such allocations can then be found by solving the following feasibility LP \eqref{eq:feasibilityLP}. In this LP, segment $k$ corresponds to the interval $[T^k ,T^{k-1}]$.

\begin{equation}\label{eq:feasibilityLP} 
\begin{aligned}
\textstyle \forall~i \in S^k : &\quad \textstyle \sum_{t \in [T^k  ,T^{k-1}] } \xit \geq r_i \\
\textstyle  \forall~i \in S^k : &\quad \textstyle \sum_ {t \in [T^k ,T^{k-1}] } p_t \xit \leq m_i  \\
\textstyle  \forall~t \in [T^{k} ,T^{k-1}] : &\quad \textstyle \sum_{i \in S^k}  \xit \leq 1 \\
\textstyle \forall~i \in S^k,\textstyle \forall~t \in [T^{k} ,T^{k-1}] : &\quad \textstyle  x_{it} \geq 0
\end{aligned}
\end{equation}

\begin{algorithm}[t]
	\caption{Algorithm to compute market equilibrium for scheduling}\label{alg.scheduling}
	\begin{algorithmic}[1]
		\State Input: $A,(m_i)_{i \in A},(r_i)_{i\in A}$
		\State Initialize $\CA'\la\CA$, $\plow \la0, T^{0} \la r(\CA)+1, \forall~t\geq T^{0}, p_{t}\la0$ and $k\la1$ 
		\While{$\CA'\neq \emptyset$}
		\State $S^k \la $ \nextseg$(\plow, \CA',(m_i)_{i \in A'},(r_i)_{i\in A'})$
		\State $\lambda_{S^k} \la 2\frac{m(S^k) - \plow r(S^k)}{r(S^k) (r(S^k) +1)}$ 
		\State $T^{k}\la T^{k-1} - r(S^k)$
		\State $\forall~t \in [T^{k},T^{k-1}],$ set $p_{t} \la \plow + (T^{k-1}-t)  \lambda_{S^k}$  
		\State Compute allocations $\x_i$ for all $i \in S^k$ by solving  LP \eqref{eq:feasibilityLP}
		\State Update  $\plow\la  p_{T^{k}}$, $\CA' \la\CA'\setminus S^k$, and $k\la k+1$
		\EndWhile
		\State Output allocations $\xx$ and prices $\pp$.
	\end{algorithmic}
\end{algorithm}

We now describe the algorithm, which is formally defined in Algorithm \ref{alg.scheduling}. It iteratively computes $S^k$, starting from $k=1$:
the last segment that corresponds to the latest slots is computed first, and then the segment to its left, and so on.
Inductively, suppose we have computed segments numbered 1 up to $k-1$. 
Let $\plow$ be the price of the earliest slot in segment $k-1$, 
and let $A' = A\setminus \{S^1 \cup \cdots \cup S^{k-1}\}$. 
For any $S \subseteq A'$, consider the sum of the prices of $r(S)$ consecutive slots to the left of this slot, forming a line segment with slope $-\lambda$ (see Figure \ref{fig:segment});  
this sum is  
\[\textstyle  \plow  r(S) + \lambda \frac{ r(S) (r(S) +1)} 2.  \] 

\begin{figure}[!h]
	\vspace{-2cm}
	\centering
	\includegraphics[width=0.7\linewidth]{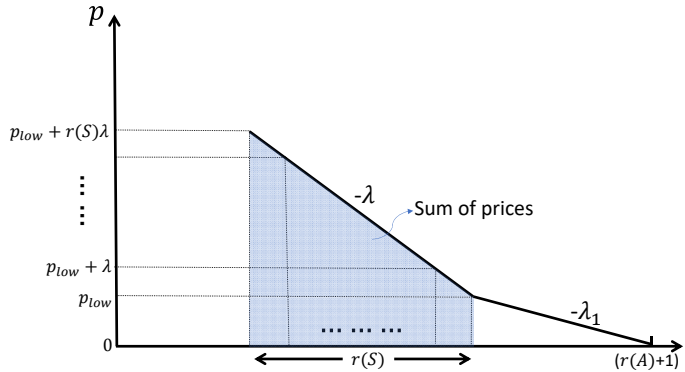}
	\vspace{-1.4cm}

	\caption{Prices on a Segment for set $S\subseteq A'$}\label{fig:segment}
\end{figure}
	
Then for any $S$, one can solve for the $\lambda$ where this would be equal to $m(S)$; we define this as a function of $S$. 
\[\textstyle  \lambda_S := 2\frac{m(S) - \plow r(S)}{r(S) (r(S) +1)}.\]
The next segment is defined to be the one with the smallest slope: 
\[ \textstyle S^k = \textrm{\nextseg}(\plow, \CA',(m_i)_{i \in A'},(r_i)_{i\in A'}) := \arg\min_{S\subseteq A'} \lambda_S.\]
With this definition of the next segment, and with prices for the corresponding slots set to be linear with slope $-\lambda_{S^k}$, 
it follows that Conditions \ref{enum:equals} and \ref{enum:less} are satisfied. 

%
It is not immediately clear how to minimize $\lambda_S$; the function need not be submodular, for instance.
The main idea here is to do a binary search over $\lambda$, as defined in Algorithm \ref{alg.leastseg}. 
Consider the function $f_{\plow,\lambda}$ as defined in line \ref{eq.fl} of this algorithm, 
and notice that $f_{\plow,\lambda} $ is decreasing in $\lambda$. 
From the preceding discussion, it follows that the segment $S^*$ we seek is such that 
$f_{\plow,\lambda_{S^*}} (S^*) = 0 $ and $\forall S \subset A, f_{\plow,\lambda_{S^*}} (S) \geq 0 $.
This implies  that $S^*$ must minimize $f_{\plow,\lambda_{S^*}}$ over all subsets of $A'$. 
Thus, given any $\lambda$ and a minimizer of $f_{\plow,\lambda}$, we can tell whether the desired $\lambda_{S^*}$ is above or below this $\lambda$, and a binary search gives us the desired segment. A minimizer of $f_{\plow,\lambda}$ can be found efficiently 
since this is (as we will show) a submodular function. 

\begin{algorithm}[t]
\caption{Subroutine \nextseg$(\plow, \CA',(m_i)_{i \in A'},(r_i)_{i\in A'})$ }\label{alg.leastseg}
\begin{algorithmic}[1]
	\State Initialize $\lambda^0\la 0, \lambda^1 \la \max_{i \in A'} m_i$
	\State Define $f_{\plow, \lambda}(S) := m(S) - \plow  r(S) - \lambda r(S) (r(S) +1)/2$ \label{eq.fl} 
	\State Set $S^0 \in \argmin_{S \subseteq A', S\neq \emptyset} f_{\plow,\lambda^0}(S)$ and $S^1 \in \argmin_{S \subseteq A',S\neq \emptyset} f_{\plow,\lambda^1}(S)$
	\While{$S{^0}\neq S{^1}$}
		\State Set $\lambda^* \la\frac{\lambda^0 + \lambda^1}{2}$ and $S^* \in \argmin_{S\subseteq A', S\neq\emptyset} f_{\plow,\lambda^*}(S)$  
		\If{$f_{\plow,\lambda^*}(S^*)>0$}
			 Set $\lambda^0 \la \lambda^*$ and $S^0 \la S^*$
		\Else\ 
			Set $\lambda^1 \la \lambda^*$ and $S^1 \la S^*$
		\EndIf
	\EndWhile
	\State Return $S^0$
\end{algorithmic}
\end{algorithm}

In addition to the feasibility of LP \eqref{eq:feasibilityLP}, the main technical aspect of proving the correctness of the algorithm is to show that each agent gets an optimal allocation. 
This follows essentially from showing Condition \ref{enum:convexprices}, that the prices indeed form a piecewise linear convex curve, or equivalently, that the $\lambda^k$s form an increasing sequence.  
It is fairly straightforward to see that the running time of the algorithm is polynomial.

\section{Algorithm under Extensibility}\label{sec:general}
In this section we present the algorithm that proves Theorem \ref{thm:algorithm}; 
we will parallel the presentation in Section \ref{sec:cloud}.  
We first present equilibrium characterization for the general model in Theorem \ref{thm:lpl2eq} (complete proof is in Appendix \ref{sec:eqchar}), and then describe the key ideas in designing the algorithm (the missing proofs and other details of this part are in Appendix \ref{sec:fullproofs}). We show run of our algorithm on a network example in Section \ref{sec:examples}, Figure \ref{fig:sp}.

Recall that in our general model, each agent has a delay function and a set of constraints on the bundle of goods she gets. Unlike in Section \ref{sec:cloud}, there is no simple ordering among the goods that enables a geometric description of an equilibrium, 
therefore some parts that are immediate in that setting may require a proof here. 
Recall that the first step in Section \ref{sec:cloud} was to find an equilibrium characterization only in terms of prices. 
This used the geometry of the instance in order to partition the time slots into segments. 
More generally, it will be more convenient to consider a partition of agents rather than a partition of goods. 
By abuse of terminology, in this section, by ``segment'' we refer to a subset of agents. 
Each agent $i$ in $A$ has a parameter $\lambda_i$, that previously corresponded to the slope of the segment they were in.
Similarly, now too, all agents in a segment $S^k$ have the same $\lambda_i$. 
This will also correspond to the optimal dual variable for \ref{eq:budget} in the agent's optimization problem at the equilibrium. 

Given a vector of $\lambda_i$s, denoted by $\ll \in \Rplus^{|A|}$, next we define a parameterized linear program and its dual. Intuition for this definition comes from the optimal bundle LP of each agent at given prices. In the following, $LP(\ll)$ has allocation variables $x_{ij}$s, 
 the constraint \ref{eq:CC} for each agent $i$, and the \supplyrespecting constraint for each good $j$. 
 The corresponding dual variables are respectively $\alik$s and $p_j$s, where $p_j$ can be thought of as the price for good $j$. 
\begin{equation}\label{eq:lpl}
\begin{array}{cc}
LP(\ll): & DLP(\ll): \\
\begin{array}{ll}
\min: & \sum_i \li \sum_j d_{ij} x_{ij} \\
s.t. & \sum_j a_{ijk} x_{ij} \ge r_{ik},\ \forall (i,k)\\
& \sum_i x_{ij} \le 1,\ \forall j\\
& x_{ij} \ge 0,\ \forall (i,j)
\end{array}
\ \ \ \ \ 
& 
\ \ \ \ \ 
\begin{array}{ll}
\max: & \sum_{i,k}r_{ik} \alpha_{ik} - \sum_j p_j \\
s.t. & \lambda_i d_{ij} \ge \sum_k a_{ijk} \alpha_{ik} - p_j,\ \forall(i,j)\\
& p_j \ge 0,\ \forall j;\ \ \ \ \alpha_{ik} \ge 0,\ \forall (i,k)
\end{array}
\end{array}
\end{equation}

Remarkably, the next theorem shows that the problem of computing an equilibrium reduces to solving the above LP and its dual for a \emph{right} parameter vector $\ll \in \Rplus^{|A|}$.

\begin{restatable}{theorem}{lpleq}
	\label{thm:lpl2eq}
For a given $\ll>0$ if an optimal solution $\xx$ of $LP(\ll)$ and an optimal solution $(\aal, \pp)$ of $DLP(\ll)$ satisfy \ref{eq:budget} for all agents $i\in \CA$ with equality, then they constitute an equilibrium of market $\CM$. 
\end{restatable}
\rh{We note that the proof of Theorem \ref{thm:lpl2eq} uses only complementary slackness conditions of optimal bundle LP of each buyer, and therefore the theorem holds for the most general model, i.e., without any of the {\em extensibility}, {\em enough demand}, or {\em strong feasibility} assumptions.}
Theorem \ref{thm:lpl2eq} gives us the ``geometry'' of an equilibrium outcome, and is roughly equivalent to 
Condition \ref{enum:convexprices} from Section \ref{sec:cloud}. It reduces the problem to one of finding a {\em right} parameter vector $\ll$; however there is still the entire $\Rplus^{|A|}$ to search from. 
As was done in Section \ref{sec:cloud}, our main goal is to further reduce this task to a sequence of single parameter searches, each involving submodular minimization and binary search.

Theorem \ref{thm:lpl2eq} is applicable when all agents spend exactly their money at a primal-dual pair of optimal solutions for a given vector $\ll$. 
Now the question is to characterize such parameter vectors $\ll$. Note that, there really is no equivalent of Condition \ref{enum:slots} from Section \ref{sec:cloud}, since 
some goods may be allocated across agents in different segments. This is the source of many of the difficulties we face.  
Next, in Lemma \ref{lem:feas}, we derive an (approximate) equivalent of Conditions \ref{enum:equals} and \ref{enum:less} from Section \ref{sec:cloud}. 
This guarantees the existence of (allocation, prices)  that satisfy the budget constraints of agents. 
One difference here is that this is going to be a global condition that involves the entire vector $\ll$, rather than a local condition that we could apply to a single segment like in Section \ref{sec:cloud}. For this we need a number of properties of optimal solutions of $LP(\ll)$ and $DLP(\ll)$ that we show in Lemma \ref{lem:lp-prop} next. 

Let us define $\delay_i(\xx) = \sum_j d_{ij} x_{ij}$ and $\pay_i(\pp,\xx)=\sum_j p_j x_{ij}$. Similarly, for a subset of
agents $S\subseteq A$, $\delay_S(\xx) = \sum_{i\in S}\sum_j d_{ij} x_{ij}$ and $\pay_S(\pp,\xx)=\sum_{i\in S} \sum_j p_j x_{ij}$. 
Let $[d]$ denote the set $\{1,\dots,d\}$ of indices. By abuse of notation, let us define
\[
\lambda(S) = \twopartdef{\text{the $\lambda$ value of agents in $S$}} {\text{if all agents in $S$ have the same $\lambda_i$.}}
						{\text{undefined}}	{\text{otherwise}} 
.\]

Using {\em extensibility}, in the next lemma we show that optimal solutions of $LP(\lambda)$ and $DLP(\lambda)$ 
satisfy some invariants regarding delays and payments of agents, e.g., we will show that higher the $\lambda$ the better the delay at primal optimal. 
For a fixed dual optimal the total payment of a segment
remains fixed at all optimal allocations, and  as the delay of a subset decreases, its payment increases.
Recall Definition \ref{def:bestfor} for ``\bestfor''.

\begin{restatable}{lemma}{lpprop}\label{lem:lp-prop}
Given $\ll$, partition agents by equality of $\lambda_i$ into sets $S_1,\dots,S_d$ such that 
$\lambda(S_1)<\dots<\lambda(S_d)$. 
\begin{enumerate}
\item At any optimal solution $\xx$ of $LP(\ll)$ delay is minimized first for set $S_d$, then for $S_{(d-1)}$, and so on, finally for $S_1$. This is equivalent to $\xx$ being \bestfor each $T_g, \forall g \in [d]$ where $T_g = \cup_{q=g}^d S_q,$, and for any other optimal solution $\yy$ we have $\delay_{S_g}(\yy)=\delay_{S_g}(\xx),\ \forall g \in [d]$.  

\item Given two dual optimal solutions $(\aal,\pp)$ and $(\aal',\pp')$, if the first part of dual objective is same at both for some $g\in [d]$, {\em i.e.,} $\sum_{i \in S_g,k} r_{ik} \alik = \sum_{i \in S_g,k} r_{ik} \alpha'_{ik}$, then for any optimal solution $\xx$ of $LP(\ll)$, $\pay_{S_g}(\xx,\pp) = \pay_{S_g}(\xx,\pp')$.

\item Given two optimal solutions $\xx$ and $\xx'$ of $LP(\ll)$, and an optimal solution $(\aal, \pp)$ of $DLP(\ll)$, if for any subset $S\subseteq S_g$ for $g\in [d]$, $\delay_S(\xx)\le \delay_S(\xx')$, then $\pay_S(\xx,\pp) \ge \pay_S(\xx',\pp)$. The former is strict iff the latter is strict too.
\end{enumerate}
\end{restatable}

In the above lemma, the first claim follows from extensibility. The second and third claim follow from the first claim together with the fact that any pair of primal and dual optimal satisfies complementary slackness.

Recall Conditions \ref{enum:equals} and \ref{enum:less} of Section
\ref{sec:cloud} requiring respectively {\em budget balanceness}, and that when a
subset of agents in a segment are given the ``best'' allocation, their total
payment should be at least their total budget (or else they will over demand
some good). Using the first and last part of Lemma \ref{lem:lp-prop} the latter can be roughly translated to 
saying that when the rest of the agents are given the ``worst'' allocation, the rest underpay in total.
Based on this intuition next we define conditions {\em budget balance ($\BB$)} and {\em subset condition ($\SC$)} in the following.

\begin{definition}\label{def:feas}
Given $(\ll, \pp)$, and a set $S\subseteq A$, we say that condition
\begin{itemize}
\item $\BB$ is satisfied: If at any optimal solution $\xx$ of $LP(\ll)$, we have $\pay_{S}(\xx,\pp)=m(S)$. 
\item $\SC$ is satisfied: $\forall~ T \subseteq S$ let $\xx$ to be an optimal solution of $LP(\ll)$ where $\delay_T$ is maximized. Then, $m(T)
\ge \pay_T(\xx,\pp)$. 
\end{itemize}
\end{definition}

We will show that if $\BB$ and $\SC$ are satisfied for each ``segment'' at any given $\ll>0$ then $\ll$ is the {\em right} parameter vector. 
We will call such a $(\ll,\pp)$ \efeasible, formally defined next. 


\begin{definition}\label{def:feasible}
We say that pair $(\ll,\pp)$ is \efeasible\ if $\exists ~\aal$ such that $(\aal,\pp)$ is an optimal solution of $DLP(\ll)$, and pair $\ll,\pp$ satisfies
$\BB$ and $\SC$ for subsets $S_g,\ \forall g\le d$, where $S_1,\dots, S_d$ is the partition of $\CA$ by equality of $\ll_i$. 
\end{definition}

The next lemma shows that the parameter vector $\ll$ corresponding to a \efeasible pair would ensure existence of allocation where each agent spends exactly her budget, and thereby will give an equilibrium using Theorem \ref{thm:lpl2eq}. 

\begin{restatable}{lemma}{feas}\label{lem:feas}
If pair $(\ll^*,\pp^*)$ is \feasible for $\ll^*>0$ then there exists an optimal solution $\xx^*$ of the primal $LP(\ll^*)$ such that $\pay_i(\xx^*,\pp^*)=m_i,\ \forall i\in A$. 
\end{restatable}


Given such a $\ll^*$ and solution $(\aal^*, \pp^*)$ of $DLP(\ll^*)$ that satisfy conditions of Lemma \ref{lem:feas}, the above lemma ensures existence of allocation that satisfies  \ref{eq:budget}, $\forall i\in \CA$. We derive following feasibility LP in $\xx$ variables to compute such an allocation. 
\begin{equation}\label{eq:feasLP}
\begin{array}{ll}
\mbox{$\hxx$ is a solution of $LP(\ll^*)$}:& \sum_{i}\lambda^*_i \sum_j d_{ij} x_{ij} = \sum_{i} \lambda^*_i \sum_j d_{ij} \hx_{ij}\\
\forall (i,k):&  \sum_{j} a_{ijk} x_{ij} \ge r_{ik} \\
\forall j: & \sum_i x_{ij} \le 1 \\
\forall i: & \sum_j p^*_j x_{ij} = m_i \\
\forall (i,j): &  x_{ij}\ge 0
\end{array}
\end{equation}

Now our goal has reduced to finding a \efeasible $(\ll,\pp)$ pair. That is, if we think of partition of agents by equality of $\lambda_i$ as ``segments'',  then we wish to find a vector $\ll$ such that $\BB$ and $\SC$ are satisfied for each ``segment''. Our algorithm, defined in Algorithm \ref{alg}, tries to fulfill exactly this goal. At a high level, like in Section \ref{sec:cloud}, our algorithm will build the segments bottom up, {\em i.e.,} lowest to
highest $\lambda$ segments. We will start by setting all the $\lambda$s to same value, and find lowest $\lambda$ value where $\BB$ and $\SC$ are satisfied for a subset. Once found we freeze this subset as a segment and start increasing $\lambda$ for the rest to find the {\em next segment}, and repeat. 

\begin{algorithm}[t]
	\caption{Algorithm to compute market equilibrium under extensibility}\label{alg}
	
	\begin{algorithmic}[1]
		\State Input: $A,(m_i)_{i \in A},$ (\ref{eq:coveringLP})$_{i\in A}$
		\State Initialize $\CA'\la\CA$, $\pcur\la0, \llcur\la 0$ and $k\la 1$ 
		\While{$\CA'\neq \emptyset$}
		\State $(S^k, \llnew, \pnew) \la $ \nextseg$(\llcur, \pcur, \CA',(m_i)_{i \in A},$ (\ref{eq:coveringLP})$_{i\in A})$
		\State $\CA' \la\CA'\setminus S^k$, and $k\la k+1$
		\State $\llcur \la \llnew$, and $\pcur \la \pnew$.
		\EndWhile
		\State Compute allocations $\x_i$ satisfying \ref{eq:budget} for all $i \in A$, by solving LP \eqref{eq:feasLP} for $\pp^*=\pcur$ and $\ll^* =\llcur$.
		\State Output allocations $\xx$ and prices $\pcur$.
	\end{algorithmic}
\end{algorithm}

In this process of finding the {\em next segment} we need to make sure that $\BB$ and $\SC$ conditions are maintained for the previous segments. 
In Section \ref{sec:cloud} we were able to do this by simply {\em fixing the prices} of the goods in 
earlier segments, because goods were not shared across segments. Here, some of the
goods allocated to agents in the earlier segments may also be allocated to
agents in the later segments, and additionally these allocations are not fixed and may keep
changing during the algorithm.  (We fix the allocation only at the end.)
Furthermore, the prices are required to be dual optimal w.r.t. the $\ll$ vector that we eventually find. 
On the other hand in order to maintain $\BB$ and $\SC$ conditions for the previous segments 
we need to ensure that the total payment of previous segments do not change.

The next lemma shows that this is indeed possible by proving that prices of goods bought by agents from previous segments can be held fixed. 
In fact, we will be able to fix $\alpha_{ik}$s as well, for agents in the previous segments. \rh{The proof involves an application of Farkas' lemma,  leveraging extensibility. During computation of {\em next segment}, we will hold fix the $\lambda$s of agents in segments found so far, and increase the $\lambda$s of the remaining agents. To facilitate this we define $\ones_S\in \{0,1\}^A$ as the indicator vector of
$S\subseteq A$, i.e., $\ones_S(i) = 1$ if $i\in S$, and is $0$ otherwise. }

\begin{restatable}{lemma}{price}\label{lem:price}
Given a $\ll$, partition agents into $S_1,\dots, S_d$ by equality of $\lambda_i$, where $\lambda(S_1)<\dots<\lambda(S_d)$. For $R\subseteq
S_d$ consider primal optimal $\hxx$ that is \bestfor $R$, and let $(\haa,\hpp)$ be a dual optimal. 
Consider for some $a > 0$, the vector $\ll'= \ll + a \ones_R$. Then $\hxx$ is optimal in $LP(\ll')$ and there exists an optimal solution $(\aal',\pp')$ of $DLP(\ll')$ such that, 
\[
\begin{array}{ll}
\forall j: & p'_j \ge \hp_j \ \ \ \mbox{ and } \ \ \ \sum_{i \notin R} \hx_{ij} > 0 \Rightarrow p'_j=\hp_j\\
\forall i\notin R,\ \forall k,& \alpha'_{ik} = \ha_{ik}\ . 
\end{array}
\]
\end{restatable}

As discussed above our algorithm will build segments {\em inductively} from the lowest to highest $\lambda$ value, by increasing $\lambda$ of only the ``remaining'' agents. Suppose, we have built segments $S_1$ through $S_{k-1}$, and 
let $\CA'=\CA\setminus \cup_{g=1}^{k-1} S_g$ be the remaining agents.  
Let $\llcur$ be the current $\ll$ vector where $\llcur(S_1) <\dots<\llcur(S_{k-1})< \llcur(\CA')$.
Let $\pcur$ be the corresponding dual price vector which is optimal for $DLP(\llcur)$. 
For ease of notations we define the following. 
	\begin{equation}\label{eq:la}
	\mbox{For any $a\ge 0$, define }\ll^a = \llcur + a\ones_{\CA'}. 
	\end{equation}
	Fix an allocation $\xxcur$ of $LP(\llcur)$. 
	 We call an optimal solution $(\aal,\pp)$ of $DLP(\ll^a)$ {\em valid} if
	 prices are monotone w.r.t. $\pcur$ and $\xxcur$, in the sense as guaranteed by Lemma
	 \ref{lem:price} (where prices of goods allocated to previous segments are
	 held fixed and prices of the rest of the goods are not decreased), and
	 $\aal_i$s are fixed for agents outside $\CA'$. We will call the corresponding prices {\em valid prices}. 
  For simplicity we will assume {\em uniqueness of valid prices}.\footnote{This is without loss of generality since perturbing the parameters of the market ensures this. A typical way to simulate perturbation is by lexicographic ordering \cite{Sch-book}.}
\begin{equation}\label{eq:pa}
\mbox{Define }\pp^a \mbox{ to be the {\em valid} price vector at dual optimal of $DLP(\ll^a)$} 
\end{equation}

Since the correctness is proved by induction, the {\em inductive hypothesis} is that w.r.t. $(\llcur,\pcur)$, both $\SC$ and $\BB$ are satisfied for $S_1, \ldots S_{k-1}$, and $\SC$ is satisfied for the remaining agents $\CA'$. The base case is easy with the $\lambda_i$s all set to $0$. 
Our next goal is to find the next segment $S_k \subseteq \CA'$, a new vector $\llnew$ and a new price vector $\pnew$ such that 
the following properties hold. 

\begin{enumerate}
	\item Parameter vector $\llnew$ is obtained from $\llcur$ by fixing $\lambda_i$s of agents outside $\CA'$, increase
	$\lambda_i$s of agents in $S_k$ by the same amount, and those of agents $\CA'\setminus S_k$ by some more. The latter increase is to separate $S_k$ from $\CA'$. That is for some $a\ge 0$ and $\epsilon > 0, \llnew=\ll^a + \epsilon \ones_{\CA'\setminus S_k}$.
	\item Price vector $\pnew$ is valid and optimal for $DLP(\llnew)$. 
	\item W.r.t. $(\llnew,\pnew)$, $S_1,\dots, S_k$ satisfy both $\BB$ and $\SC$, and $ \CA'\setminus S_k$ satisfies $\SC$. 
\end{enumerate}


\begin{algorithm}[t]
\caption{Subroutine \nextseg$(\llcur, \pcur, \CA',(m_i)_{i \in A},$ (\ref{eq:coveringLP})$_{i\in A})$ }\label{alg.main.leastseg}
\begin{algorithmic}[1]
	\State Initialize $a^0\la 0, a^1 \la \Delta$, where $\Delta = (\sum_{i,j,k} |\aijk|+ \sum_{i,k} |r_{ik}|+ \sum_{i,j} |d_{ij}|+
	\sum_i |m_i|)^{2mn|C|}$. 
	\State Define 
	function $f_a$ as in (\ref{eq.fa}). 
	\State Set $S^0 \in \argmin_{S \subseteq A', S\neq \emptyset} f_{a^0}(S)$ and $S^1 \in \argmin_{S \subseteq
	A',S\neq \emptyset} f_{a^1}(S)$
	\While{$S{^0}\neq S{^1}$}
		\State Set $a^* \la\frac{a^0 + a^1}{2}$ and $S^* \in \argmin_{S\subseteq A', S\neq\emptyset}
		f_{a^*}(S)$  
		\If{$f_{a^*}(S^*)>0$}
			 Set $a^0 \la a^*$ and $S^0 \la S^*$
		\Else\ 
			Set $a^1 \la a^*$ and $S^1 \la S^*$
		\EndIf
	\EndWhile
	\State $S^*\la S^0$. Compute $a^*$ by solving feasibility LP for $S^*$ mentioned in Lemma \ref{lem:bin} such that $f_{a^*}(S^*)=0$ 
	\State // Next we compute maximal minimizer of function $f_{a^*}$ containing set $S^*$.
	\State $A'\la A'\setminus S^*$.
	\While{$A'\neq \emptyset$}
	\State $S\la \argmin_{S\subseteq A', T\neq \emptyset} f_{a^*}(S\cup S^*)$
	\State {\bf if} $f_{a^*}(S\cup S^*)>0$ {\bf then} {\bf break}
	\State {\bf else} set $S^* \la S^* \cup S$, $A'\la A'\setminus S$
	\EndWhile
	\State Set $\llnew \la \ll^{a^*}$, $\lambda_i^{new} \la \lambda_i^{new} +\epsilon \ones_{A'}$, and $\pnew\la $valid price at $\llnew$, where $\epsilon \la
	\frac{1}{\Delta}$
	\State Return $(S^*,\llnew,\pnew)$
\end{algorithmic}
\end{algorithm}

The computation of the next segment $S^k$ satisfying the above conditions is done by the subroutine \nextseg, which we describe next, and which is
formally defined in Algorithm \ref{alg.main.leastseg}. As in Section \ref{sec:cloud}, the basic idea is to reduce this problem to a single parameter binary search. Since Condition $\SC$ is satisfied for the remaining agents $\CA'$ at $(\llcur,\pcur)$ while $\BB$ is not, total payment of $\CA'$ is less than their total budget $m(\CA')$. In order to keep track of this surplus budget, consider the following function on $S\subseteq \CA'$. 
\begin{equation}\label{eq:f}
f_{\ll,\pp}(S) = m(S) - \pay_S(\xx,\pp), \mbox{ where $\xx$ is an optimal solution to $LP(\ll)$ that maximizes $\delay_S$.}
\end{equation}

We translate Condition 3 above in terms of this function, in the following lemma, which essentially reduces the problem to a single parameter search. 

\begin{restatable}{lemma}{segCond}\label{lem:segCond}
Suppose that for some $a\ge 0$,  
$$\textstyle S_k \in \argmin_{S\subseteq \CA',
		S\neq \emptyset} \{ f_{\ll^a,\pp^a}(S) \}.$$
	Further, suppose that $f_{\ll^a,\pp^a}(S_k) = 0$, and $S_k$ be a maximal such set. 
Then there exists a rational number $\epsilon >0$ of polynomial-size such that,
w.r.t. $(\llnew,\pnew)$ as defined above, $S_1,\ldots, S_{k}$ satisfy both $\BB$
and $\SC$, and $\CA'\setminus S_k$ satisfy $\SC$. 
\end{restatable}

The above lemma reduces the task of finding next segment to finding an appropriate $a$ such that the minimum value of
$f_{\ll^a,\pp^a}$ is zero under the {\em valid} price $\pp^a$. 
This requires two things: first we need to find a minimizer of $f_{\ll^a,\pp^a}$ for a given $a>0$, and second we need to find the right value of $a$. The next lemma shows that the first can be done using an algorithm for submodular minimization, and therefore in a polynomial time \cite{schrijver.book}. 
For convenience of notation, we define the following functions. 
\begin{equation}\label{eq.fa}
f_a(S):= f_{\ll^a,\pp^a}(S)\ \ \ \mbox{ and } \ \ g(a):=\min_{S\subseteq \CA', S\neq \emptyset} f_a(S)
\end{equation}

\begin{restatable}{lemma}{submodular}\label{lem:submodular}
Given $a\ge 0$, function $f_a$ is submodular over set $\CA'$.
\end{restatable}

Now the question remains, how does one find an $a$ such that the minimum value is $0$, {\em i.e.,} $g(a)=0$. 
%
We will do binary search for the same. 
In the next lemma we derive a number of properties of $g$ that facilitates binary search, under {\em sufficient demand} assumption
(Definition \ref{def:ED}), while crucially using Lemmas \ref{lem:lp-prop} and \ref{lem:price}. 

\begin{restatable}{lemma}{bin}\label{lem:bin}
Function $g$ satisfies the following: $(i)$ $g(0) \geq 0$.
$(ii)$ $f_a(S)$ is continuous and monotonically decreasing in $a$, $\forall S\subseteq \CA'$, therefore $g$ is continuous and monotonically decreasing.
$(iii)$ $\exists a_h\ge 0$ a rational number of polynomial size such that $g(a_h)\le 0$, and $g$ has a zero of polynomial-size.
$(iv)$ Given a set $S\subseteq \CA'$, if $f_a(S)>0$ and $f_{a'}(S)<0$ for $a'>a>0$, then $\exists a^*\ge 0$ such
that $f_{a^*}(S)=0$ and such an $a^*$ can be computed by solving a feasibility linear
program of polynomial-size. 
\end{restatable}

The first part follows essentially from the fact that  $\CA'$ satisfies $\SC$ condition w.r.t. $(\llcur,\pcur)$ equivalently $(\ll^0,\pp^0)$. 
For the second part, we show that for any $S\subseteq \CA'$, function $f_a(S)$ is monotonically decreasing and continuous in $a$. 
Since $\min$ of continuous and decreasing functions is also continuous and decreasing, we get the same property for $g$. 
It turns out that whatever be the current value of $g(a)$, there is another $a'>a$ where $g(a')$ is strictly smaller (using {\em sufficient demand} assumption). This strict decrease property ensures existence of $a$ where $g(a)$ is zero. Setting $a_h$ to higher than any such $a$ would give
$g(a_h)\le 0$ since $f_a(S)$ for all $S$ are monotonically decreasing in $a$,
thereby we get the third part. Finally for the fourth part, existence of $a^*$ follows from monotonicity and continuity of $f_a(S)$ in $a$, and 
using the fact that complementary slackness ensures optimality we construct a
feasibility linear program to compute $a^*$, given $S\subseteq \CA'$, such that
$f_{a^*}(S)=0$. 



We initialize our binary search with a lower pivot $a_0=0$ and a higher pivot $a_1$ to a value that is guaranteed to be higher than
where some set goes tight (third part of Lemma \ref{lem:bin}).
Finally, since submodular minimization, binary search over a polynomial-sized range, and solving linear program all can be done in 
polynomial-time, we get our main result, Theorem \ref{thm:algorithm}, using Lemmas \ref{lem:feas}, \ref{lem:segCond} and \ref{lem:bin}, and Theorem
\ref{thm:lpl2eq}.

\newcommand{\p}{\mbox{\boldmath $p$}}

\section{Examples}\label{sec:examples}
In this section we show several interesting examples that illustrate important properties of our market and of equilibria.  
In addition we demonstrate a run of our algorithm on a routing example. 

\paragraph{Non-existence of equilibria.} 
Consider the networks (typically used to show Braess' paradox) in Figure \ref{fig:non-exist}, where the label on each edge specifies its (capacity, delay cost). There are two agents, each with a requirement of 1 from $s$ to $t$. Their $m_i$s are $100$ and $1$ respectively.
The network on the left has enough capacity to route two units of flow, but does not satisfy {\em strong feasibility} condition and does not have an equilibrium. This demonstrates importance of {\em strong feasibility} condition, without which even a simple market may not have an equilibrium. 

The network in the middle does satisfy {\em strong feasibility} but not {\em extensibility} and has an equilibrium. The network on the right satisfies {\em extensibility} but not {\em strong feasibility} and has an equilibrium. This demonstrates that conditions of {\em strong feasibility} and {\em extensibility} are incomparable.

\begin{figure}[!h]
\vskip -2cm
\includegraphics[width=\textwidth]{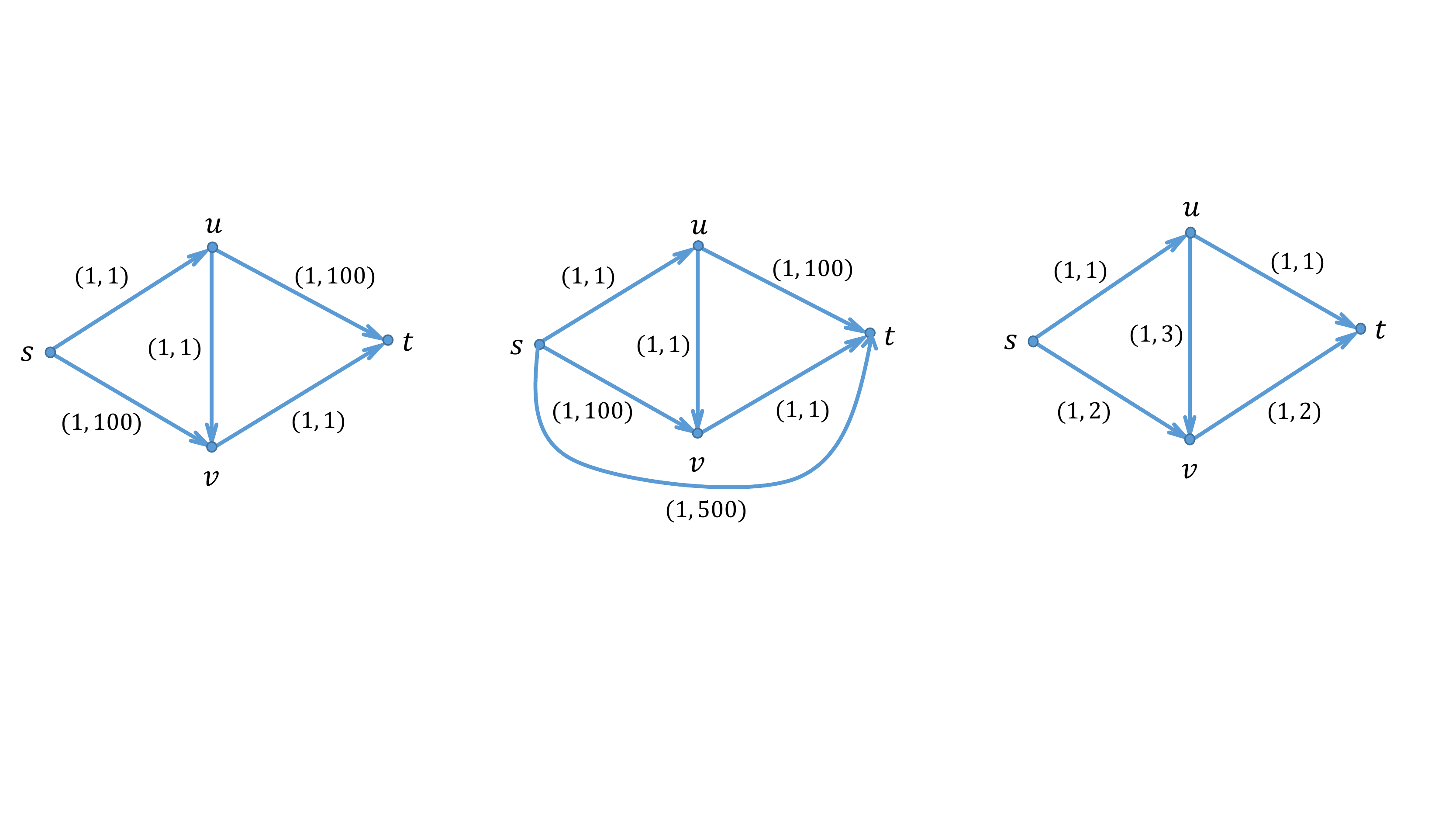}
\vskip -3cm
\caption{Non-existence of equilibria: The network on the left does not satisfy {\em strong feasibility} and has no equilibria. Edge labels specify (capacity, delay\ cost), there are two agents, each with a requirement of 1 from $s$ to $t$. The $m_i$s are $100, 1$. The same network with an additional edge with huge cost in the middle, does satisfy {\em strong feasibility} but not {\em extensibility} and has an equilibrium. The network on the right with different edge labels satisfies the {\em extensibility} but not {\em strong feasibility}. It has an equilibrium and our algorithm will find one.}\label{fig:non-exist} 
\end{figure}

\paragraph{Non-convexity of equilibria.}
Consider a market in the scheduling setting of Section \ref{sec:cloud}, with 6 agents, each with a requirement of 1. 
Their $m_i$s are $30, 17, 9, 4, 3, 1$. Table \ref{fig:hole1} depicts some of the equilibrium prices for this instance. 
A bigger  example with 9 agents is in  Table \ref{fig:hole5}. 
These examples show that the equilibrium set is not convex, but forms a connected set. 
Since these are in high dimension, it is not easy to determine the exact shape of the entire equilibrium set, 
but one can see that it is quite complicated. 

\begin{table}[!h]
  \caption{An example in the scheduling setting of Section \ref{sec:cloud}  where the set of equilibrium prices is non-convex. There are 6 agents, each with a requirement of 1.
  	Their $m_i$s are $30, 17, 9, 4, 3$ and $1$. We depict only a subset of all equilibria here. In particular, we depict 6 equilibrium prices, $\p_1, \p_2, \cdots, \p_6$. 
  	 All prices either along solid lines connecting any two of these points, or in the shaded region are equilibria. 
  	 However,  if any two of these prices are not connected by a solid line, then none of the points on the line joining them is an equilibrium. For example, none of the prices on the line joining $\p_1$ and $\p_6$,  $\p_2$ and $\p_5$, or $\p_3$ and $\p_4$ is an equilibrium. There are more equilibrium points not depicted here. As far as we can tell, the shape of the equilibrium set is something akin to a cup, with  empty space inside, but forming a single connected region. }
\vskip -2cm
\begin{minipage}{\textwidth}
\hspace{-1cm}
  \begin{minipage}{0.6\textwidth}
	\includegraphics[width=1.6\textwidth]{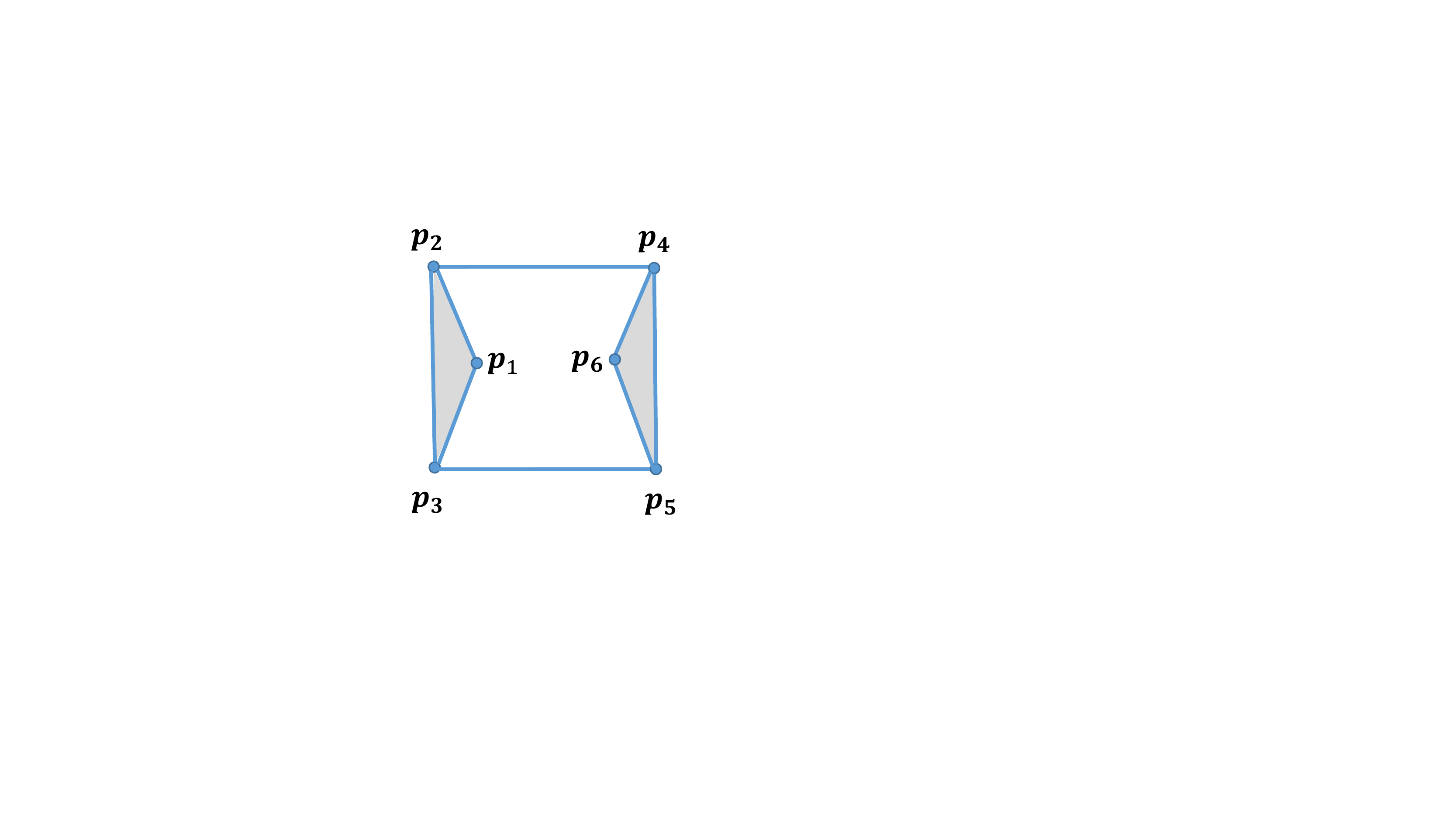}
  \end{minipage}
  \begin{minipage}{0.4\textwidth}
	\begin{tabular}{|cl|}
	\hline
	$\p_1:$&     \hspace{-0.3cm} $(30,$  $17,$  $9,$  $5,$  $2,$  $1)$ \M\\
	$\p_2:$&     \hspace{-0.3cm} $(30,$  $17,$  $9,$  $4\nfrac13,$  $2\nfrac23,$  $1)$ \M\\
	$\p_3:$&     \hspace{-0.3cm} $(34,$  $13,$  $9,$  $5,$  $2,$  $1)$ \M\\
	$\p_4:$&     \hspace{-0.3cm} $(34,$  $13,$  $9,$  $5\nfrac13,$  $2\nfrac23,$  $0)$ \M\\
	$\p_5:$&     \hspace{-0.3cm} $(35,$  $13,$  $8,$  $4\nfrac13,$  $2\nfrac23,$  $1)$ \M\\
	$\p_6:$&     \hspace{-0.3cm} $(35,$  $13,$  $8,$  $5\nfrac13,$  $2\nfrac23,$  $0)$ \M\\
	\hline
	\end{tabular}
    \end{minipage}
  \end{minipage}
 \vskip -2cm 
  	\label{fig:hole1}
 \end{table}
 \begin{table}[!h]
 	\caption{An example in the scheduling setting of Section \ref{sec:cloud}  where the set of equilibrium prices is non-convex. 
 		There are 9 agents, each with a requirement of 1.
 		Their $m_i$s are $56, 45, 33, 23, 17, 10, 4, 3$ and $1$. 
 		We depict only a subset of all equilibria here. In particular, we depict 20 equilibrium prices, $\p_1, \p_2, \cdots, \p_{20}$. 
 		All prices either along solid lines connecting any two of these points, or in the shaded region are equilibria. 
 		However,  if any two of these prices are not connected by a solid line, then none of the points on the line joining them is an equilibrium.}
	\vspace{-1.5cm}
 	\begin{minipage}{\textwidth}
 		\begin{minipage}{0.6\textwidth}
 			\hspace{-3.5cm}\includegraphics[width=1.6\textwidth]{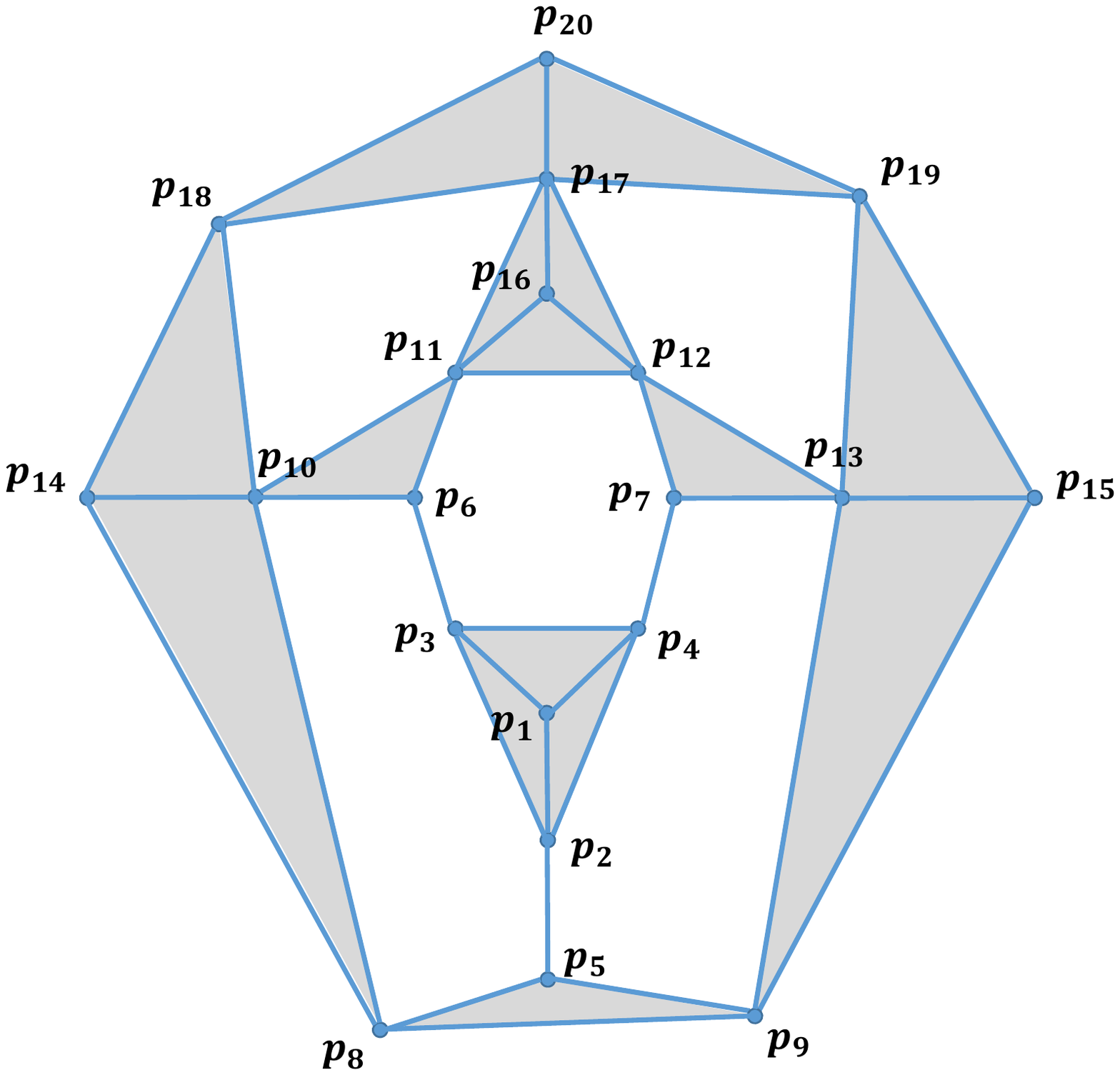}
 		\end{minipage}
 		\hspace{-1cm}
 		\begin{minipage}{0.4\textwidth}
 			{\footnotesize
 				\begin{tabular}{|cl|}
 					\hline
 					$\p_1:$ & \hspace{-0.3cm} $(57,$ $44,$  $33,$  $24,$  $16,$  $10,$  $5,$  $2,$  $1)$ \M\\
 					$\p_2:$&     \hspace{-0.3cm} $(57,$ $44\nfrac23,$  $32\nfrac13,$  $24,$  $16,$  $10,$  $5,$  $2,$  $1)$ \M\\
 					$\p_3:$&     \hspace{-0.3cm} $(57,$ $44,$  $33,$  $24,$  $16\nfrac23,$  $9\nfrac13,$  $5,$  $2,$  $1)$ \M\\
 					$\p_4:$&     \hspace{-0.3cm} $(57,$ $44,$  $33,$  $24,$  $16,$  $10,$  $5,$  $2\nfrac23,$  $\nfrac13)$ \M\\
 					$\p_5:$&     \hspace{-0.3cm} $(57\nfrac13,$  $44\nfrac23,$  $32,$  $24,$  $16,$  $10,$  $5,$  $2,$  $1)$ \M\\
 					$\p_6:$&     \hspace{-0.3cm} $(57,$  $44,$  $33,$  $24\nfrac13,$  $16\nfrac23,$  $9,$  $5,$  $2,$  $1)$ \M\\
 					$\p_7:$&     \hspace{-0.3cm} $(57,$  $44,$  $33,$  $24,$  $16,$  $10,$  $5\nfrac13,$  $2\nfrac23,$  $0)$ \M\\
 					$\p_8:$&   \hspace{-0.3cm} $(57\nfrac13,$  $44\nfrac23,$  $32,$  $24,$  $16\nfrac23,$  $9\nfrac13,$  $5,$  $2,$  $1)$ \M\\
 					$\p_9:$&   \hspace{-0.3cm} $(57\nfrac13,$  $44\nfrac23,$  $32,$  $24,$  $16,$  $10,$  $5,$  $2\nfrac23,$  $\nfrac13)$ \M\\
 					$\p_{10}:$&   \hspace{-0.3cm} $(57,$  $44\nfrac23,$  $32\nfrac13,$  $24\nfrac13,$  $16\nfrac23,$  $9,$  $5,$  $2,$  $1)$ \M\\
 					$\p_{11}:$&   \hspace{-0.3cm} $(57,$  $44,$  $33,$  $24\nfrac13,$  $16\nfrac23,$  $9,$  $5,$  $2\nfrac23,$  $\nfrac13)$ \M\\
 					$\p_{12}:$&   \hspace{-0.3cm} $(57,$  $44,$  $33,$  $24,$  $16\nfrac23,$  $9\nfrac13,$  $5\nfrac13,$  $2\nfrac23,$  $0)$ \M\\
 					$\p_{13}:$&   \hspace{-0.3cm} $(57,$  $44\nfrac23,$  $32\nfrac13,$  $24,$  $16,$  $10,$  $5\nfrac13,$  $2\nfrac23,$  $0)$ \M\\
 					$\p_{14}:$&  \hspace{-0.3cm} $(57\nfrac13,$  $44\nfrac23,$  $32,$  $24\nfrac13,$  $16\nfrac23,$  $9,$  $5,$  $2,$  $1)$ \M\\
 					$\p_{15}:$&  \hspace{-0.3cm} $(57\nfrac13,$  $44\nfrac23,$  $32,$  $24,$  $16,$  $10,$  $5\nfrac13,$  $2\nfrac23,$  $0)$ \M\\
 					$\p_{16}:$&  \hspace{-0.3cm} $(57,$  $44,$  $33,$  $24\nfrac13,$  $16\nfrac23,$  $9,$  $5\nfrac13,$  $2\nfrac23,$  $0)$ \M\\
 					$\p_{17}:$&\hspace{-0.3cm}  $(57,$  $44\nfrac23,$  $32\nfrac13,$  $24\nfrac13,$  $16\nfrac23,$  $9,$  $5\nfrac13,$  $2\nfrac23,$  $0)$ \M\\
 					$\p_{18}:$&\hspace{-0.3cm}  $(57\nfrac13,$  $44\nfrac23,$  $32,$  $24\nfrac13,$  $16\nfrac23,$  $9,$  $5,$  $2\nfrac23,$  $\nfrac13)$ \M\\
 					$\p_{19}:$&\hspace{-0.3cm}  $(57\nfrac13,$  $44\nfrac23,$  $32,$  $24,$  $16\nfrac23,$  $9\nfrac13,$  $5\nfrac13,$  $2\nfrac23,$  $0)$ \M\\
 					$\p_{20}:$ &\hspace{-0.3cm} $(57\nfrac13,$  $44\nfrac23,$  $32,$  $24\nfrac13,$  $16\nfrac23,$  $9,$  $5\nfrac13,$  $2\nfrac23,$  $0)$ \M\\
 					\hline
 				\end{tabular}
 			}
 		\end{minipage}
 	\end{minipage}
 	\label{fig:hole5}
 \end{table}

\paragraph{A run of the algorithm.} 
We describe the run of our algorithms on simple examples here. 
The run of Algorithm \ref{alg.scheduling} on the example in Table \ref{fig:hole1} is as follows. 
Each row below depicts one iteration, where we find a new segment. 
We first give the set of agents in this new segment, then the corresponding $\lambda$, and then 
the prices of the slots determined in this iteration. 
The last column shows the sets which give the second and third lowest $\lambda$s in that iteration, and hence were not selected. 

\begin{center}
\begin{tabular}{llll}
$S^1 = \{6\}$, & $\lambda_{S^1} = 1$, & $p_6 = 1$ & ($\lambda_{\{5,6\}} = 1\nfrac13, \lambda_{\{4, 5, 6\}} = 1\nfrac13, \dots$)\\
$S^2 = \{4, 5\}$, & $\lambda_{S^2} = 1\nfrac23$, & $p_5 = 2\nfrac23,\ p_4 = 4\nfrac13$ & ($\lambda_{\{5\}} = 2, \lambda_{\{3, 4, 5\}} = 2\nfrac16, \dots$)\\
$S^3 = \{3\}$, & $\lambda_{S^3} = 4\nfrac23$, & $p_3 = 9$ & ($\lambda_{\{2,3\}} = 5\nfrac79, \lambda_{\{1, 2, 3\}} = 7\nfrac16$)\\
$S^4 = \{2\}$, & $\lambda_{S^4} = 8$, & $p_2 = 17$ & ($\lambda_{\{1,2\}} = 9\nfrac23$)\\
$S^5 = \{1\}$, & $\lambda_{S^5} = 13$, & $p_1 = 30$ & \\
\end{tabular}
\end{center}

The equilibrium price found in this run is the point $\p_2$ in Table \ref{fig:hole1}. 
This price curve is shown in Figure \ref{fig:price-plot}. 
The allocation obtained by solving the feasibility LP \eqref{eq:feasLP} is as follows: 
$x_{11} = 1, x_{22} = 1, x_{33} = 1, x_{44} = 4/5, x_{45} = 1/5, x_{54} = 1/5, x_{55} = 4/5, x_{66} = 1$. 

\begin{figure}[!h]
	\vskip -4.5cm
	\includegraphics[width=0.7\textwidth]{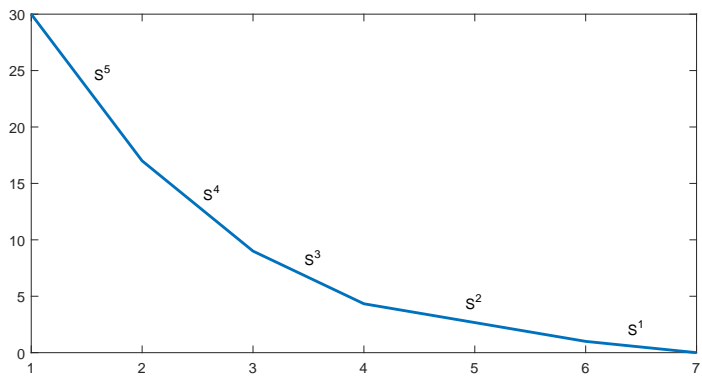}
	\vskip -4.5cm
	\caption{Piecewise linear convex decreasing curve of equilibrium prices obtained by the algorithm for the example in Table \ref{fig:hole1}.}\label{fig:price-plot}
\end{figure}


%

We next describe the run of the algorithm on a network flow example, described  in Figure \ref{fig:sp}. 
 The figure shows the network structure and the edge labels specify (capacity, delay\ cost). 
 There are five agents with requirements $10, 11, 12, 13, 14$ from $s$ to $t$ respectively. 
 Their $m_i$s are $12, 10, 4, 2, 2$. 
 This network is not series-parallel, yet it satisfies the \emph{extensibility} condition, so our algorithm finds an equilibrium. 

\begin{figure}[!h]
\vskip -3cm
\includegraphics[width=\textwidth]{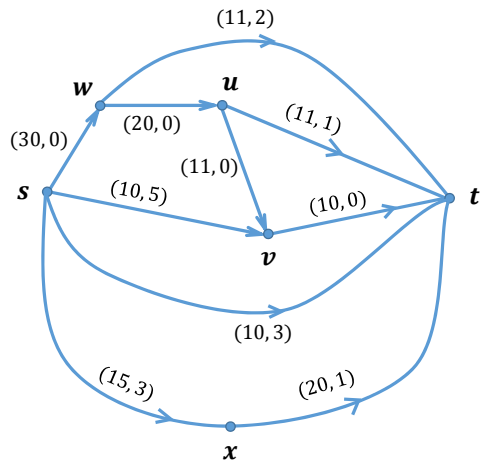}
\vskip -3cm
\caption{A network flow example where there are 5 agents with requirements $10, 11, 12, 13, 14$ from $s$ to $t$ respectively. Their $m_i$'s are $12, 10, 4, 2, 2$. The network satisfies the \emph{extensibility} condition, so our algorithm finds an equilibrium.}
\label{fig:sp}
\end{figure}

The run of Algorithm \ref{alg} on this example (in Figure \ref{fig:sp}) is as follows. 
Once again, each row below depicts one iteration, where we find a new segment. 
We first give the set of agents in the new segment, then the corresponding $\lambda$, and then 
the prices of the edges that are fixed in this iteration. 
The last column shows the second and third lowest $\lambda$s in that iteration. 

\begin{center}
\begin{tabular}{llll}
$S^1 = \{3,4,5\}$, & $\lambda_{S^1} = \nfrac{8}{37}$, & $p_{sx} = p_{xt} = p_{wt} = 0, p_{st} = \nfrac{8}{37}, p_{sw} = \nfrac{16}{37}$ & ($\lambda_{\{4,5\}} = \nfrac{4}{13},$ \\ & & & $\lambda_{\{2, 3, 4, 5\}} = \nfrac{9}{35}$)\\
$S^2 = \{1, 2\}$, & $\lambda_{S^2} = \nfrac{478}{1147}$, & $p_{wu} = \nfrac{478}{1147},\ p_{vt} = \nfrac{478}{1147}, p_{sv} = p_{uv} = p_{ut} = 0$ & ($\lambda_{\{2\}} = \nfrac{194}{407}$)\\
\end{tabular}
\end{center}

The allocation from the feasibility LP \eqref{eq:feasLP}:
\begin{itemize}
	\item Agent $1$ sends $\nfrac{2012}{239}$ units of flow on path $s-w-u-v-t$ and $\nfrac{378}{239}$ units of flow on path $s-w-u-t$. 
	\item Agent $2$ sends $\nfrac{2251}{239}$ units of flow on path $s-w-u-t$ and $\nfrac{378}{239}$ units of flow on path $s-w-u-v-t$. 
	\item Agent $3$ sends $8$ units of flow on path $s-w-t$, $\nfrac52$ units of flow on path $s-t$ and $\nfrac32$ units of flow on path $s-x-t$. 
	\item Agent $4$ sends $2$ units of flow on path $s-w-t$, $\nfrac{21}{4}$ units of flow on path $s-t$ and $\nfrac{23}{4}$ units of flow on path $s-x-t$. 
	\item Agent $5$ sends $2$ units of flow on path $s-w-t$, $\nfrac{21}{4}$ units of flow on path $s-t$ and $\nfrac{27}{4}$ units of flow on path $s-x-t$. 
\end{itemize}


%
%
%
%
%

\bibliographystyle{plainnat}
\bibliography{kelly}

\appendix

\section{Related work on computation and applications of market equilibrium}\label{sec:marketapps}
\paragraph{Computation and Complexity:}
The computational complexity of market equilibrium has been extensively studied for the tradition models in the past decade and a half. 
This investigation has involved many algorithmic techniques, such as  primal-dual and flow based methods \citep{DPSV,DV,orlin2010improved,vegh2012concave,vegh2012strongly}, auction algorithms \citep{GK}, 
ellipsoid \citep{JainAD} and other convex programming based techniques \cite{CMV}, cell-decomposition \citep{DPS,DK08,VY},
distributed price update rules \cite{cole2008fast,wu2007proportional,cheung2013tatonnement}, 
and complementary pivoting algorithms \citep{GMSV,GMV}, to name some of the most prominent. 
The algorithms have been complemented by hardness results, either for  PPAD \cite{CSVY,Chen.plc} or for FIXP \citep{EY07, GargFIXP}, pretty much closing the gap between the two. 
Most of these papers focus on  traditional utility functions used in the economics literature. 
A notable exception that considers \emph{combinatorial} utility functions is 
\cite{jain2010eisenberg}, that study 
a market where agents want to send flow in a network, motivated by rate control algorithms governing the traffic in the Internet. 


Beyond being an important component in the complexity theory of total functions \cite{megiddo1991total}, 
the computation of market equilibria has been studied by economists for much longer \cite{BSAD,scarf.book,Smale}.  
The classic case for the use of equilibrium \emph{computation} is counter-factual evaluation of policy or design changes \cite{Shoven}, 
based on the assumption that markets left to themselves operate at an equilibrium. 

\paragraph{Fair allocation:}
Recently, market equilibrium outcomes have been used for fair allocation. 
Market equilibrium conditions are often considered inherently fair, 
therefore equilibrium outcomes have been used to allocate resources by a central planner seeking a fair allocation 
even when there is no actual market or monetary transfers. 
E.g., the proportional fair allocation, which is well known to be equivalent to the equilibrium allocation in a Fisher market \citep{KV}, is widely used in the design of computer networks. 
Exchange of bandwidth in a bittorrent network is modeled as a process that converges to a market equilibrium by \citet{wu2007proportional}. 
\citet{BudishJPE11} proposes ``competitive outcome from equal incomes" (CEEI) as a way to allocate courses to students: 
the allocation is an equilibrium in a market for courses in which the students participate with equal budgets (with random perturbations to break ties). {This scheme has been successfully used at the Wharton business school \cite{Budish2014changing}. }
\citet{ColeGG13} show that a suitable modification of the Fisher market equilibrium allocation can be used as a solution to a problem of fair resource allocation, without money. The mechanism is truthful, and satisfies an approximate per-agent welfare guarantee. Truthful mechanisms have also been designed for scheduling, where it is the auctioneer who has jobs to be scheduled and the agents are the one providing the required resources e.g., see \cite{NisanRonen,LaviSwamy}. This is in contrast to our setting where the agents have scheduling requirements. 

\paragraph{Market based mechanisms:} 
 There is also a long history of ``market based mechanisms'', where a mechanism (with monetary transfers) implements an equilibrium outcome. 
 The New York Stock Exchange uses such a mechanism to determine the opening prices, and copper and gold prices in London are fixed using a similar procedure \cite{Rustichini1994convergence}.  
There are different ways to do this:  use a sample (either historic or random)  or a probabilistic model of the population to compute the equilibrium price, and offer these prices to new agents. This is preferable to asking the bidders to report their preferences,
 computing the equilibrium on reported preferences and offering the equilibrium prices back. 
 The latter leads to obvious strategic issues; 
 \citet{Hurwicz1972}
 shows that strategic behavior by agents participating in such a mechanism can lead to inefficiencies.
 \citet{BabaioffLNL14} show price of anarchy bounds on such mechanisms. In any case, such mechanisms are ``incentive compatible in the large'', meaning that as the market size grows and each agent becomes insignificant enough to affect prices on his own, 
 his best strategy is to accept the equilibrium outcome. 
Nonetheless such mechanisms have been proposed and used in practice, e.g., for selling TV ads \cite{nisan2009google}. 

\paragraph{Budget constraints:}
Budget constraints in auctions has gained popularity in the last decade due to ad auctions 
\cite{MSVV.design,dobzinski2012multi,bhattacharya2010incentive,fiat2011single}, but  has been studied for quite some time \cite{laffont1996optimal,che2000optimal}. 
There has also been a recent line of work considering budget constraints in a procurement setting \cite{singer2010budget,badanidiyuru2012learning}.

\section{Relation to Myerson's ironing}
\label{app:myerson}

\begin{figure}
	\label{fig:ironing}
	\centering
	\includegraphics[width=\linewidth]{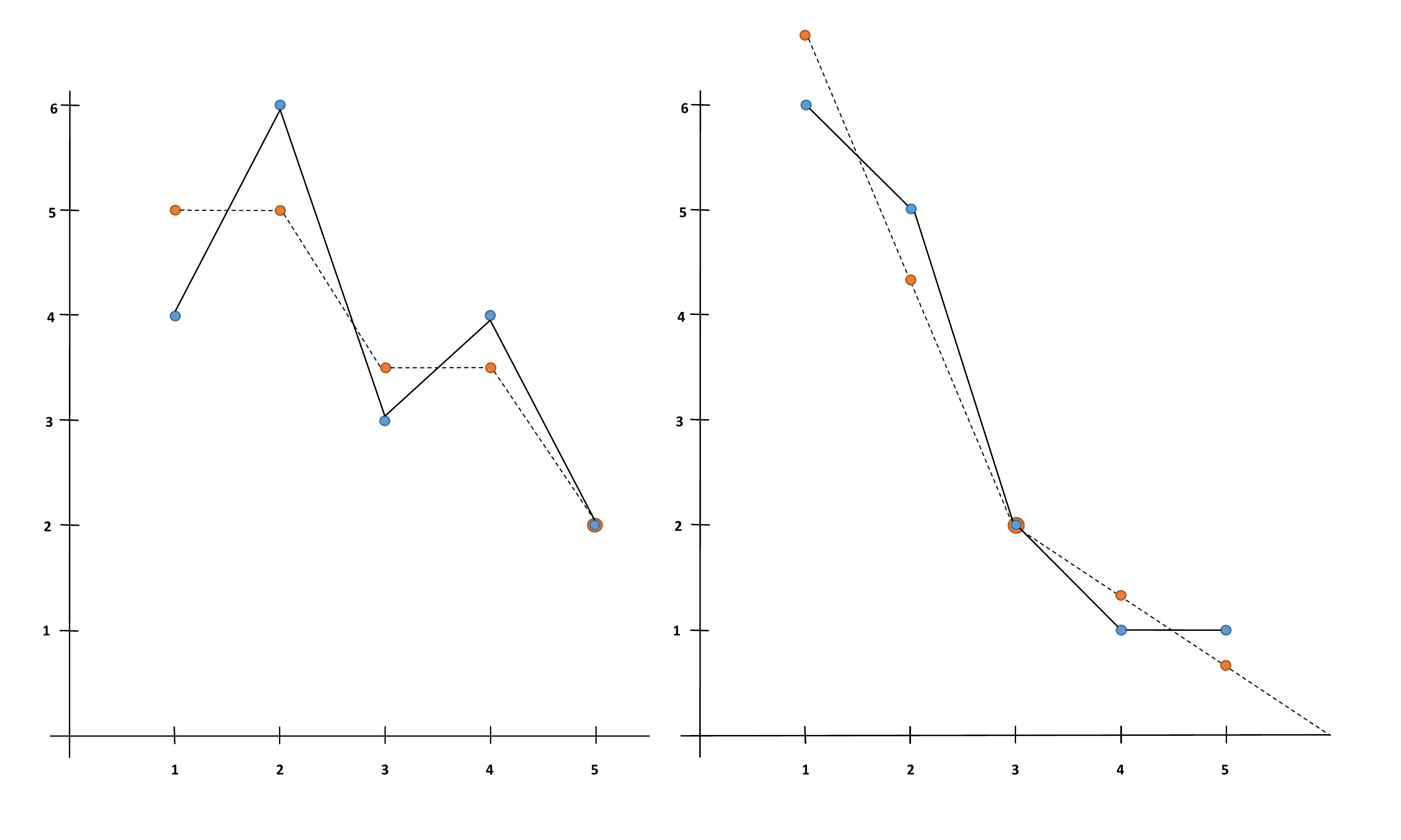}
	\caption{On the left is an example of Myerson's ironing. The solid curve with 
		the blue dots is the given curve, which is non-monotone. The dashed curve with the orange dots is the ironed curve, which is monotone. On the right is an example of our problem. The solid curve is the money function, which is monotone but not convex. The dashed one is the price function, which is convex. Both dashed curves are such that their ``area under the curve" is higher than that for the solid curves, and satisfy a minimality condition among all such curves.}
		\label{fig:MyersonandSchedulingMEironing}
		\end{figure}

		For a special case of the scheduling setting, when $r_i =1, \forall~i\in A$, we show that equilibrium conditions are equivalent to a set of conditions that are reminiscent of the \emph{ironing} procedure used in the characterization of optimal auctions by \citet{Myerson}.
		It is in fact ``one higher derivative" analog of Myerson's ironing. 
		Let's first restate Myerson's ironing procedure for the case of a uniform distribution over a discrete support.
		Suppose that we plot on the $x$-axis the quantiles, 
		in the decreasing order of value, and on the $y$-axis the corresponding virtual values. 
		This is possibly a non-monotone function, and Myerson's ironing asks for an ironed function that is \emph{monotone non-increasing},
		and is such that the area under the curve (starting at 0) of the ironed function is always higher 
		than that for the given function. Further, the ironed function given by this procedure is the {\em minimal} among all such functions. This means that wherever 
		the area under the curve differs for the two functions, the ironed function is \emph{constant}. (See \prettyref{fig:MyersonandSchedulingMEironing}.)
		
		In the special case of our model with a single good and when requirements are all one, the equilibrium price of the good as a function of time is obtained as an ironed analog of the {\em money function}: the function $i\mapsto m_i$, where we assume the $m_i$s are sorted in the decreasing order. 
		This money function is monotone non-increasing by definition but it need not be a convex function. 
		The price as a function of time must be a {\em monotone non-increasing  and convex} function. 
		The area under the curve of the price function must always be higher than that of the money function;
		further, wherever the two areas are different, the price function must be \emph{linear}. One can see that the conditions are the same as that of Myerson's ironing, except each condition is replaced by a higher derivative analog. 
		Unlike Myerson's, the solution to our problem is no longer unique and the solution set may be non-convex!

\section{Existence of Equilibrium under Strong Feasibility}\label{sec:existenceproof}
In this section we show existence of equilibrium for market instances satisfying {\em strong feasibility} (Definition
\ref{def:strongfeasibility} in Section \ref{sec:model}).
Given such an instance $\CM$ with set $A$ of $n$ agents and set $G$ of $m$ goods, let us create another
instance $\CM'$ by adding an extra good $s$ with ``large quantity'' and very high delay cost. 
Recall that the number of goods and agents in market $\CM$ be $m$ and $n$ respectively. 

Set of agents and goods in market $\CM'$ are respectively $A$ and $G'=G\cup \{s\}$. 
After normalizing to get quantity $1$ for good $s$, we set coefficient of variable $x_{is}$ in all the constraints of \ref{eq:CC} to $a_s=(n+1)r_{max}$ for each agent, and set the delay cost for good $s$ and for every agent to $d_s=m^{(m+1)} d_{max} a_s \left(\frac{a_{max}}{a_{min}}\right)^m$,
where $a_{max}=\max_{i,j,k} |a_{ijk}|, a_{min}=\min\{ 1, \min_{i,j,k} |a_{ijk}|\},$ and $r_{max}=\max_{i,k} r_{ik}$.  
Thus, given prices $(p_1,\dots,p_m, p_s)$ of goods in $\CM'$, the optimal bundle of agent $i$ at these prices can be found by solving the following linear program. 

\begin{equation}\label{eq:opti}
\begin{array}{rl}
OPT_i(\pp) = \argmin: & \sum_j d_{ij} x_{ij} + d_s x_{is}\\
s.t. & \sum_j a_{ijk} x_{ij} +a_s x_{is}\ge r_{ik},\ \forall k \in C\\
& \sum_j p_j x_{ij}+p_s x_{is} \le m_i \\
& x_{ij}\ge 0,\ \forall j\in G'
\end{array}
\end{equation}

The next lemma follows from the construction of market $\CM'$.

\begin{lemma}\label{lem:sf-cmp}
If $\CM$ satisfies strong feasibility then so does $\CM'$.
\end{lemma}

Price vector $\pp$ is said to be at equilibrium, if when every agent is given its optimal bundle, there is no excess demand of any
good, and goods with excess supply have price zero. That is, $(\xx,\pp)$ such that, 

\begin{equation}\label{eq.eq}
\forall i\in A,\ \x_i \in OPT_i(\pp),\ \ \ \mbox{ and } \ \ \ \forall j\in G',\  \sum_{i\in A} x_{ij}\le 1;\ \ \ p_j >0 \Rightarrow \sum_{i\in A}
x_{ij}=1 
\end{equation}

\begin{lemma}\label{lem:pszero}
If $\x_i \in OPT_i(\pp),\ \forall i \in A$ at prices $\pp\ge 0$ for $\CM'$, then 
$\sum_i x_{is} < 1$. That is $p_s=0$ at equilibrium. 
\end{lemma}
\begin{proof}
It is easy to see that $x_{is} \le \frac{1}{n+1},\ \forall i$, and hence the proof follows.
\end{proof}

Next we show that equilibria of $\CM$ and $\CM'$ are related. 

\begin{lemma}\label{lem:s}
If $\CM$ satisfies strong feasibility, then every equilibrium of $\CM'$ gives an equilibrium of $\CM$.  
\end{lemma}
\begin{proof}
Let $(\xx^*,\pp^*)$ respectively be an equilibrium allocation and prices of $\CM'$. From Lemma \ref{lem:pszero}, 
we know that $p^*_s=0$. It suffices to show that $x^*_{is}=0,\ \forall i\in A$, for the lemma to follow. 

To the contrary suppose for some agent $u$, $x^*_{us}>0$. We will construct another bundle $\x'_u$ that is affordable to agent $u$ at prices $\pp^*$, satisfies CC$(u)$, and has a lower delay than $\x^*_u$, contradicting optimal bundle condition at equilibrium. 

Due to {\em strong feasibility}, after all $i\neq u$ is given their bundle $\x^*_i$, there will be a bundle $\x'_u$ left for $u$ to buy among goods in $G$ such that $CC(u)$ constraints are satisfied. Clearly, one way to construct such $\x'_u$ is that the agent keeps buying all goods $j\neq s$ as in $\x^*_u$, and
starts decreasing $x^*_{us}$ and increasing allocation for some other available good. Note that all such available goods are under-sold at equilibrium and therefore has zero price in $\pp^*$. Thus payment for $\x'_u$ and $\x^*_u$ are the same at prices $\pp^*$. In other words $\x'_u$ is affordable at prices $\pp^*$. 

If good $j$ is increased by $\delta_j$ as we go from $\x^*_u$ to $\x'_u$, then we claim that $\delta_j \le m! a_s \left(\frac{a_{max}}{a_{min}}\right)^m x^*_{us}$, where $m=|G|$. This is because, in a constraint even if coefficient of variable $x_{uj}$ is minimum possible, and it needs to compensate for increase in other goods due to their negative coefficients, this cascade could at most harm by a factor of $m! \left(\frac{a_{max}}{a_{min}}\right)^m$. Difference in delay is 
\[
\begin{array}{lcl}
\sum_j d_{uj} (x'_{uj} -x^*_{uj}) - d_s x^*_{us} & = & \sum_j d_{uj} \delta_j - d_s x^*_{us}\\
& \le&  m d_{max}( \max_j \delta_j) - d_s x^*_{us}\\
& \le & \left(m^{(m+1)} d_{max} a_s \left(\frac{a_{max}}{a_{min}}\right)^m -d_s\right) x^*_{us} <0
\end{array}
\]
\end{proof}

Due to Lemma \ref{lem:s}, to show existence of equilibrium for market $\CM$ it suffices to show one for $\CM'$. 
Next to show existence of equilibrium for market $\CM'$ it suffices to consider price vectors where $p_s=0$ due to Lemma \ref{lem:pszero}, and therefore we consider the following set of possible price vectors. 

\[P=\{\pp\in \Rplus^{(m+1)}\ |\ p_s=0;\ \sum_{j\in G} p_j \le M\} \ \ \ \mbox{ where $M=\sum_i m_i$}\]

Let us first handle trivial instances. 
It is easy to see that the feasible set of $\x_i$s in LP (\ref{eq:opti}) at
$\pp=\zeros$ is a superset of the feasible set at any other prices $\pp$.
Therefore, for agent $i$ if $x_i=\zeros\in OPT_i(\zeros)$, then she will not
buy anything at any prices. In that case, it is safe to discard her from the
market.  Further, if there is an allocation $\xx$ satisfying \ref{eq:supply}
for market $\CM'$ such that $\x_i \in OPT_i(\zeros),\ \forall i\in A$, then we
get a trivial equilibrium of $\CM'$ where all the prices are set to zero; note
that in this case zero prices also constitute an equilibrium of market $\CM$ by
Lemma \ref{lem:s}. To show existence for non-trivial instances, w.l.o.g. now on
we assume the following for market $\CM'$.
\medskip

\noindent{\bf Weak Sufficient Demand (WSD):} 
If $\xx$ is such that $\forall i \in A,\ \x_i \in OPT_i(\zeros)$, then $\x_i\neq \zeros,\ \forall i$, and there
exists a good $j\in G'$ such that $\sum_i x_{ij} > 1$. Clearly, $j\neq s$ due
to Lemma \ref{lem:pszero}.

\begin{lemma}\label{lem:ne}
For any $\pp \in P$, $OPT_i(\pp)$ is non-empty, and assuming WSD, $\zeros \notin OPT_i(\pp),\ \forall i\in A$.  
\end{lemma}
\begin{proof}
The first part of the proof is easy to see due to the extra good $s$ whose price is zero in $P$. 
For the second part, to the contrary suppose for $\x_i \in OPT_i(\pp)$ we have $\x_i=\zeros$. By WSD assumption we know that for any $\x^0_i \notin OPT_i(\zeros)$, we have $\x^0_i\neq \zeros$. Further, feasible set of LP (\ref{eq:opti}) is at prices $\pp$ is a subset of the feasible set of this LP at zero prices. Hence we know that $\sum_j d_{ij} x^0_{ij} < \sum_j d_{ij}x_{ij}=0$.

It must be the case that $\x^0_i$ is not affordable at prices $\pp$, i.e., $m_i < \sum_j x^0_{ij} p_j$. For $\lambda=m_i/\sum_j x^0_{ij} p_j$, set $\x'_i = \lambda\x^0_i + (1-\lambda) \x_i$. Since $\x_i$ and $\x^0_i$ both satisfy \ref{eq:CC}, so does $\x'_i$. And since $\x_i=\zeros$ bundle $\x'_i$ is affordable at prices $\pp$, thereby $\x'_i$ is feasible in $OPT_i(\pp)$. The delay at $\x'_i$ is $\sum_i d_{ij} x'_{ij}  = \sum_j d_{ij} (\lambda x^0_{ij} + (1-\lambda) x_{ij}) = \lambda \sum_j d_{ij} x^0_{ij} < 0 = \sum_j d_{ij} x_{ij}$, a contradiction to $\x_i$ being optimal bundle at prices $\pp$.
\end{proof}

Next we will construct a correspondence whose fixed points are exactly the market equilibria of $\CM'$.
Let $c_j^{max}$ be the maximum possible demand of good $j$; we can compute $c_j^{max}$ by maximizing $\sum_i x_{ij}$ over the \ref{eq:CC} constraints of all agents $i\in A$. 
Define domain $$D=\{(\xx,\pp)\ |\ \pp \in P;\ \ \ \ \xx\ge 0;\  \ \ \ \forall j,\ \sum_i x_{ij}\le c_j^{max}\}$$
Let $\delta=\min_i m_i$.  Define correspondence $F:D\rightarrow D$ as
follows where for a given $(\bxx,\bpp)\in D$, we have $(\xx',\pp') \in F(\bxx,\bpp)$,

\begin{equation}\label{eq.F}
\forall i\in A,\ \x'_i \in OPT_i(\bpp), \ \ \ \mbox{ and } \ \ \ \pp' \in \argmax_{\pp\in P,\ \delta \le \sum_j  p_j \le \max\{\delta,
\sum_{i\in A, j\in G'} \bar{x}_{ij}\bar{p}_j\}} \sum_{i\in A, j\in G'} \bar{x}_{ij} p_j 
\end{equation}

The correspondence is well defined due to Lemma \ref{lem:ne}.  If $F(\bxx,\bpp)$ is a convex set, and graph of $F$ is closed, then
Kakutani's Theorem \cite{kakutani} implies that F has a fixed point, i.e., $\exists (\xx^*,\pp^*)\in D$ such that $(\xx^*,\pp^*)\in F(\xx^*,\pp^*)$. 
Next we show the same. 

\begin{lemma}\label{lem:closed}
Correspondence $F$ has a fixed point.
\end{lemma}
\begin{proof}
Clearly, $F(\bxx,\bpp)$ is a convex set since it is a cross product of solution sets of LPs. 
The lemma follows using Kakutani's fixed-point Theorem \cite{kakutani} if graph of $F$ is closed.

Let $(\bxx^t,\bpp^t)$ for $t=1,2,\dots$ be a sequence of points in $D$, and let $(\xx^t,\pp^t)$ be another sequence such that $(\xx^t,\pp^t)\in
F(\bxx^t,\bpp^t)$. Then essentially, $\xx^t$ and $\pp^t$ are solutions of LPs that are continuously changing with $(\bxx,\bpp)$.
Therefore, if $\lim_{t\ra \infty} (\bxx^t,\bpp^t) = (\bxx^*,\bpp^*)$ and $\lim_{t\ra \infty} (\xx^t,\pp^t) = (\xx^*,\pp^*)$, then by
continuity of parameterized LP solutions, we get that $(\xx^*,\pp^*) \in F(\bxx^*,\bpp^*)$, implying graph of $F$ is closed. 
\end{proof}

\begin{lemma}\label{lem:dj}
If $(\xx^*,\pp^*)$ is a fixed-point of $F$ then $\forall j\in G', \sum_{i\in A} x^*_{ij} \le 1$. 
\end{lemma}
\begin{proof}
Since $\x^*_i \in OPT_i(\pp^*),\ \forall i\in A$, $\sum_{i\in A} x^*_{is}\le 1$ follows using Lemma \ref{lem:pszero}.
Among the rest of the goods, suppose, for $j'\neq s$ we have $\sum_{i\in A} x^*_{ij'} > 1$. 
Let $U=\max\{\delta,\sum_{i,j} x^*_{ij} p^*_j\}$. Then the optimization problem of (\ref{eq.F}) is 

\[
\max_{\pp \in P,\ \delta \le \sum_j p_j \le U} \sum_{i,j} x^*_{ij}p_j = \max_{\pp\in P,\ \delta \le \sum_j p_j \le U} \sum_{j}p_j \sum_i
x^*_{ij} 
\]
The above quantity can be made more than $U$ by setting $p_{j'}=U$, and therefore optimal value is strictly more than $U$. However, due
to the fixed-point condition $\pp^*$ is a solution of the above, which implies $
\sum_{i,j} x^*_{ij} p^*_j > U$, a contradiction. 
\end{proof}

\begin{lemma}\label{lem:fpf}
Assuming {\em weak sufficient demand (WSD)}, if $(\xx^*,\pp^*)$ is a fixed-point of $F$ then $\pp^*\neq 0$ and $\sum_{i,j}
x^*_{ij} p^*_j \ge \delta$.  \end{lemma}
\begin{proof}
For the first part, to the contrary suppose $\pp^*=0$. Then $\forall i \in A, \x^*_i \in OPT_i(\zeros)$ since $(\xx^*,\pp^*)$, and therefore by WSD condition $\exists j\neq s,\ x^*_{ij} > 1$. However $\sum_{i,j} x^*_{ij}p^*_j=0$. Therefore, 
\[
\max_{\pp\in P,\ \delta \le \sum_j  p_j \le \max\{\delta, \sum_{i,j} x^*_{ij} p^*_j\}} \sum_{i,j} x^*_{ij} p_j = \max_{\pp\in P,\
\sum_j p_j = \delta} \sum_{i,j} x^*_{ij} p_j >0
\]
This contradicts the fact that $\pp^*$ is a maximizer of the above.

For the second part, to the contrary suppose $\sum_{i,j}x^*_{ij} p^*_j<\delta$, then $\max\{\delta, \sum_{i,j}x^*_{ij} p^*_j\}=\delta$.
Further, due to Lemma \ref{lem:ne} (together with the WSD assumption) we have $\forall i \in A, \x^*_i \neq \zeros$, and therefore at maximum $\sum_j p_j
=\delta$. Since $\pp^*$ is a maximizer of the above, $\sum_j p^*_j = \delta$. 

However since $\sum_{i,j}x^*_{ij} p^*_j<\delta$, we get $\forall i, \sum_j x^*_{ij} p^*_j < m_i$ where $\x^*_i$ is feasible in LP
(\ref{eq:opti}) at prices $\pp^*$. 
Let $\xx^0$ be a demand vector when all the prices are zero. Due to Lemma \ref{lem:dj} and WSD assumption we get that
for some agent $i$, $\x^*_i \notin OPT_i(\zeros)$. Let $i'$ be this agent. This implies $\sum_j d_{i'j} x^0_{i'j} < \sum_j d_{i'j}
x^*_{i'j}$. Despite lower cost at $\x^0_{i'}$ she demands $\x^*_{i'}$ at prices $\pp^*$, hence it should be the case that she can
not afford $\x^0_{i'}$ at those prices. However, since $i'$ is also not spending all the money at prices $\pp^*$, there exists some
$0<\tau<1$ such that she can
afford $\xx'_{i'}=\tau\x^*_{i'} + (1-\tau) \x^0_{i'}$ at $\pp^*$. Since both $\x^0_{i'}$ and $\x^*_{i'}$ satisfy the
$CC(i')$ constraints of LP (\ref{eq:opti}) for agent $i'$, so does $\xx'_{i'}$. Thus, $\xx'_{i'}$ is a feasible point in
$OPT_{i'}(\pp^*)$ LP, and $\sum_j d_{i'j} x'_{i'j}<\sum_j d_{i'j} x^*_{i'j}$, a contradiction to $\x^*_{i'} \in OPT_{i'}(\pp^*)$. 
\end{proof}

Next we show the main result using Lemmas \ref{lem:closed}, \ref{lem:dj} and \ref{lem:fpf}.

\begin{theorem}\label{thm:eqcmp}
If $\CM'$ satisfies strong feasibility then it has an equilibrium. 
\end{theorem}
\begin{proof}
Due to Lemma \ref{lem:closed}, we know that there exists a fixed-point of correspondence $F$. Let this be $(\xx^*,\pp^*)$. We will show
that it is a market equilibrium of $\CM'$. Clearly, optimal allocation condition is satisfied because $\x^*_i \in OPT_i(\pp^*)$. Market
clearing remains to be shown, which requires: $(a)$
$\forall j,\ \sum_i x^*_{ij} \le 1$, and $(b)$ $p^*_j >0 \Rightarrow \sum_i x^*_{ij} = 1$.

$(a)$ follows from Lemma \ref{lem:dj}. For $(b)$, 
let $U=\max\{\delta,\sum_{i,j} x^*_{ij} p^*_j\}$, then due to Lemma \ref{lem:fpf}, $U=\sum_{i,j} x^*_{ij}
p^*_j$. The optimization problem of (\ref{eq.F}) is 

\[
\max_{\pp\in P,\ \delta \le \sum_j p_j \le U} \sum_{i,j} x^*_{ij}p_j = \max_{\pp\in P,\ \delta \le \sum_j p_j \le U} \sum_{j}p_j \sum_i
x^*_{ij} 
\]

Clearly, at optimal solution of the above $p_j$ is non-zero only where $\sum_i
x^*_{ij}$ is maximum. Since $\pp^*$ is a solution, if $\exists j,\ \sum_i x^*_{ij}=1$, then $(b)$ follows. 

On the other hand suppose for all $j$ we have $\sum_i x^*_{ij}< 1$, then
clearly the optimal value of the above is strictly less than $U$. However since $\pp^*$ is a maximizer it implies that
$\sum_{i,j} x^*_{ij} p^*_j<U$, a contradiction to Lemma \ref{lem:fpf}. 
\end{proof}

The next theorem follows using Lemmas \ref{lem:sf-cmp}, \ref{lem:s}, and Theorem \ref{thm:eqcmp}.

\existence*

\begin{remark}
By similar argument, we can show existence of equilibrium for market instances satisfying only extensibility, and sufficient demand
conditions (see Definitions \ref{def:extensibility} and \ref{def:ED}). However, since our algorithm returns an equilibrium of such a
market, it already gives a constructive proof of existence.  
\end{remark}

\subsection{Quasi-concave utility functions}
\label{sec:quasiconcave}
In this section, we show that the preferences of agents in our model can be captured by quasi-concave utility functions. 
Notation: the symbol $\leq$ when used for vectors represents a co-ordinate wise relation, and $<$ represents that at least one of the inequalities is strict.
Define the utility of an agent $i$ for an allocation $\x_i$ 
to be the smallest delay of a feasible allocation dominated by $\x_i$, times $-1$: 
$$ \textstyle U_i(\x_i) = - \min \left\{ \d_i \cdot \x_i' : \x_i' \leq \x_i~\& ~\x_i' \text{ is feasible for \ref{eq:coveringLP}} \right\} .$$
If there is no  $\x_i' \leq \x_i$ that is feasible for \ref{eq:coveringLP} then the utility is  $-\infty$. 
It is easy to check that this utility function is quasi-concave, and 
induces the same preferences as in our model. 

\section{Special Cases}
\label{app:specialcases}
In this section, we show how several natural problems satisfy the extensibility condition (Definition \ref{def:extensibility}) in
Section \ref{sec:model}.  

\paragraph{Scheduling.}

Recall the scheduling problem from Section \ref{sec:specialcases}. The agents are jobs that require $d$ different types of machines, and the set of machines of type $k$ is $M_k$; the machines are the goods in the market. Each agent needs $r_{ik}\in \Rplus$ units of machines in all of type $k$, which is captured by the covering constraint $\sum_{j \in M_k} x_{ijk} \ge r_{ik},\ \forall k\in [d]$. All agents experience the same delay $d_{jk}$ from machine $j$ in type $k$. 

\begin{lemma}
Scheduling problem satisfies the extensibility condition.  
\end{lemma}
\begin{proof}
Consider an arbitrary set of agents $S$ and an agent $\hat{i}$ outside of this set. Let $(\x_i)_{i\in S}$ to be a feasible allocation that minimizes the total delay, i.e., $\sum_{i\in S, j\in M_k, k\in [d]} d_{jk}x_{ijk}$. Since the delays are the same for each agent the total delay minimizes when the agents in $S$ gets $\sum_{i\in S} r_{ik}$ units of machines of type $k$ with the smallest delay. Therefore, if we assign the next $r_{ik}$ units of machines of type $k$ with the smallest delay to agent $\hat{i}$ then $(x_i)_{i\in (S\cup \hat{i})}$ would be the feasible allocation that minimizes the total delay. Therefore, this problem satisfies the extensibility condition.
\end{proof}

\paragraph{Restricted assignment with laminar families.}
%

\begin{lemma}
 Restricted assignment with laminar families satisfies the extensibility condition if the following assumption holds for the instance
\begin{itemize}
\item (Monotonicity) $\forall i,i' \in A$, such that  $S_{i'} \subset S_i $ then $\max_{ j \in S_i \setminus S_i'}d_{jk} \leq \min_{j'\in S_i'} d_{j'k}$ for each type $k \in [d]$.
\end{itemize}
\end{lemma}
\begin{proof}
Since the requirement and variables for set of machines of each type is separate, it is enough to show this for $k=1$.
Consider an arbitrary set of agents $T$ and an agent outside of this set $\hat{i}$.
Let $T'=T\cup i$.
Let $(\x_i)_{i\in T}$ to be a feasible allocation that minimizes total delay, i.e., $\sum_{i\in T}\d . \x_i$.
For simplicity let $S_{\hat{i}} = \{ 1,2,\ldots,m\}$ such that $d_{1}\le \cdots \le d_m$. 
Consider two agent $i$ and $i'$. Note that if $j\in S_i$ and $j\in S_{i'}$ then $\forall j' \in S_i$ such that $d_{j'} \leq d_j$ we have $j\in S_{i'}$ and vice versa because $S_i$s form a laminar family and the monotonicity condition. 
Therefore, an optimal allocation allocates only a prefix of machines in $S_{\hat{i}}$.
Let's assign to $\hat{i}$ the next $r_{\hat{i}}$ machines with smallest delay in subset $S_{\hat{i}}$. 
We claim the new allocation $(\x_i)_{i\in S'}$ is minimizing the total delay, i.e., $\sum_{i\in S'}\d . \x_i$.
Let's prove the claim by contradiction.
Suppose the allocation is not optimal. Therefore there exists an optimal allocation $(\x'_i)_{i\in S' \cup \{\hat{i} \} }$ with less total delay.
Suppose allocation $\xx'$ has allocated machines 1 through $l'$ in set $S_{\hat{i}}$ and allocation $\xx$ has allocated machines 1 through $l$ in set $S_{\hat{i}}$.
Note that in allocation $\xx'$ we can assume agent $\hat{i}$ is getting the last machines in $\xx'$ among the machines that has been allocated
 in $S_{i}$ because if there exists agent $\bar{i}$ that has allocation on the right side of the first machine that is allocated to $\hat{i}$ then we can swap the allocations so the total delay wouldn't change. 
If $l=l'$ then the delay of agent $\hat{i}$ is the same in $\xx$ and $\xx'$. This is a contradiction because then it would conclude the total delay allocation 
of $(\x'_i)_{i\in S}$ is less than total delay of $(\x_i)_{i\in S}$ but we assumed $\xx$ is an optimal allocation for $S$.
There are two cases.\\
\textbf{Case1.} $l<l'$. Consider allocation $\xx'$ after removing agent $\hat{i}$. Since $\hat{i}$ is getting the last machines it is easy to see the remaining allocation is optimal for $S$. Since $l<l'$ there are agents that have less allocation in $S_i$ in $\xx$ compare to $\xx'$. 
Because of monotonicity assumption they have been allocated to machines with highest delay instead of available machines in $S_i$. This is a contradiction with the fact $\xx$ is a optimal allocation for $S$.\\
\textbf{Case2.} $l>l'$. This case is very similar to the last case. With the same argument we can argue that this case has contradiction with the fact $\xx'$ is an optimal allocation.
\end{proof}
%

\paragraph{Network Flows.}
Recall that in this setting agent $i$ wants to sent $r_i$ units of flow from $s$ to $t$ in a directed (graph) network where each edge has a capacity and cost per unit flow specified. Here edges are goods, and the covering constraints of agent $i$ has variable $f_{ie}$ for each edge $e$ representing her flow on edge $e$. The constraints impose flow conservation at all nodes except $s$ and $t$, and that net outgoing and incoming flow at $s$ and $t$ respectively is $r_i$. 

\begin{lemma}
A series-parallel network satisfies the extensibility condition.
\end{lemma}
\begin{proof}
Consider an arbitrary set of agents $S$ and an agent outside of this set $\hat{i}$.
Let $(f_i)_{i\in S}$ to be a feasible min cost flow.
Let's remove the allocated capacities from the graph and allocate min cost flow of size $r_{\hat{i}}$
to agent $\hat{i}$ in the remaining graph. It is known that this greedy algorithm gives a min cost flow of size $\sum_{i\in S \cup \hat{i}} r_i$ \cite{SeriesParallel}.
\end{proof}

We can consider independent copies of any of these special cases above. E.g., each
agent might want some machines for job processing, as well as send some flows through a network, but have
a common budget for both together. Note that it would remain extensible because each copy is independent and extensible itself.

\section{Equilibrium Characterization}\label{sec:eqchar}
In this section we characterize equilibria of most general market instances. 
Recall that allocation $\x_i$ of agent $i$ has to satisfy its covering constraints \ref{eq:CC}. 
Next we derive sufficient conditions for prices $\pp$ and allocation $\xx$ to be an equilibrium.
As we discussed in Section \ref{sec:model}, given prices $\pp$ the optimal bundle of each agent $i$ is captured by

\[\tag{OB-LP($i$)}
\begin{array}{lll}
\min & \sum_{j\in \CG} d_{ij} x_{ij}\\
s.t. & \sum_{j\in \CG} a_{ijk} x_{ij} \geq r_{ik} & \forall k\in \CC\\
	& \sum_{j\in \CG} p_{j} x_{ij} \leq m_i\\
	& x_{ij}\ge 0,\ \forall j \in \CG
\end{array}
\]

It is well know that solutions of a linear program are exactly the ones that satisfy its complementary slackness conditions \cite{Sch-book}.
Let $\beta_{ik}$ and $\gamma_i$ be the dual variables of first and second of constraints in OB-LP($i$). 
Then, corresponding complementary slackness conditions are, 

\begin{equation}\label{eq:ob-csc}
\begin{array}{llll}
\forall k \in \CC: & \sum_{j\in G} a_{ijk} x_{ij} \geq r_{ik} & \perp & \beta_{ik} \geq 0 \\
\forall j \in \CG: & d_{ij} \geq \sum_{k} a_{ijk} \beta_{ik} - \gamma_i p_j & \perp&  x_{ij} \geq 0	\\
& \sum_{j\in G} p_{j} x_{ij} \leq m_i & \perp & \gamma_i \geq 0 \\
\end{array}
\end{equation}

Here $\perp$ symbol between two inequalities means that both inequalities should be satisfied, and at least one of them has to hold with
equality. Next, under a mild condition of {\em sufficient demand} (see Definition \ref{def:ED}), we show that w.r.t. equilibrium prices $\gamma_i>0,\ \forall i \in A$. This would also imply that each agent spends all of its budget at equilibrium (thereby maximizing profit of the seller). 
We note that the following lemma is not needed for Theorem \ref{thm:lpl2eq}, but will be used to show the reverse result in Lemma \ref{lem:eq2lpl}, namely equilibrium prices and allocation gives solutions of the $LP(\ll)$ and $DLP(\ll)$ for appropriately chosen parameter vector $\ll$. 

\begin{lemma}\label{lem:money}
Let $(\xx^*,\pp^*)$ be an equilibrium of market $\CM$ satisfying sufficient demand, and for all $i\in A$ let $(\gamma^*_i,\B^*_i)$ be the corresponding dual variables of OB-LP$(i)$ at $\pp^*$. Then $\gamma^*_i>0,\ \forall i \in A$. In other words, every agent spends all her money at any equilibrium. 
\end{lemma}
\begin{proof}
Since $\xx^*$ is an equilibrium allocation, no good is over demanded. Therefore, we know that $\x^*_i\le \ones,\ \forall i\in A$. Let $\pp^0$ be all zero price vector. By {\em sufficient demand} condition $\x^*_i$ cannot be an {\em optimal allocation} of agent $i$ at prices $\pp^0$. 
We will derive a contradiction to this if $\gamma^*_i=0$ for some agent $i$. 

To the contrary suppose for agent $i\in A$, we have $\gamma_i=0$. 
Since $\x^*_i$, $\gamma^*_i$ and $\B^*_i$ satisfy complementary slackness conditions (\ref{eq:csc}) of OB-LP($i$) at prices $\pp^*$, it is easy to see that they also satisfy these conditions at prices $\pp^0$. 
This would imply that $\x^*_i$ is an optimal allocation of agent $i$ at prices $\pp^0$, a contradiction. 
\end{proof}

Note that due to Lemma \ref{lem:money}, it is no loss of generality to assume $\gamma_i>0,\ \forall i\in A$. 
We simplify the conditions with change of variables $\lambda_i=1/\gamma_i$ and
$\alpha_{ik}=\beta_{ik}/\gamma_i$, and write conditions for all the agents together.

\begin{equation}\label{eq:csc}
\begin{array}{llll}
\forall i \in A,\ \forall k \in \CC: & \sum_{j\in G} a_{ijk} x_{ij} \geq r_{ik} & \perp &  \alpha_{ik} \geq 0 \\
\forall i\in A, \forall j \in \CG: & \lambda_i d_{ij} \geq \sum_{k} a_{ijk} \alpha_{ik} - p_j & \perp & x_{ij} \geq 0	\\
\forall i \in A: & \sum_{j\in G} p_{j} x_{ij} = m_i & \mbox{and} & \lambda_i > 0 \\
\end{array}
\end{equation}

At market equilibrium prices, every agent should get an optimal bundle, and market should clear, i.e., \ref{eq:supply} 
are satisfied and every good with positive price should be fully sold (see Section \ref{sec:model} for the formal definition of market
equilibrium). Since optimal allocations at given prices are solutions of
OB-LP($i$) for each $i$, they must satisfy (\ref{eq:csc}). 
This follows from the fact that primal-dual feasibility and complementary slackness conditions are necessary and sufficient for solutions of a linear program. 
We get the following characterization. 

\begin{lemma}\label{lem:eqchar}
If $(\hll,\hxx,\hpp,\haa)$ satisfies (\ref{eq:csc}), and $\forall j\in G,\
\sum_{i\in A} \hx_{ij} \le 1 \ \perp\ \hp_j \ge 0$, then $(\hxx,\hpp)$
constitutes an equilibrium allocation and prices.
\end{lemma}

Motivated from Lemma \ref{lem:eqchar} we next define parameterized LP that captures complementary slackness conditions of OB-LP$(i)$ for all the agents together. Suppose we are given $\lambda_i$'s, let us define the following linear program parameterized by the $\ll$ vector, that we call $LP(\ll)$,
and its dual $DLP(\ll)$ (same as (\ref{eq:lpl}) defined in Section \ref{sec:general}):

\[
\begin{array}{cc}
LP(\ll) & DLP(\ll)\\
\begin{array}{lll}
\min : & \sum_{i,j} \lambda_i d_{ij} x_{ij}\\
 s.t. & \sum_{j\in G} a_{ijk} x_{ij} \geq r_{ik}  & \forall i\in A, k\in C \\ 
  & \sum_{i\in A} x_{ij} \le 1& \forall j\in G \\
 & x_{ij} \ge 0& \forall i\in A,j\in G\\
\end{array}
\ \ \ \ & \ \ \ \ 
\begin{array}{lll}
\max:&  \sum_{i\in A,k\in C} r_{ik} \alpha_{ik}  - \sum_{j} p_j \\
 s.t. & \lambda_i d_{ij} \geq \sum_k a_{ijk} \alpha_{ik} - p_j  & \forall i\in A, j\in G  \\
 & \alpha_{ik}, p_j \geq 0 & \forall j\in G, k\in C
\end{array}
\end{array}
\]

Using the equilibrium characterization of Lemma \ref{lem:eqchar} together with complementary slackness conditions between constraints
of $LP(\ll)$ and $DLP(\ll)$, next we show that solutions of $LP(\ll)$ and $DLP(\ll)$ exactly capture the equilibria if given appropriate value of
parameter vector $\ll$. 

\lpleq*
\begin{proof}
It suffices to show that $(\ll,\haa,\hpp,\hxx)$ satisfy conditions of Lemma \ref{lem:eqchar}. 
Out of these the last one of (\ref{eq:csc}) is already assumed in the hypothesis. For the remaining conditions, 
let us write the complementary slackness conditions for $LP(\ll)$.
\begin{equation}\label{csc_3}
\sum_{j} a_{ijk} x_{ij} \geq r_{ik}\ \ \  \perp\ \ \  \alpha_{ik}\geq 0, \ \  \forall i,k
\end{equation}
\begin{equation}\label{csc_1}
\lambda_i d_{ij} \geq \sum_k a_{ijk} \alpha_{ik} - p_j\ \ \ \perp \ \ \ x_{ij}\geq 0, \ \  \forall i,j
\end{equation}
\begin{equation}\label{csc_2}
\sum_{i} x_{ij} \le 1\ \ \  \perp\ \ \  p_j \geq 0, \ \  \forall j
\end{equation}

Conditions (\ref{csc_3}) and (\ref{csc_1}) are exactly the first two conditions of (\ref{eq:csc}), 
and (\ref{csc_2}) ensures market clearing. Thus proof follows using Lemma \ref{lem:eqchar}. 
\end{proof}

Next we show converse of the above theorem under {\em sufficient demand} condition. 

\begin{lemma}\label{lem:eq2lpl}
If market $\CM$ satisfies {\em sufficient demand} and $(\hxx,\hpp)$ is its equilibrium, then $\hxx$ and $(\haa,\hpp)$ gives 
solution of $LP(\hll)$ and $DLP(\hll)$ respectively for some $\hll$ and $\haa$. 
\end{lemma}
\begin{proof}
Due to optimal bundle condition of equilibrium we know that $\hhx_i$ is an optimal solution of OB-LP($i$) at prices $\hpp$. Let
$\hbt_{ik}$ and $\hgm_i$ be the value of corresponding dual variable in
(\ref{eq:ob-csc}). Using Lemma \ref{lem:money} we have $\hgm_i>0,\ \forall i$.
Then it is easy to see using (\ref{eq:csc}) that for $\hl_i=\nfrac{1}{\hgm_i},\ \forall i\in A$ and
$\ha_{ik} = \frac{\hbt_{ik}}{\hgm_i},\ \forall i\in A, k\in C$, $(\hxx,\haa,\hpp)$ satisfy
complementary slackness conditions of $LP(\hll)$ and $DLP(\hll)$. Therefore the proof follows.
\end{proof}

Using the equilibrium characterization given by Theorem \ref{thm:lpl2eq} crucially, we design a polynomial-time algorithm to find
an equilibrium for markets with extensibility (Definition \ref{def:extensibility}) in Section \ref{sec:general}.

\section{Missing Proofs and Details of Section 4}\label{sec:fullproofs}

The proof of first theorem of Section \ref{sec:general}, namely Theorem \ref{thm:lpl2eq} is in Appendix \ref{sec:eqchar} where we characterize
market equilibria.
Next we give proof of Lemma \ref{lem:lp-prop}. For this we basically use the fact that any pair of primal, dual solutions of a linear program has to
satisfy complementary slackness \cite{Sch-book}. Recall the $LP(\ll)$ and $DLP(\ll)$ of (\ref{eq:lpl}). 

\begin{equation}
\begin{array}{cc}
LP(\ll): & DLP(\ll)\\
\begin{array}{ll}
\min: & \sum_i \li \sum_j d_{ij} x_{ij} \\
s.t. & \sum_j a_{ijk} x_{ij} \ge r_{ik},\ \forall k\\
& \sum_i x_{ij} \le 1,\ \forall i\\
& x_{ij} \ge 0,\ \forall (i,j)
\end{array}
\ \ \ \ \ 
& 
\ \ \ \ \ 
\begin{array}{ll}
\max: & \sum_{i,k}r_{ik} \alpha_{ik} - \sum_j p_j \\
s.t. & \lambda_i d_{ij} \ge \sum_k a_{ijk} \alpha_{ik} - p_j,\ \forall(i,j)\\
& p_j \ge 0,\ \forall j;\ \ \ \ \alpha_{ik} \ge 0,\ \forall (i,k)
\end{array}
\end{array}
\end{equation}

For any given $\ll>0$, the optimal solutions of the $LP(\ll)$, namely $\xx$, and $DLP(\ll)$, namely $(\aal,\pp)$ has to satisfy the following complementary slackness conditions. 

\begin{equation}
\begin{array}{llcl}
\forall (i,k): & \alpha_{ik} \ge 0 & \perp  & \sum_j a_{ijk} x_{ij}\ge r_{ik}\\
\forall j: & p_j \ge 0 & \perp & \sum_i x_{ij} \le 1\\
\forall (i,j): & x_{ij} \ge 0 & \perp & \lambda_i d_{ij} \ge \sum_k a_{ijk} \alpha_{ik} - p_j
\end{array}
\end{equation}

Recall that the $\perp$ symbol between two inequalities means that both
inequalities should be satisfied, and at least one of them has to hold with
equality. Also recall that for subset $S\subseteq A$, we defined $\delay_S(\xx,\pp)=\sum_{i\in S, j} d_{ij} x_{ij}$ and $\pay_S(\xx,\pp)=\sum_{i\in S,j} x_{ij} p_j$, and for ease of notation we use $\delay_i$ when $S=\{i\}$ and similarly $\pay_i$. Using this we first show a relation between $\delay_i$ and $\pay_i$
next. 

\begin{lemma}\label{lem:csc}
For a given $\ll>0$ if $\hxx$ and $(\haa,\hpp)$ are optimal solutions of $LP(\ll)$ and $DLP(\ll)$ respectively, while $\xx'$ and $(\aal',\pp')$ are feasible in $LP(\ll)$ and $DLP(\ll)$ then, such that, 
\[
\begin{array}{ll}
\forall i: & \lambda_i \delay_i(\hxx) = \sum_{k} r_{ik} \ha_{ik} - \pay_i(\hxx,\hpp)\\
\forall i: & \lambda_i \delay_i(\xx') \ge \sum_k r_{ik} \alpha'_{ik} - \pay_i(\xx',\pp')
\end{array}
\]
\end{lemma} 
\begin{proof}
For feasible solutions we have $\forall (i,k),\ \alpha'_{ik} \ge 0$ and $\sum_j a_{ijk} x'_{ij}\ge r_{ik}$. Multiplying the two and
summing over $k$ for each $i$ gives, $\sum_k \alpha'_{ik} \sum_{j} a_{ijk} x'_{ij} \ge \sum_k \alpha'_{ik} r_{ik},\ \forall i$. Another pair of constraints are
$\forall (i,j),\ x'_{ij} \ge 0$ and $ \lambda_i d_{ij} \ge \sum_k a_{ijk} \alpha'_{ik} - p'_j$. Multiplying the two gives,
\[
\begin{array}{lcl}
\forall i:\ \lambda_i \delay(\xx') & =&  \lambda_i \sum_j x'_{ij} d_{ij}\\
& \ge&  \sum_j x'_{ij} (\sum_k a_{ijk} \alpha'_{ik}) - \sum_j p'_j x'_{ij} \\
& = & \sum_k \alpha'_{ik} \sum_j a_{ijk} x'_{ij} - \pay_i(\xx',\pp')\\
& \ge&  \sum_k \alpha'_{ik} r_{ik} - \pay_i(\xx',\pp')
\end{array}
\]
This gives the second part. Since optimal solutions satisfy complementary slackness, all inequalities satisfy with equality in the above for
$(\hxx,\haa,\hpp)$ and we get the first part.
\end{proof}

Recall notation $[n]=\{1,\dots,n\}$ for any positive integer $n$, and Definition \ref{def:bestfor} of an allocation $\xx$ being $\bestfor$ a subset of agents $S$.

\lpprop*
\begin{proof}
Note that, $T_1=A$, $T_d=S_d$, and $S_g=T_g\setminus T_{g+1},\ \forall g\in[d-1]$.
For the first part, let us rewrite the objective function of $LP(\ll)$ as 
$$\sum_{g=2}^d (\lambda(S_g)-\lambda(S_{g-1})) \sum_{i\in T_g, j} d_{ij} x_{ij}+ \lambda(S_1) \sum_{i\in T_1, j} d_{ij} x_{ij}$$ 
Since $\CM$ satisfies extensibility, we can construct a minimum delay allocation $\xx^*$ where $S_d$ gets the best, then next
best to $S_{d-1}$, and so on to finally $S_1$. In other words, $\xx^*$ is \bestfor\ $T_g,\ \forall g\le d$. 
Let $\xx'$ be an arbitrary optimal solution of $LP(\ll)$, not constructed as $\xx^*$. Then,
\[
\forall 1\le g\le d,\ \sum_{i\in T_g, j} d_{ij} x^*_{ij} \le \sum_{i\in T_g, j} d_{ij} x'_{ij},
\]
To the contrary suppose at least one is strict inequality. Then, since $\lambda(S_d)>\lambda(S_{d-1})>\dots >\lambda(S_1)>0$ we have
\[
\begin{array}{ll}
\lambda(S_1) \sum_{i\in T_1, j} d_{ij} x^*_{ij} + \sum_{g=2}^d
(\lambda(S_g)-\lambda(S_{g-1})) \sum_{i\in T_g, j} d_{ij} x^*_{ij} < \\ 
\lambda(S_1) \sum_{i\in T_1, j} d_{ij} x'_{ij} + \sum_{g=2}^d
(\lambda(S_g)-\lambda(S_{g-1})) \sum_{i\in T_g, j} d_{ij} x'_{ij}
\end{array}
\]
A contradiction to $\xx'$ being optimal solution of $LP(\ll)$. Since any minimum delay allocation for a subset gives the same total delay,
the first part follows.

The second part essentially follows by applying Lemma \ref{lem:csc} twice. For optimal pair $\xx$ and $(\haa,\hpp)$, and pair $\xx$
and $(\aal',\pp')$ we get, $\forall g$
\[
\lambda(S_g) \delay_{S_g}(\xx) = \sum_{i \in S_g,k} r_{ik} \ha_{ik} - \pay_{S_g}(\xx,\hpp),\ \ \ \mbox{ and }\ \ \ \lambda(S_g)
\delay_{S_g}(\xx) = \sum_{i \in S_g,k} r_{ik} \alpha'_{ik} - \pay_{S_g}(\xx,\pp')
\]
implying, $\pay_{S_g}(\xx,\hpp) = \pay_{S_g}(\xx,\pp')$ if and only if $\sum_{i \in S_g,k} r_{ik} \ha_{ik}=\sum_{i \in S_g,k} r_{ik} \alpha'_{ik}$.

For the third part, applying Lemma \ref{lem:csc} to each agent $i \in S$ and then taking the sum gives $\lambda(S_g)
\delay_S(\hxx) = \sum_{i \in S, k} r_{ik} \alpha_{ik} - \pay_S(\hxx,\pp)$ and $\lambda(S_g)
\delay_S(\xx') = \sum_{i \in S, k} r_{ik} \alpha_{ik} - \pay_S(\xx',\pp)$. 
Combining the two, we get $\lambda(S_g)(\delay_S(\hxx) -\delay_S(\xx')) = \pay_S(\xx',\pp) - \pay_S(\hxx,\pp)$, thereby the lemma follows.
\end{proof}

Next we show existence of a primal-dual solution where every agent spends exactly her money if $BB$ (budget balance) and $SC$ (subset condition) of Definition \ref{def:feas} are satisfied at the current value of $\ll$ and $\pp$. Recall the definition of \feasible pair $(\ll,\pp)$ (Definition \ref{def:feasible}).

\feas*
\begin{proof}
	Let $\tau_i(\xx)=m_i-\pay_i (\xx, \pp^*)$. Without loss of generality suppose agents are ordered such that 
	$\tau_1(\xx)\geq \tau_2(\xx) \geq \cdots \geq \tau_{n}(\xx)$. Define $T_k(\xx) = \sum_{1 \leq i \leq k} \tau_{i}(\xx)$.
	Let's define the following potential function for every allocation $\xx$.
	The potential function is $f(\xx)=\sum_{k} T_k(\xx)$.

	Let $\xx^*$ be an optimal solution of $LP(\ll^*)$ that minimizes $f$.
	Order the agents such that $\tau_1(\xx^*) \geq \tau_2(\xx^*) \geq \cdots \geq \tau_n(\xx^*)$.
	Note that $T_{n}(\xx^*) = 0$ (condition $BB$).
	Therefore, $T_{i}(\xx^*) \geq 0 \ \forall i$ since $\tau_i$s are in decreasing order.
	Therefore, $f(\xx^*)\geq 0$.
	If $f(\xx^*)=0$ then it must be the case that $T_i(\xx^*)=0,\ \forall i\in [n]$ and $\tau_i(\xx^*)=0$, $\forall i\in[n]$. This gives $m_i-\pay_i(\xx^*,\pp^*)=0,\ \forall i$ as we desired.

	To the contrary suppose $f(\xx^*)>0$. Let $\hat{\xx}$ be an optimal allocation of $LP(\ll^*)$ where delay of agent $1$ is minimum, then of agent $2$, and so on, finally of agent $n$. 
	\begin{claim}
	$\sum_{i\leq r} \tau_i (\hat{\xx}) \leq 0, \ \forall r \in [n]$.
	\end{claim}
	\begin{proof}
	Fix an $r\in[n]$ and define $S=\{1,\ldots,r\}$ and $\bar{S} = \{ r+1,\ldots,n\}$.
	Since the total delay of all the agents is same at both $\xx^*$ and $\hxx$, their total payment is also same (first and third part of Lemma \ref{lem:lp-prop}). Therefore it suffices to show $\sum_{i\in \bar{S}} \tau_i (\hat{\xx}) \geq 0$ because $\sum_{1\leq i\leq n} \tau_i(\hat{\xx}) =0$.
	Let's define $L_g=\bar{S}\cap S_g$. 
	Note that $\hat{\xx}$ is an optimal allocation in which delay of $\bar{S}$ is maximized. 
	We will show that in an optimal allocation if delay of $\bar{S}$ is maximized then delay of $L_g$ is maximized for all $g$ and so $m(L_g) \geq \pay_{L_g}(\hat{\xx},\pp^*)$ (SC condition). Therefore, $m(\bar{S}) \geq \pay_{\bar{S}} (\hat{\xx},\pp^*)$. That completes the proof. 
	
	In the following we show that if the delay of $\bar{S}$ is maximized then delay of $L_g$ maximized for all $g$  in an optimal allocation.
	Consider an optimal allocation $\xx'$ of $LP(\ll^*)$ which is constructed by first optimizing for $S_d\backslash L_d$, then $L_d$, then $S_{d-1}\backslash L_{d-1}$ then $L_{d-1}$ and so on. This is a valid construction due to the extensibility property. We claim that $\xx'$ is an optimal allocation which maximizes delay of $L_g$, $\forall g$ individually.
	Total delay of $S_g$ $\forall g$ is the same for all optimal allocations (Lemma \ref{lem:lp-prop}) and delay of 
	$(\cup^{d}_{q=g+1} S_q) \cup  (S_{g}\backslash L_g)$ is minimized in $\xx'$. Therefore, delay of $S_{g}\backslash L_g$ is minimized in $\xx'$ and so delay of $L_g$ is maximized since sum of delays of $S_{g}\backslash L_g$ and $L_g$ is constant among all optimal allocations.
	\end{proof}
	Using the above claim we get $\exists \hat{r}$ such that $\sum_{i\leq \hat{r}} \tau_i (\hat{\xx}) < 0$ because otherwise $\tau_i(\hat{\xx}) = 0, \forall i $ and $f(\hat{\xx})=0$. 
	Therefore, 
	\begin{equation}\label{feasibility:ineq} \sum_{r\leq n }\sum_{i\leq r} \tau_i(\hat{\xx}) < 0 \end{equation}
	
	Let's define $\xx(\delta)=(1-\delta)\xx^*+\delta \hat{\xx}$. Since optimal solutions of an LP forms a convex set, $\xx(\delta)$ is an optimal solution of $LP(\ll^*)$ for all $\delta\in[0\ 1]$. 
	For every pair $i$ and $j$ such that $i<j$ and $\tau_i(\xx^*)=\tau_j(\xx^*)$, we assume $\frac{\partial \tau_i (\xx(\delta))}{\partial \delta}>\frac{\partial \tau_j (\xx(\delta))}{\partial \delta}$. Note that the assumption is without loss of generality.
	Therefore, there exists $\delta$ small enough such that
	the order of $\tau_i$'s is the same for $\xx(\delta)$ as at $\xx^*$. Considering that small $\delta$, we have the following
	$$ f(\xx(\delta)) = \sum_{r} \sum_{i\leq r} \tau_i(\xx(\delta))$$
	$$ = \sum_{r} \sum_{i\leq r} (\tau_i(\xx^*)-\delta(\tau_i(\xx^*)-\tau_i(\hat{\xx})))$$
	$$ = f(\xx^*) - \delta(f(\xx^*) -  \sum_{r\leq n }\sum_{i\leq r} \tau_i(\hat{\xx}) ) $$
	$$  = (1-\delta) f(\xx^*) + \delta( \sum_{r\leq n }\sum_{i\leq r} \tau_i(\hat{\xx}) ))$$
	$$ < f(\xx^*)\ \ \ \ \ \ \ \ \ \ \ \ \ \ \ \ \  \text{(Using  (\ref{feasibility:ineq}))}$$
	Therefore, we get $f(\xx(\delta))<f(\xx^*)$ which is a contradiction to $\xx^*$ being optimal solution where $f$ is minimized.
\end{proof}

While searching for the next segment during the algorithm, we would like to fix total payment of the segments that already have been
created. Note that unlike scheduling on a single machine (Section \ref{sec:cloud}), in this general setting we are not able to fix the allocation of the agents on the previous segments. In fact the allocation will heavily depend on
what segments are created later on. However from Lemma \ref{lem:lp-prop} we know that the total delay of previous segments will remain unchanged. If delay is fixed then payments can be controlled using the last part of Lemma \ref{lem:lp-prop} if prices of goods they buy also remain unchanged. 
Therefore the only way to ensure that total payments are fixed seems to be, fixing 
prices of goods they are buying, and also hold their dual variables $\alpha_{ik}$s. 
Next lemma shows that this is indeed possible. 
In a number of proofs that follows we will use the following version of Farkas' lemma. 

\begin{lemma}[Farkas' lemma \cite{schrijver.book}]\label{lem:farkas}
Given a matrix $A$, if the system $A\zz=b$ and $\zz\geq 0$ is infeasible then
there exists vector $\y$ such that $\y^T A \geq 0 $ and $\y^T b <0$.
\end{lemma}

Recall notation $\ones_S\in \{0,1\}^{|A|}$ for a subset $S\subseteq A$ denoting indicator vector of set $S$.
\price*
\begin{proof}
Due to Lemma \ref{lem:lp-prop} it follows that $\hxx$ is an optimal solution of $LP(\ll')$. 
Suppose $\pp'=\hat{\pp}+\ddelta$, then it is sufficient to show that the following system is feasible. 
$$\ddelta\ge 0$$
\begin{equation}\label{cl_1}
\begin{array}{llll}
\forall j & \ s.t.\ & \sum_{i \not \in R} \hat{x}_{ij} >0  \text{ or } \sum_{i} \hat{x}_{ij} < c_j: & \delta_j = 0
\end{array}
\end{equation}
\begin{equation}\label{cl_2}
\begin{array}{llll}
\forall i \in R, \forall j & \ s.t. \ &   \hat{x}_{ij} >0 :  & \lambda'_i d_{ij}  = \sum_{k} a_{ijk} \alpha'_{ik} - \hat{p}_j - \delta_j
\end{array}
\end{equation}
\begin{equation}\label{cl_3}
\begin{array}{llll}
\forall i \in R, \forall j & \ s.t. \ & \hat{x}_{ij} = 0 : & \lambda'_i d_{ij} \geq \sum_{k} a_{ijk} \alpha'_{ik} - \hat{p}_j - \delta_j
\end{array}
\end{equation}

The proof is by contradiction.
Suppose the system is infeasible. We will show a contradiction using Farkas' lemma (Lemma \ref{lem:farkas}).
To convert (\ref{cl_3}) to equality we add slack variable $\gamma_{ij}$.
In addition, we remove all $\delta_j$ that are set to zero in (\ref{cl_1}) from (\ref{cl_2}) and (\ref{cl_3}), and remove (\ref{cl_1}) itself from the system.
Let $T$ denote the set of goods $j$ such that $\sum_{i \not \in R} \hat{x}_{ij} >0  \text{ or } \sum_{i} \hat{x}_{ij} < c_j$  and $\bar{T}$ denote the set of good not in $T$. 
The remaining system can be written as follows in $A\zz=\pb$ form in variables $\delta_j$s and $\gamma_{ij}$s:
\begin{equation}\label{clp_2}
\begin{array}{llll}
\forall i\in R, \forall j & \ s.t. \ &   \hat{x}_{ij} >0 \ \& \ (j\in \bar{T})  :  & \lambda'_i d_{ij}+\hat{p}_j  = \sum_{k} a_{ijk} \alpha' _{ik} - \delta_j \\
\forall i\in R, \forall j & \ s.t. \ &   \hat{x}_{ij} =0 \ \& \ (j\in \bar{T}) :  & \lambda'_i d_{ij}+\hat{p}_j  = \sum_{k} a_{ijk} \alpha' _{ik} - \delta_j + \gamma_{ij}\\
\forall i\in R, \forall j & \ s.t. \ &   \hat{x}_{ij} >0 \ \& \ ( j \in T)  :  & \lambda'_i d_{ij}+\hat{p}_j  = \sum_{k} a_{ijk} \alpha' _{ik}\\
\forall i\in R, \forall j & \ s.t. \ &   \hat{x}_{ij} =0 \ \& \ ( j \in T)  :  & \lambda'_i d_{ij}+\hat{p}_j  = \sum_{k} a_{ijk} \alpha' _{ik} +\gamma_{ij}\\
\end{array}
\end{equation}
Due to Farkas' lemma if the above system is infeasible then there exists $\y$ such that $\y^TA =0$ and $\y^T \pb <0$. 
That is for variables $y_{ij},\ \forall i\in R,\ \forall j$:
\begin{equation}\label{Farkas_1}
\y^T A \geq 0 \Rightarrow 
\begin{cases}
\begin{array}{llll}
\forall i\in R, \forall k & & & : \sum_j a_{ijk} y_{ij} \geq 0 \\
\forall i \in R,\forall j & s.t. &  x_{ij}=0 & : y_{ij} \geq 0 \\
\forall j\in \bar{T} & & & : \sum_{i\in R} y_{ij} \leq 0
\end{array}
\end{cases}
\end{equation}
\begin{equation}\label{Farkas_2}
y^T b < 0 \Rightarrow 
\begin{array}{llll}
\sum_{i\in R,j } \lambda'_i d_{ij} y_{ij} + \sum_{i\in R,j} \hat{p}_j y_{ij} < 0\\
\end{array}
\end{equation}
Let's consider two cases.\\
\textbf{Case 1.} $\sum_{i\in R,j} d_{ij} y_{ij} \geq 0 $. Since $\forall i\in R,\ \lambda(S_d)=\lambda_i < \lambda'_i=\lambda(S_d)+\tau$ and $y^Tb<0$ we get
\begin{equation}\label{case1_eq_1}
\begin{array}{llll}
 \lambda(S_d) \sum_{i\in R,j }d_{ij} y_{ij} + \sum_{i\in R,j} \hat{p}_j y_{ij} < 0\\
\end{array}
\end{equation}
On the other hand, note that
\[
\begin{array}{lll}
\lambda(S_d) \sum_{i\in R,j } d_{ij} y_{ij} + \sum_{i\in R,j} \hat{p}_j y_{ij} & = 
\sum_{i\in R,j} y_{ij} ( \lambda_i d_{ij}+\hat{p}_j ) \\
& \geq \sum_{i\in R,j} y_{ij} ( \sum_k a_{ijk} \hat{\alpha}_{ik} ) & (\because \text{ $(\hat{p},\hat{\alpha})$ is a solution for $DLP(\lambda)$)}\\
& = \sum_{i\in R}\sum_k  \hat{\alpha}_{ik} \sum_j y_{ij} a_{ijk} \geq 0 & (\mbox{Using (\ref{Farkas_1})}) \\
\end{array}
\]
That is a contradiction.\\
\textbf{Case 2.} $\sum_{i\in R,j} d_{ij} y_{ij} < 0 $. We will show that $\hat{\xx}$ does not give min-cost allocation to agents of $R$.
Consider $x_{ij}=\hat{x}_{ij}+\epsilon y_{ij},\ \forall i\in R,\ \forall j$ for a small amount $\epsilon>0$ and $x_{ij}=\hat{x}_{ij},\ \forall i \notin R,\ \forall j$. Using (\ref{Farkas_1}) it follows that $\xx$ is feasible in $LP(\ll)$.
In addition, $$\sum_{i} \lambda_i \sum_j d_{ij} (x_{ij} - \hat{x}_{ij}) = \lambda(S_g) \sum_{i\in R,\ j} d_{ij} (\epsilon y_{ij}) <0$$
This contradicts $\hxx$ being optimal solution of $LP(\ll)$. 
\end{proof}

While creating next segment we need to maintain $BB$ and $SC$ conditions for all the previous segments. Lemma \ref{lem:price} ensures that prices of goods bought by agents in previous segments is fixed. If we also manage to ensure that total $\delay$ remains unchanged for previous segments, then we will be able to leverage properties from Lemma \ref{lem:lp-prop} to show $BB$ and $SC$ do remain satisfied for previous segments. The next lemma establishes exactly this. 

\begin{lemma}\label{lem:segCond1}
Given $\ll>0$ let the partition of agents by equality of $\lambda_i$ be $S_1, \dots, S_{k-1},\CA'$ such that $\lambda(S_1)<\dots<\lambda(S_{k-1})<\lambda(A')$.
For an $S_k \subset \CA'$, let $\ll'\ge \ll$ be such that the induced partition is $S_1, \dots, S_{k},\CA'\setminus S_k$ and $\lambda(S_1)<\dots<\lambda(S_{k})<\lambda(A'\setminus S_k)$. Then for any group $g<k$, we have $\delay_{S_g}(\xx)=\delay_{S_g}(\xx')$ where $\xx$ and $\xx'$ are solutions of $LP(\ll)$ and
$LP(\ll')$ respectively. And for any subset $T\subset S_g$, 
$$\max_{\xx \mbox{ optimal of } LP(\ll)} \delay_T(\xx)= \max_{\xx \mbox{ optimal of } LP(\ll')} \delay_T(\xx)$$
\end{lemma}
\begin{proof}
An optimal solution of both $LP(\ll)$ and $LP(\ll')$ first minimizes total delay of $\CA'$ then of $S_{k-1}$ and so on finally of
$S_1$ (Lemma \ref{lem:lp-prop}). Within $\CA'$, $LP(\ll')$ may first minimize for $\CA'\setminus S_k$ and then of $S_k$. Thus, the optimal solution set may shrink
as we go from $\ll$ to $\ll'$. However, due to extensibility, if $\xx$ and $\xx'$ are optimal solution of $LP(\ll)$ and $LP(\ll')$
then, $\delay_{S_g}(\xx)=\delay_{S_g}(\xx')$. Further, by Lemma \ref{lem:lp-prop} the optimal solution of both $LP(\ll)$ and $LP(\ll')$
where $\delay_T$ is maximized essentially minimizes total delay of $\CA'\cup_{q=g+1}^{k-1} S_{q} \cup (S_g\setminus T)$ and then of
$T$, then of $\cup_{q=1}^{g-1}S_{q}$. By extensibility condition $\delay_T$ at any such allocation remains the same. 
\end{proof}

The next lemma shows that if before we start our search for next segment, already created segments $S_1,\dots, S_{k-1}$ satisfies $BB$
and $SC$ w.r.t. $(\llcur,\pcur)$, and the remaining agents satisfy $SC$, then for the value of $a$ where
minimum of $f_a$ is zero, the minimizer gives the next segment without ruining the former. 
Recall that, for a given $a\ge 0$ and $\ii > 0, \llnew=\ll^a + \ii \ee{\CA'\setminus S_k}$, and the prices $\pnew$ are valid and
optimal for $DLP(\llnew)$, where $\ee{\CA'\setminus S_k}$ is an indicator vector of set $\CA'\setminus S_k$, and $\ll^a$ as defined in
(\ref{eq:la}). We will also use notation $\pp^a$ to denote prices at the valid optimal solution of $DLP(\ll^a)$ (See (\ref{eq:pa})), i.e., in the sense guaranteed by Lemma \ref{lem:price}. 

\segCond*
\begin{proof}
First part of Lemma \ref{lem:lp-prop} implies that every optimal solution of $LP(\llnew)$ minimizes delay of sets $\CA'\setminus S_k$, $S_k$, \dots, $S_1$ in that sequence, while optimal of $LP(\llcur)$ minimizes delay of sets in sequence $\CA'$, $S_{k-1}$,\dots, $S_1$. Therefore clearly the set of optimal solutions of $LP(\llnew)$ is a subset of the optimal solutions of $LP(\llcur)$. 
This together with the fact that at $(\llcur,\pcur)$ $BB$ and $SC$ are satisfied for each $g \le k-1,\ S_g$, for any optimal $\xx'$ of $LP(\llnew)$,
we have $\pay_{S_g}(\xx',\pcur)=m(S_g)$. Let $\aacur$ be corresponding dual variable, then Lemma \ref{lem:csc} gives $\lambda(S_g)
\delay_{S_g}(\xx') = \sum_{i\in S_g, k} r_{ik} \acur_{ik} - \pay_{S_g}(\xx',\pcur)$

Since $\pnew$ is a valid solution of $DLP(\llnew)$, at corresponding valid $(\aanew,\pnew)$ value of $\aanew_i$ for each $i \notin \CA'$ is same as $\aacur_i$. Using Lemma \ref{lem:csc} we get $\lambda(S_g) \delay_{S_g}(\xx') = \sum_{i\in S_g, k} r_{ik} \anew_{ik} - \pay_{S_g}(\xx',\pnew)= \sum_{i\in S_g, k} r_{ik} \acur_{ik} - \pay_{S_g}(\xx',\pnew)$. This
together with the above equality gives $\pay_{S_g}(\xx',\pnew) = \pay_{S_g}(\xx',\pcur)=m(S)$. Hence $BB$ is satisfied by
$S_1$,\dots,$S_{k-1}$ at $(\llnew,\pnew)$. By the same reasoning, and using Lemma \ref{lem:segCond1} and third part of Lemma \ref{lem:lp-prop} we get that
they also satisfy $SC$.

Note that $\llnew$ is $\ll^a$ with $\ii$ added to $\lambda_i$s of agents in $\CA'\setminus S_k$. And $\pnew$ is a valid optimal of
$DLP(\llnew)$ obtained starting from $\pp^a$ where $\alpha_{ik}$s of agents not in $\CA'\setminus S_k$, and prices of 
goods ``bought by them'' are held fixed (Lemma \ref{lem:price}).
Since function $f_{\ll,\pp}$ keeps track of surplus budget, and $S_k$ is the minimizer of $f_{\ll^a,\pp^a}$ where surplus budget of $S_k$ is zero at $\pp^a$ in addition, it follows that for every subset of $S_k$ the surplus budget is non-negative. Thus we get $SC$ for $S_k$ at $(\llnew,\pnew)$ using Lemma \ref{lem:segCond1}. For $BB$ note that we have $\llnew(S_k)<\llnew(\CA'\setminus S_k)$. Hence, due to first part of Lemma \ref{lem:lp-prop},
optimal allocations $\xx$ of $LP(\llnew)$ will give first to $\CA'\setminus S_k$ minimum delay, and then next minimum to $S_k$. 
This is exactly same as maximizing $\delay_{S_k}(\xx)$ among optimal of $LP(\ll^a)$ (where $\ll^a(S^*)=\ll^a(\CA'\setminus
S^*)>\ll^a(S_g),\ \forall g\le k-1$). Due to the fact that $f_{\ll^a,\pp^a}(S_k)=0$, at such an allocation we also have 
$m(S_k)=\pay_{S_k}(\xx,\pp^a)$. This will be same as $\pay_{S_k}(\xx,\pnew)$ due to construction of $\pnew$ from $\pp^a$ in Lemma
\ref{lem:price}. Since at every such allocation $\delay_{S_k}(\xx)$ remains the same, $\pay_{S_k}(\xx,\pnew)$ remains the same (Lemma \ref{lem:segCond1}). 

For the second part, namely $\CA'\setminus S_k$ satisfies $SC$ w.r.t. $(\llnew,\pnew)$, 
it suffices to show existence of $\ii>0$ such that $f_{\llnew,\pnew}(T) \ge 0,\ \forall T\subset \CA'\setminus S_k$. 
We will show this in Lemma \ref{lem:ii} below.
\end{proof}

Above lemma implies that if the minimizer of $f_a$ gives zero value, then it forms the next segment. One crucial task therefore is to find a
minimizer of $f_a$ efficiently. Next lemma and (\ref{eq.fa}) show that $f_a$ is a submodular function, implying its minimizer can be found in polynomial
time.

\submodular*
\begin{proof}
For ease of notation let us use $f$ to denote function $f_a$, $\ll=\ll^a$, $\pp=\pp^a$, and $\alpha$ be the dual vector that forms {\em valid} solution of $DLP(\ll^a)$ together with $\pp^a$. Recall that agent set $A$ is partitioned by equality of $\lambda^a_i$ into sets $S_1,\dots,S_{(k-1)},A'$ such that $\ll^a(S_1)<\dots<\ll^a(S_{(k-1)})<\ll^a(A')$. 

Let $S\subset T \subset A'$ and $a\not \in T$.  Define $S'=S\cup \{a\}$ and $T'=T\cup \{a\}$. 
It suffices to show the following
\[
f(S')-f(S) \geq f(T')-f(T)
\]
Let's recall the following two complementary slackness conditions.
\[
\begin{array}{llll}
\forall i \in A,\ \forall k \in \CC: & \sum_{j\in G} a_{ijk} x_{ij} \geq r_{ik} & \perp &  \alpha_{ik} \geq 0 \\
\forall i\in A, \forall j \in \CG: & \lambda_i d_{ij} \geq \sum_{k} a_{ijk} \alpha_{ik} - p_j & \perp & x_{ij} \geq 0	\\
\end{array}
\]
Using these it is easy to get the following, where $\lambda^*$ is the $\lambda_i$ of agents in $A'$ which we know is the same. 
\begin{equation}\label{submodular:eq}
\lambda^* \delay_S(\xx)=\pay_S(\xx,\pp)  +\sum_{i\in S,k}\alpha_{ik} r_{ik}
\end{equation}
For set $S\subseteq A'$ let us denote its complement within $A'$ by $\bar{S}=A'\setminus S$. From the first part of Lemma \ref{lem:lp-prop} we know that optimal solution of $LP(\ll)$ will first minimize delay of set $A'$ since it has highest $\lambda$ value. Note that, $\bar{S}=\bar{S'} \cup \{a\}$. If $\xx^S$ is an optimal solution of $LP(\ll)$ where delay of $S$ is maximized then by extensibility it can constructed as follows: within minimization for set $A'$ first minimizing delay for $\bar{S'}$, then for $a$, and lastly for $S$.
Then we have:
	\[
	\begin{array}{lll}
	f(S')-f(S) & = m_a -  \pay_{S'}(\xx^S,\pp) + \pay_{S}(\xx^S,\pp)\\
	& = m_a - \sum_{i\in S',k} \alpha_{ik}r_{ik} 
			+ \lambda^* \delay_{S'}(\xx^{S}) 
			+ \sum_{i\in S,k} \alpha_{ik}r_{ik} 
			- \lambda^* \delay_S(\xx^S)  & \text{(Using (\ref{submodular:eq}))}\\
	& = m_a - \sum_k \hat{\alpha}_{ak} r_{ak} + \lambda^* (\delay_{S'}(\xx^S)
			- \delay_{S} (\xx^{S}))\\
	& = m_a -\sum_k \hat{\alpha}_{ak} r_{ak} + \lambda^* \delay_a (\xx^{S}) & (\xx^S \text{ construction})
	\end{array}
	\]
	Similarly we get,
	$$ f(T')-f(T) = m_a -\sum_k \hat{\alpha}_{ak} r_{ak} + \lambda^* \delay_a (\xx^T) $$
	where $\xx^T$ is an optimal solution where delay of $T$ is maximized, which can be constructed by first minimizing delay for $\bar{T'}$, next for $a$, and last for $T$. Recall that $S\subset T$ and so $\bar{T} \subset \bar{S}$. We constructed $\xx^S$ and $\xx^T$ by first optimizing for $\bar{S'}$ and $\bar{T'}$ and then adding $a$. Therefore, $\delay_a(\xx^S) \leq \delay_a(\xx^T)$ and so
	\[
		f(S')-f(S)\geq f(T')-f(T)
	\]
\end{proof}

The \nextseg\ subroutine does binary search on the value of $a$ to find the one where minimizer of $f_a$ gives zero. In next few lemmas
we show why binary search is the right tool to find this critical value of $a$.  Essentially we show Lemma \ref{lem:bin} which 
has four parts and we will show them in separate lemmas.

\begin{lemma}\label{lem:f0}
If $\CA'$ satisfies $SC$ condition of Definition \ref{def:feas} at $(\llcur,\pcur)$ then $f_{\llcur,\pcur}(T)\ge 0,\ \forall T\subset \CA'$.
In other words $g(0)\ge 0$. 
\end{lemma}
\begin{proof}
Since $\CA'$ satisfies $SC$ w.r.t. $(\llcur,\pcur)$, by definition of $SC$ condition (Definition \ref{def:feas}), 
it follows that for any $T\subset \CA'$, $m(T) - \pay_T(\xx^T,\pcur)\ge 0$, where $\xx^T$ is an optimal solution of $LP(\llcur)$ where $\delay_T(\xx)$ is maximized. 
Thus, $f_{\llcur,\pcur}(T) = m(T) - \pay_T(\xx^T,\pcur)\ge 0$. 
\end{proof}

\begin{lemma}\label{lem:mona}
Assuming sufficient demand, for any $T\subset \CA'$, value $f_a(T)$ monotonically and continuously decreases with increase in $a$. Further,
for any given $a\ge 0$, $\exists a'>a$ such that $f_{a'}(A')<f_a(A')$. 
\end{lemma}
\begin{proof}
There is a unique valid price vector $\pp^a$ and $\pp^{a'}$ constituting optimal solution of $DLP(\ll^a)$
and $DLP(\ll^{a'})$. If we apply Lemma \ref{lem:price} for $\ll=\ll^a$ and $\ll'=\ll^{a'}$, then we get that $\pp^a \le \pp^{a'}$. 
Note that the partition of agents by equality of $\lambda_i$ is the same at both $\ll^a$ and $\ll^{a'}$ and further
their ordering by the value of $\lambda(S)$ is also same. Therefore,
due to the first part of Lemma \ref{lem:lp-prop} every optimal solutions of $LP(\ll^a)$ are the same $LP(\ll^{a'})$. Hence, among them the
ones minimizing $\delay_T(\xx)$ for any given $T\subset \CA'$ are the same, say $\xx^T$ is one of them. Note that $\delay_T$ remains the same at all such allocations.  
\[
f_a(T) = m(T) - \pay_T(\xx^T,\pp^a) = m(T) - \sum_{i \in T,j} x^T_{ij} p^a_j \ge m(T) - \sum_{i \in T,j} x^T_{ij} p^{a'}_j = f_{a'}(T)
\]
Since solutions of linear programs change continuously with change in parameters, the first part follows. 

For the second part, we know that $\forall a'\ge a$, valid prices $\pp^{a'}\ge \pp^a$ from the above. It suffices to show that $\exists
a' \ge a$ such that $p^{a'}_j > p^a_j$ for some good $j$ that is bought by agents of $\CA'$. Suppose not, then $\pp^{a'}=\pp^a,\ \forall a'\ge a$ since prices of the rest of the goods are fixed anyway. Let $\hpp = \pp^{a}$,
$\hll=\ll^{a}$, and $\hxx$ be an optimal solution of $LP(\hll)$. Note that $\hxx$ is an optimal of $LP(\ll^{a'})$
as well for any $a'>a$. Then, the following system has a solution for any $a' \ge a$, since $\hpp$
together with some $\aal$ gives a optimal solution of $DLP(\ll^{a'})$ that satisfies complementarity with $\hxx$. Here $\aal$ are variables, and $\hl=\hll(\CA')$.
\[
\begin{array}{ll}
\forall i\in \CA', \forall j: & \hx_{ij} > 0 \Rightarrow (\hl+\delta)d_{ij} = \sum_k \aijk \alpha_{ik} - \hp_j\\
\forall i\in \CA', \forall j: & \hx_{ij} = 0 \Rightarrow (\hl+\delta)d_{ij} \ge \sum_k \aijk \alpha_{ik} - \hp_j\\
\forall i\in \CA', \forall k: & \sum_j a_{ijk} \hx_{ij} > r_{ik} \Rightarrow \alpha_{ik} =0 \\
& \aal\ge 0
\end{array}
\]

Removing $\alpha_{ik}$ variables that are set to zero, and then applying Farkas' lemma (Lemma \ref{lem:farkas}) we get that following system in $x^o_{ij}$ variables is infeasible for any $\delta \ge 0$:
\[
\begin{array}{ll}
\forall i\in \CA',\ \forall k: & \sum_j \aijk \hx_{ij} = r_{ik} \Rightarrow \sum_j \aijk x^o_{ij} \ge 0 \\
\forall i\in \CA',\ \forall k: & \hx_{ij} =0 \Rightarrow x^o_{ij} \ge 0\\
& (\hl+\delta) \sum_{i\in \CA', j} d_{ij} x^o_{ij} + \sum_{i\in \CA',j} \hp_j x^o_{ij} <0
\end{array}
\]

Let $u\in \CA'$ be an agent and let $\x'_u$ be its optimal allocation at zero prices. Since $\forall j, \hx_{uj}\le 1$ due to
feasibility in $LP(\hll)$, it can not be optimal at zero prices due to sufficient demand condition. Hence $\sum_{j} d_{uj} x'_{uj} <
\sum_j d_{uj} \hx_{uj}$. Set $x^o_{ij}=0,\ \forall i\neq u,\ \forall j$, and
$x^o_{uj}=x'_{uj} - \hx_{uj},\ \forall j$. Clearly $\xx^o$ satisfies the first set of conditions above for all $i\in \CA'$, $i\neq u$. For agent $u$
since $\x'_u$ is optimal of (\ref{eq:coveringLP}), $\forall k$ with $\sum_j a_{ijk} \hx_{ij}=r_{ik}$, we have 
$$\sum_j \aijk x'_{ij} \ge r_{ik} = \sum_j a_{ijk} \hx_{ij} \Rightarrow \sum_j \aijk (x'_{ij} - \hx_{ij}) \ge 0 \Rightarrow \sum_j \aijk x^o_{ij} \ge 0.$$

Second set of conditions are satisfied by construction. Hence setting $\delta >
\frac{|\sum_{i\in \CA',j} \hp_j x^o_{ij} + \hl \sum_{i\in \CA', j} d_{ij}
x^o_{ij})|}{|\sum_{i\in \CA', j} d_{ij} x^o_{ij}|}$, makes the above
feasible. A contradiction. 
\end{proof}

\begin{lemma}\label{lem:rational}
There exists a rational $a_h\ge 0$ of polynomial size such that $g(a_h) \le 0$.
Further, given existence of $a>0$ for set $S\subseteq \CA'$ such that
$f_a(S)=0$, such an $a$ can be found by solving a feasibility LP of polynomial-size.
\end{lemma}
\begin{proof}
Since value of $g(a) \le f_a(S),\ \forall S\subset\CA'$, it suffices to show that there exists some rational $a_h>0$ such that $f_{a_h}(A')\le 0$. 
By definition of $\ll^a$ (\ref{eq:la}), partitions of agents by equality of
$\lambda_i$ and ordering of these groups by value of it, are the same for all $a\ge 0$. Let $\hxx$ be an optimal solution of $LP(\ll^0)$. 
From the proof of Lemma \ref{lem:mona} we have that $\pp^a$ is monotonically increasing.
By continuity of parameterized LP solutions we
know that it is increasing continuously. Therefore, it suffices to show that there is some $a\ge a_0$ such that $f_a(A')\le 0$ or in
other words $\pay_{A'}(\hxx,\pp^a)\ge m(A')$. 

To the contrary suppose not. 
Let $\hll=\ll^{0}$ and $\hpp=\pp^{0}$. Solve the following linear program where $a$, $\delta_j$s and
$\alpha_{ik}$s are variables, and $\hl=\hll(\CA')$, $c_j = \sum_{i\in A'} \hx_{ij}$,

\[
\begin{array}{lll}
\max: & \sum_j (\hp_j+\delta_j) c_j \\
s.t. \\
& \forall j: & \sum_{i\notin \CA'} \hx_{ij} > 0 \mbox{ or } \sum_i \hx_{ij} < 1 \Rightarrow \delta_j =0\\
& \forall i\in \CA', \forall j: & \hx_{ij} > 0 \Rightarrow (\hl+a)d_{ij} = \sum_k \aijk \alpha_{ik} - \hp_j -\delta_j\\
& \forall i\in \CA', \forall j: & \hx_{ij} = 0 \Rightarrow (\hl+a)d_{ij} \ge \sum_k \aijk \alpha_{ik} - \hp_j -\delta_j\\
& \forall i\in \CA', \forall k: & \sum_j a_{ijk} \hx_{ij} > r_{ik} \Rightarrow \alpha_{ik} =0 \\
& & a\ge 0,\ \ddelta\ge \zeros,\ \aal \ge \zeros
\end{array}
\]

Note that at any feasible point of the above, we have $\sum_j (\hp_j+\delta_j) c_j \ge \pay_{A'}(\hxx,\pp^a)$. It has a finite
optimal or else there is a $a\ge 0$ where $\pay_{A'}(\xx^T,\pp^a)\ge m(T)$. Let $(a^*,\ddelta^*,\aal^*)$ be an optimal solution.
By construction of above linear program $\pp^{a^*} = \hpp+\ddelta$ is the valid price vector at $\ll^{a^*}$. By
Lemma \ref{lem:mona}, we can find an $a'>a^*$ where $f_{a'}(A')<f_{a^*}(A) \Rightarrow \pay_{A'}(\hxx,\pp^{a'})>
\pay_{A'}(\hxx,\pp^{a^*})\Rightarrow \sum_j p^{a'}_j c_j > \sum_j p^{a^*}_j c_j$. If $\aal^{a'}$ together with $\pp^a$ forms 
a valid optimal of $DLP(\ll^{a'})$, then for $\ddelta'=\pp^{a^*} - \hpp$ and $\aal'=\aal^{a'}$, point $(a', \ddelta',\aal')$ is a 
feasible point in the above linear program with better objective value. This contradicts optimality of $(a^*,\ddelta^*,\aal^*)$ in it. 

From the above argument and by monotonicity of $f_a(A')$ in $a$ (Lemma \ref{lem:mona}) it follows that there exists an $a_h>0$ such that $f_{a_h}(A')=0$. Again since value of $g(a_h)\le f_{a_h}(S),\ \forall S\subseteq A'$, existence of such an $a_h$ that is a rational number of polynomial sized implies the first part of the lemma. Next we will show this for any $S\subseteq A'$ thereby proving the first and second part simultaneously. 

For $S\subseteq \CA'$, existence of $a\ge 0$ with $f_a(S)=0$ implies that valid $(\pp^a=\hpp+\ddelta, \aal^a)$ of $DLP(\ll^a)$ 
is a feasible point in the following where $\hxx$ is the optimal solution of $LP(\ll^a)$ where delay of $S$ is maximized:

\[
\begin{array}{ll}
 \forall j: & \sum_{i\notin \CA'} \hx_{ij} > 0 \mbox{ or } \sum_i \hx_{ij} < 1 \Rightarrow \delta_j =0\\
 \forall i\in \CA', \forall j: & \hx_{ij} > 0 \Rightarrow (\hl+a)d_{ij} = \sum_k \aijk \alpha_{ik} - \hp_j -\delta_j\\
 \forall i\in \CA', \forall j: & \hx_{ij} = 0 \Rightarrow (\hl+a)d_{ij} \ge \sum_k \aijk \alpha_{ik} - \hp_j -\delta_j\\
 \forall i\in \CA', \forall k: & \sum_j a_{ijk} \hx_{ij} > r_{ik} \Rightarrow \alpha_{ik} =0 \\
& \sum_{i\in S, j} (\hp_j +\delta_j) (\sum_{i\in S} \hx_{ij}) = m(S) \\
 & a\ge 0,\ \ddelta\ge \zeros,\ \aal \ge \zeros
\end{array}
\]

Hence existence of rational such $a$ of polynomial-size, and that it can be found by solving a feasibility LP follow.
\end{proof}

The next lemma follows essentially from Lemmas \ref{lem:f0}, \ref{lem:mona}, and \ref{lem:rational}.

\bin*
\begin{proof}
Part $(i)$ follows from Lemma \ref{lem:f0}. Part $(ii)$ from Lemma \ref{lem:mona} and the fact that $g(a)$ is minimum of $f_a(S)$ over
all subsets $S\subset \CA'$. Minimum of continuously decreasing functions is also continuously decreasing. Parts $(iii)$ and $(iv)$ from 
Lemma \ref{lem:rational}. 
\end{proof}

From Lemma \ref{lem:segCond} we know that at the end of \nextseg\ subroutine, when we have $S^*$ a maximal minimizer of $f_{a^*}$ and
$f_{a^*}(S^*)=0$ then $S^*$ together with previous segments $S_1,\dots,S_{k-1}$ will satisfy $BB$ and $SC$ at $(\llnew,\pnew)$. 
Next we show existence of appropriate $\ii>0$ so that $SC$ condition is satisfied for $\CA'\setminus S^*$.
This will complete the missing part of Lemma \ref{lem:segCond}.

\begin{lemma}\label{lem:ii}
If $a^*>0$ such that $g(a^*)=0$ and $S^*\subset \CA'$ be maximal set such that $f_{a^*}(S^*)=0$, then $\exists \ii>0$ rational of
polynomial size such that $SC$ condition is satisfied for $\CA'\setminus S^*$ w.r.t. $(\llnew,\pnew)$ where $\llnew=\ll^{a^*} + \ii
\ee{}$, $\pnew$ is a valid optimal of $DLP(\llnew)$, and $\ee{}$ is indicator vector of set $\CA'\setminus S^*$.
\end{lemma}
\begin{proof}
To show $SC$ for $\CA'\setminus S^*$ w.r.t. $(\llnew,\pnew)$, it suffices to show the following:
\begin{itemize}
\item $f_{\llnew,\pnew}(T)\ge 0$, $\forall T\subset \CA'\setminus S^*$.
\end{itemize}
Since $S^*$ was the maximal set where value of function $f_{a^*}$ is $0$ and minimum value of $f_{a^*}$ is zero, we have that for any $S\subset \CA'\setminus
S^*$, $f_{a^*}(S\cup S^*) >0$. If there is an $a^T> a^*$ such that $f_{a^T}(T)=0$, then by applying Lemma \ref{lem:rational} it is a rational of polynomial size and can be computed by solving a linear program, otherwise set $a^T=\infty$. Among
all of these pick the least one, lets call it $a^{min}$. It has to be strictly more than $a^*$. 

Let $\CA''=\CA'\setminus S^*$, and fix set $S\subset \CA''$. Let $\hll=\ll^{a^*}$, $\hxx$ be an optimal allocation of $LP(\hll)$ where 
$\CA''\setminus S$ gets the best, then $S$, then $S^*$ and then rest of the segments. Let $\hpp=\pp^{a^*}$. 
Similar to Lemma \ref{lem:rational} solve the following LP to compute maximum value of $a^{S}$ such that $m(S)\ge \pay_S$ when 
$\lambda_i$s of only $\CA''$ is increased. Here $b,\ddelta$ and $\aal$ variables,
$c^S_j = \sum_{i \in S} \hx_{ij},\ \forall j$, and $\hl=\hl_i,\ i\in \CA''$.
\[
\begin{array}{ll}
\max: b  \ \ \ \mbox{s.t.}\\
 \forall j: & \sum_{i\notin \CA''} \hx_{ij} > 0 \mbox{ or } \sum_i \hx_{ij} < 1 \Rightarrow \delta_j =0\\
 \forall i\in \CA'', \forall j: & \hx_{ij} > 0 \Rightarrow (\hl+b)d_{ij} = \sum_k \aijk \alpha_{ik} - \hp_j -\delta_j\\
 \forall i\in \CA'', \forall j: & \hx_{ij} = 0 \Rightarrow (\hl+b)d_{ij} \ge \sum_k \aijk \alpha_{ik} - \hp_j -\delta_j\\
 \forall i\in \CA'', \forall k: & \sum_j a_{ijk} \hx_{ij} > r_{ik} \Rightarrow \alpha_{ik} =0 \\
& \sum_{i\in S, j} (\hp_j +\delta_j) c^S_j \le m(S) \\
 & b\ge 0,\ \ddelta\ge \zeros,\ \aal \ge \zeros
\end{array}
\begin{array}{ll}
\forall (i,j)
\end{array}
\]
Clearly, $b=0$, $\delta_j=0,\ \forall j$, and $\alpha_{ik}=\ha_{ik},\forall (i,k)$ is feasible where $(\haal,\hpp)$ is the valid optimal solution of $LP(\hll)$. If there is a finite optimal of the above LP then set $a^S=b$ otherwise set $a^S=\infty$. Since $a^T>a^*$, we have $a^S>0$.
Taking $\ii$ to be less than $a^S,\ \forall S\subset \CA''$ will suffice, due to monotonicity of $f_a$
function (Lemma \ref{lem:mona}). Since $a^S$s are polynomial size, there is an $\ii$ of polynomial size.
\end{proof}

Putting everything together next we argue that at the end of the algorithm all the created segments satisfy $BB$ and $SC$ w.r.t. 
$(\llcur,\pcur)$. In other words, $(\llcur,\pcur)$ are \efeasible. 

\begin{lemma}\label{lem:alg}
The $\llcur$ and $\pcur$ obtained at the end of Algorithm \ref{alg} are \feasible\ (Definition \ref{def:feas}).
\end{lemma}
\begin{proof}
Lemmas \ref{lem:segCond} and \ref{lem:ii} imply that if at the beginning of $k^{th}$ call to $\nextseg$, w.r.t. $(\llcur,\pcur)$, 
$S_1,\dots,S_{k-1}$ satisfies $BB$ and $SC$, and $\CA'$ satisfies $SC$, then at the end of it, w.r.t. $(\llnew,\pnew)$, $S_1,\dots,
S_k$ satisfies $BB$ and $SC$, and $\CA'\setminus S_k$ satisfies $SC$. Applying this inductively, starting from $k=1$ where $\CA'=\CA$, 
and resetting $\CA'=\CA'\setminus S_k$ every time, the theorem follows. 
This is because if we made $s$ calls in total to $\nextseg$\ forming segments $S_1,\dots,S_s$, and $\CA'=\emptyset$ at the end, then
all $s$ segments satisfy $BB$ and $SC$ w.r.t. $(\llcur,\pcur)$ at the end. Thus $(\llcur,\pcur)$ are \efeasible.
\end{proof}

Next we show correctness of Algorithm \ref{alg} using Lemmas \ref{lem:feas}, \ref{lem:bin} and \ref{lem:alg}, and Theorem \ref{thm:lpl2eq},
and the next theorem follows.

\begin{theorem}\label{thm:correct}
Given a market $\CM$ satisfying extensibility and sufficient demand, Algorithm \ref{alg} returns its equilibrium allocation and
prices in time polynomial in the size of the bit description of $\CM$.
\end{theorem}
\begin{proof}
The fact that if Algorithm \ref{alg} terminates in polynomial time then it
returns equilibrium allocation and prices of market $\CM$ follows from Lemmas
\ref{lem:alg} and \ref{lem:feas}, and Theorem \ref{thm:lpl2eq}.

The question is why should the algorithm terminate, and that too in polynomial-time. Note that, every call to \nextseg\ reduces size of
active set of agents. Therefore, if the instance has $n$ agents then algorithm makes at most $n$ calls to \nextseg. Subroutine
\nextseg\ does binary search for value of $a$ between $0$ and a polynomial-sized rational. In each iteration of binary search it
minimizes a submodular function, and in each call to the submodular function we will be solving at most constantly many linear programs
of polynomial size. Thus submodular minimization can be done in polynomial time. 
Due to Lemma \ref{lem:bin} and in particular
Lemma \ref{lem:rational} value of $a$ where $g(a)$ becomes zero for some set is rational of polynomial-size. Hence overall the
binary search has to terminate in polynomial time. 
The next loop again takes at most $O(n)$ iterations, with sub-modular minimization in each. Thus the \nextseg\ subroutine terminates in
polynomial-time. Finding allocation satisfying \ref{eq:budget} of all the agents $i\in \CA$ at the end of the algorithm, given prices $\pcur$
and $\lambda$ values $\llcur$ is equivalent to solving the feasibility linear program (\ref{eq:feasLP}). Thus, overall the algorithm
terminates in polynomial time. 
\end{proof}

We get our main result using Theorem \ref{thm:correct}.

\algo*

\section{Fairness and Incentive Compatibility Properties}
\label{app:prop}
In this section we show fairness properties of our general model, and incentive compatibility properties of our algorithm as a mechanism in scheduling application. Yet another utility model is that of quasi-linear utilities, where the agent also specifies an ``exchange rate'' between delay and payments, and wants to minimize a linear combination of the two. We show in Section \ref{app:QL} that for such a utility model \emph{there is no }IC mechanism that is also Pareto optimal and anonymous, even with a single good and two agents. 

\subsection{Fairness Properties}\label{app:fair}
The first welfare theorem for traditional market models says that at equilibrium the utility vector of agents is Pareto-optimal among utility vectors at all possible feasible allocations. We first show similar result for our markets that satisfy mild condition of {\em sufficient demand} (see Definition \ref{def:ED}). 
For our model the set of feasible delay cost vectors are $$\begin{array}{l}\CD = \{(\delay_1(\xx),\dots,\delay_n(\xx))\ |\ \x_i \mbox{ is feasible in \ref{eq:CC} for each agent $i\in A$,}\\
\hspace{6cm} \mbox{and $\xx$ satisfies \ref{eq:supply} for each good $j\in G$}\}
\end{array}$$

\begin{theorem}
Given market $\CM$ satisfying {\em sufficient demand}, the delay cost vector at any of its equilibrium is Pareto-optimal in set $\CD$.
\end{theorem}
\begin{proof}
Let $(\xx^*,\pp^*)$ be an equilibrium. 
Using Lemma \ref{lem:eq2lpl} we know that for some vector $\ll^*>0$, $\xx^*$ is a solution of $LP(\ll^*)$. Let $d^*_i = \delay_i(\xx^*)$. If there is $\d\in \CD$ such that $d_i\le d^*_i,\ \forall i\in A$ with at least one strict inequality, then the allocation corresponding to $\d$ is feasible in $LP(\ll^*)$ and would give strictly lower objective value, a contradiction. 
\end{proof}

The next theorem establishes envy-freeness for the general model and follows directly from the equilibrium condition that every agent demands an optimal bundle at given prices. 
\begin{theorem}
Equilibrium allocation of a given market $\CM$ is envy-free.
\end{theorem}

Next we show that at equilibrium each agent gets a ``fair share'': the equilibrium allocation Pareto-dominates an ``equal share" allocation, where each agent gets an equal amount of each resource. This property is also known as \emph{sharing incentive} in the scheduling literature \cite{Ghodsi2011}.

\begin{theorem}
Given a market $\CM$, let $\xx$ be an allocation where agent $i$ gets $\frac{m_i}{\sum_{i\in A} m_i}$ amount of each good, i.e., $x_{ij} = \frac{m_i}{\sum_{i\in A} m_i},\ \forall i \in A, \forall j \in G$. 
Then at any equilibrium $(\xx^*,\pp^*)$ of market $\CM$, $\delay_i(\xx^*)\le \delay_i(\xx),\ \forall i \in A$. 
\end{theorem}
\begin{proof}
At equilibrium under-sold goods have price zero, and no agent spends more than its budget. This gives $$\sum_{i,j} x^*_{ij} p^*_j \le \sum_i m_i \Rightarrow \sum_j p^*_j \sum_i x^*_{ij} \le \sum_i m_i \Rightarrow \sum_j p^*_j \le \sum_i m_i$$. 

Therefore, for agent $i$, bundle $\x_i$ where amount of every good is exactly $\frac{m_i}{\sum_{i\in A} m_i}$ costs money $\sum_j x_{ij} p^*_j = \frac{m_i}{\sum_i m_i} \sum_j p^*_j \le m_i$. Thus, it is affordable at prices $\pp^*$. However, she preferred $\x^*_i$ instead, which implies either $\x'_i$ is not feasible in \ref{eq:CC} in which case $\delay_i(\x'_i)$ is infinity, or she prefers $\x^*_i$ to $\x'_i$. In either case we get $\delay_i(\xx^*)\le\delay_i(\xx)$. 
\end{proof}


\subsection{Scheduling: Algorithm as a Truthful Mechanism}\label{app:IC}

Market based mechanisms are usually not (dominant strategy) {\em incentive compatible} (IC), 
except in the large market assumption where each individual agent is too small to influence the price, 
and therefore can be assumed to act as a price taker.  
Somewhat surprisingly, we can show IC, in a certain sense, of the market based mechanism for the special case of our market that corresponds to the scheduling setting presented in Section \ref{sec:cloud}, and its generalization described in \ref{sec:specialcases} with multiple machine types. 

We show that our market based mechanism is IC  in the following sense: 
non-truthful reporting of $m_i$ and $r_{ik}$s can never result in an allocation with a lower delay cost.   
A small modification to the payments, keeping the allocation the same, makes the entire mechanism incentive compatible for the setting in which agents want to 
first minimize their delay and subject to that, minimize their payments. 

The first incentive compatibility assumes that utility of the agents is only the delay, and does not depend on the money spent (or saved). 
Such utility functions have been considered in the context of online advertising \citep{borgs2007dynamics,feldman2007budget,muthukrishnan2010stochastic}. 
It is a reflection of the fact that companies often have a given budget for procuring compute resources, 
and the agents acting on their behalf really have no incentive to save any part of this budget. 
Our model could also be applied to scenarios with virtual currency in which case the agents truly don't have any incentive to minimize payments. 

The second incentive compatibility does take payments into account, but gives a strict preference to delay over payments. Such preferences are also seen in the online advertising world, where advertisers want as many clicks as possible, and only then want to minimize payments. 
The modifications required for this are minimal, and essentially change the payment from a ``first price'' to a
``second price'' wherever required.

Yet another utility model is that of quasi-linear utilities, where the agent also specifies an ``exchange rate'' between delay and payments, and wants to minimize a linear combination of the two. 
We show in Section \ref{app:QL} that for such a utility model \emph{there is no }IC mechanism that is also Pareto optimal and anonymous, 
even with a single good and two agents. 
Pareto optimality is a benign notion of optimality that has been used as a benchmark for designing 
combinatorial auctions with budget constraints \cite{dobzinski2012multi,fiat2011single,goel2015polyhedral}. 
Anonymity is also a reasonable restriction, which disallows favoring any agent based on the identity. 
In the face of this impossibility, our mechanism offers an attractive alternative.

\subsubsection{Pure delay minimization}
\label{sec:ICFT} 
Suppose that there are $j$ independent copies of the basic scheduling setting in Section \ref{sec:cloud}, 
with the requirement of agent $i$ for the $j^\text{th}$ copy being $\rik$. 
In this section we show that our algorithm is actually IC, i.e., the agents have no incentive to misreport $m_i$s or $\rik$s, assuming that agents only want to minimize their delay cost and don't care about their payments as long as they are within the budgets. 
Note that reporting lower $m_i$ or higher $\rik$ are the only possible types of misreport.
Fixing preferences of all agents except agent $i$, consider two runs of the algorithm, one where 
agent $i$ is truthful and another where she misreports her preferences. 
In particular, say agent $i$ either reports a lower budget $m_i'$, and/or a higher requirements $\rik'$ for good $j$. 
 
Consider the first iteration in which the two runs differ, and let $(S_1, \lambda_1)$ and $(S_2,\lambda_2)$ be the segments found
respectively in the truthful and non-truthful runs in this iteration.  For any $\ll$, any $\pp$, and any set $S$ that does not contain $i$,
$f_{\pp,\ll}(S)$ remains the same between the two runs; for any set $S$ that contains $i$, $f_{\pp,\ll}(S)$ is strictly smaller in the
non-truthful run. Hence, $i$ does not belong to any of the segments found in earlier iterations, and $S_2$ necessarily contains
$i$.\footnote{Consider the possibilities where $i \notin S_2$ and note that $S_2$ cannot be the minimizer in the non-truthful run given
that $S_1$ is the minimizer in the truthful run.} Further, $\lambda_2 < \lambda_1$. 

Let $\CA'$ be the set of agents who are not in one of the segments found prior to the current iteration. By definition $\CA'$ is the same for
both the runs, and includes $i$, as argued in the previous paragraph.  Let $X^1$ and $X^2$ be respectively the allocations output
by the algorithm for the truthful and the non-truthful runs.  We will show the existence of a weakly feasible allocation $X'$ such
that (1) For every agent $i'\in \CA', i' \neq i,$ his delay in $X'$ is no higher than his delay in $X^1$, and (2) For agent
$i$, his allocation in $X'$ is the same as his allocation in $X^2$.  

This implies that $i$ is no better off in the non-truthful run, because of the following reasoning. The total delay of all
the agents in $\CA'$ is minimized in $X^1$,
therefore the total delay of all the agents in $\CA'$ cannot be lower in allocation $X'$, even when the
delay for agent $i$ is calculated using only his actual requirements.  
Since no other agent has a higher delay in $X'$, it is
impossible for $i$ to get a lower delay.

It remains to show the existence of $X'$ as claimed. We define $X'$ differently based on whether the agent is in $S_2$ or not. 
\begin{description}
	\item [Case 1: $i'\in S_2$:] In this case, $\x'_{i'} = \x^2_{i'}$.  This satisfies the second requirement since $i\in S_2$.
	Since $\lambda_2 <\lambda_1$, every agent in $S_2$ faces a smaller price, for every copy $j$ and every time slot in which she is allocated.
	For $i'\neq i$, given the same budget and the same requirements, this actually implies that her delay in $X^2$ is strictly smaller than her delay in $X^1$.
	\item[ Case 2: $i' \notin S_2$:] In this case, we first start with the allocation $X^1$, in the slots  $[1,r_j(B\setminus S_2)]$ for each copy $j$.  
	Note that these slots have not been allocated at all in Case 1. 
	Consider the total deficit after this allocation.  
	This must be equal to the total amount of slots in  $[1,r_j(B\setminus S_2)]$ that are allocated to agents in $S_2$
	by $X^1$, because of feasibility of $X^1$.  
	Now re-allocate these empty slots in $[1,r_j(B\setminus S_2)]$ to make up for the remaining requirement of these agents, and note that this can only lower the delay. 
\end{description}

\subsubsection{Secondary preference for payments} 

In this section, we consider the utility model where an agent wants to first minimize her flow-time, and subject to that,  wants to further minimize her payments.  
We keep the same allocation as Algorithm \ref{alg.scheduling}, but change the payments of some agents, and show that this is IC. 

We first define the set of agents whose payments will be modified. 
Recall that Algorithm \ref{alg.scheduling} outputs a sequence of segments, where each segment corresponds to a pair $(\lambda,S)$. 
Call an agent \emph{marginal} if he gets the latest slots in his segment. This includes agents who are in singleton segments, as well as agents who just happen to get such an allocation even though they are in a segment with other agents. We modify the payments of only the marginal agents; 
all non-marginal agents pay their budget. 

\begin{lemma}
	Any non-marginal agent gets a strictly higher delay cost for any misreport of his information. 
\end{lemma}
\begin{proof}
	Consider the proof of incentive compatibility for only delay cost minimization in Section \ref{sec:ICFT}, and the notation therein. Note that if $S_2 \neq \{i'\}$, then the delay cost of $i'$ strictly increases. Now suppose $S_2= \{i'\}$. In the new allocation $f^2$, agent $i'$ gets the latest slots among all agents in $B$. Since $i'$ is not a marginal agent, he was getting a strictly better allocation in $f^1$, and the lemma follows. 
\end{proof}
This shows that the mechanism is IC for non-marginal agents, even with their payments equal to the budgets. 

We now define the modification to payments for marginal agents. As in Section \ref{sec:ICFT}, misreports can still not get a better delay cost for marginal agents, since the allocation remains the same. The only possibility is that 
misreporting can decrease payments, while keeping the delay cost the same. Marginal agents can decrease their budgets, still get the same allocation, and pay less in the equilibrium payment. 
This has a limit; at some lower budget declaration, they get ``merged'' with a previous segment, and  any further lowering of the budget will strictly lower their delay cost. 
\emph{The payment of a marginal agent is defined to be the infimum 
of all budget declarations for which the lower segments are unaltered,  
i.e., the run of the algorithm up to the previous segment remains unchanged}.

We now argue that this mechanism is IC, for marginal agents. 
We only need to consider misreports that don't change the allocation, since those that do only give a higher delay cost. 
Among these, misreporting the budget clearly has no effect on the payment. Finally, we argue that reporting a higher $r_{ik}$ can only lead to a higher payment. 
This is because the budget at which the agent merges with the previous segment happens at a higher value, as can be seen from the formula for $\lambda_S$. 

\subsection{Quasi Linear Utility Model}\label{app:QL}
In this section we consider a quasi linear utility model for the agents. In this model, agents can choose to tradeoff payment for delay cost, as specified by an ``exchange rate'', denoted by $\eta_i$, for agent $i$. 
We consider the design of incentive compatible (IC) auctions, that are also Pareto optimal. 
In the related literature of IC auctions for combinatorial auctions with budget constraints, this has been adopted 
as the standard notion of optimality. The usual notion of social welfare is ill fitted for the case of budgets.\footnote{Of course, the revenue objective is also widely considered, and continues to make sense even in the presence of budgets.} 
 
As in Section \ref{sec:model} let the allocation of agent $i$ for good $j$ denoted by $x_{ij}$, but now we don't have prices for the slots. 
Instead we simply have a payment for each agent, denoted by $payment(i)$ for agent $i$. The allocation and the payments are together called the outcome of the auction. Agent $i$ now wants to minimize the objective 
\[ \sum_j { d_{ij}   {x_{ij}  } } + \eta_i  payment(i) . \]

A type of an agent is its budget $m_i$, its covering constraints \ref{eq:CC}, and its $\eta_i$. 
An auction is (dominant strategy) IC if for any agent, misreporting its type does not lead to an outcome with a lower objective, no matter what the other agents report. 
An outcome is \emph{Pareto optimal} if for no other outcome, 
\begin{enumerate}
	\item all agents, including the auctioneer, are at least as well off as in the given outcome, and 
	\item at least one agent is strictly better off. 
\end{enumerate}
The  auctioneer's objective is to simply maximize the sum of all the payments. 

We also restrict the auction to be \emph{anonymous}, which means that the auction cannot rely on 
the identity of the agents. Formally, an auction is anonymous if it is invariant under all permutations of agent identities. 

The main result of this section is an impossibility. 
\begin{theorem}
	\label{thm:QLimpossibility}
	There is no IC, Pareto optimal,  and anonymous auction for our scheduling problem with quasi linear utilities, for the case of a single good and two agents. 
\end{theorem}
Since a single good and two agents is the most basic case, an impossibility follows for all generalizations as well. 

The theorem follows from a reduction to a combinatorial auction with additive valuations, and an impossibility result of \citet{dobzinski2012multi}. Consider an auction for a single \emph{divisible} item, with budget constraints. 
Agent $i$ has valuation of $v_i$ per unit quantity of the item, and a budget $m_i$. 
The outcome of the auction is an allocation $x_i$ and payment $payment(i)$ for agent $i$, 
such that $\sum_i x_i \leq 1$ and $x_i \in [0,1]$. 
The utility of agent $i$ is $v_i x_i - payment(i)$, and the budget constraint as before is that $p_i \leq m_i$. 
IC and Pareto optimality are as before, and we need an additional notion of \emph{individual rationality} (IR): 
$v_i x_i - payment(i)\geq 0$. \citet{dobzinski2012multi} showed the following impossibility. 
\begin{theorem}[\citet{dobzinski2012multi} ]
	There is no IC, Pareto optimal, IR and anonymous auction for auctioning a single divisible good to 2 agents with budget constraints. 
\end{theorem}

\begin{proof}[of Theorem \ref{thm:QLimpossibility}]
Consider an instance of the scheduling problem of Section \ref{sec:cloud} with a single machine and two agents, where each agent requires 1 unit of the good. 
Pareto optimality implies that goods are not wasted, so the entire first two slots are completely allocated. 
If agent $i$ gets $x_i$ units of the slot $t=1$, then his delay cost is $x_i + 2(1-x_i)$. 
His objective is then 
\[ \textstyle x_i + 2(1-x_i) + \eta_i payment(i) = 2 -\eta_i\left(   \frac{1}{\eta_i}  x_i - payment(i) \right)   .\]
Minimizing this objective is equivalent to maximizing $ \tfrac{1}{\eta_i}  x_i - payment(i)$, which is exactly as in 
the single divisible good auction with $v_i = \tfrac{1}{\eta_i}$. 
We also show that the IC constraint for the scheduling problem implies the IR constraint for the divisible good case. 
If the IR constraint is violated, i.e., $\tfrac{1}{\eta_i} x_i < payment(i)$, then the value of the objective of agent $i$ for this outcome is strictly smaller than 2. 
Then the agent is better off stating a budget of 0. This will force his payment to 0. 
The worst delay cost he can get is 2, so his total objective value is $2$.

Therefore, an IC, Pareto optimal, and anonymous auction for our scheduling problem implies an 
IC, Pareto optimal, IR, and anonymous auction for the divisible good case, and the theorem follows. 
\end{proof}



\end{document}

%% file: vijay_edited.tex
\usepackage{verbatim}

\newcommand{\nfrac}{\nicefrac}

\long\def\symbolfootnote[#1]#2{\begingroup%
\def\thefootnote{\fnsymbol{footnote}}\footnote[#1]{#2}\endgroup}

\newcommand{\CG}{\mbox{$G$}}

\newcommand{\CM}{\mbox{${\mathcal M}$}}
\newcommand{\CC}{\mbox{${\mathcal C}$}}

\newcommand{\CD}{\mbox{${\mathcal D}$}}

\newcommand{\la}{\leftarrow}

\newcommand{\ra}{\rightarrow}
\newcommand{\R}{\mathbb{R}}
\newcommand{\Rplus}{\R_+}
\newcommand{\integers}{\mathbb{Z}}

\newcommand{\li}{{\lambda_i}}

\renewcommand{\tt}{\mbox{\boldmath $t$}}

\newcommand{\yy}{\mbox{$Y$}}
\newcommand{\zz}{\mbox{\boldmath $z$}}

\newcommand{\x}{\mbox{\boldmath $x$}}
\renewcommand{\d}{\mbox{\boldmath $d$}}
\renewcommand{\r}{\mbox{\boldmath $r$}}
\newcommand{\y}{\mbox{\boldmath $y$}}

\newcommand{\B}{\mbox{\boldmath $\beta$}}

\newcommand{\CA}{\mbox{$A$}}

\newcommand{\ones}{{\mbox{\boldmath $1$}}}
\newcommand{\zeros}{{\mbox{\boldmath $0$}}}
\newcommand{\ee}[1]{{\mbox{\boldmath $1$}_{#1}}}

\newcommand{\haa}{{\mbox{\boldmath $\hat{\alpha }$}}}
\newcommand{\hll}{{\mbox{\boldmath $\hat{\lambda}$}}}
\renewcommand{\ll}{\mbox{\boldmath $\lambda$}}
\newcommand{\ddelta}{\mbox{\boldmath $\delta$}}
\newcommand{\ii}{\mbox{$\epsilon$}}
\newcommand{\llcur}{\mbox{$\ll^{cur}$}}
\newcommand{\aacur}{\mbox{$\aal^{cur}$}}
\newcommand{\acur}{\mbox{$\alpha^{cur}$}}
\newcommand{\pcur}{\mbox{$\pp^{cur}$}}
\newcommand{\xxcur}{\mbox{$X^{cur}$}}

\newcommand{\llnew}{\mbox{$\ll^{new}$}}
\newcommand{\pnew}{\mbox{$\pp^{new}$}}

\newcommand{\aanew}{\mbox{$\aal^{new}$}}
\newcommand{\anew}{\mbox{$\alpha^{new}$}}
\newcommand{\aal}{{\mbox{\boldmath $\alpha$}}}
\newcommand{\hp}{\mbox{$\hat{p}$}}
\newcommand{\hl}{\mbox{$\hat{\lambda}$}}
\newcommand{\hgm}{\mbox{$\hat{\gamma}$}}
\newcommand{\hbt}{\mbox{$\hat{\beta}$}}
\newcommand{\ha}{\mbox{$\hat{\alpha}$}}
\newcommand{\haal}{\mbox{$\hat{\aal}$}}
\newcommand{\delay}{\mbox{\rm delay}}
\newcommand{\BB}{\mbox{\rm BB}}
\newcommand{\SC}{\mbox{\rm SC}}
\newcommand{\pay}{\mbox{\rm pay}}

\newcommand{\alik}{\mbox{$\alpha_{ik}$}}

\newcommand{\feasible}{proper }
\newcommand{\efeasible}{{\em proper }}

\newcommand{\bpp}{\mbox{\boldmath $\bar{p}$}}
\newcommand{\bxx}{\mbox{$\bar{X}$}}

\newcommand{\pp}{\mbox{\boldmath $p$}}

\newcommand{\pb}{\mbox{\boldmath $b$}}

\newcommand{\plow}{p_{\textrm{low}}}

\newcommand{\xx}{X}

\DeclareMathOperator*{\argmax}{arg\,max}
\DeclareMathOperator*{\argmin}{arg\,min}

\newcommand{\rik}{r_{ik}}
\newcommand{\aijk}{a_{ijk}}

\newcommand{\dij}{d_{ij}}
\newcommand{\xij}{x_{ij}}
\newcommand{\xit}{x_{it}}
\newcommand{\price}{\mbox{\rm price}}
\newcommand{\hxx}{{\mbox{$\hat{X}$}}}
\newcommand{\hhx}{{\mbox{$\hat{\x}$}}}
\newcommand{\hx}{{\mbox{$\hat{x}$}}}
\newcommand{\hpp}{{\mbox{$\hat{\pp}$}}}

\newcommand{\nextseg}{NextSeg}
\newcommand{\supplyrespecting}{supply respecting\xspace}
\newcommand{\bestfor}{jointly optimal for\xspace}

\newcommand{\rh}[1]{\textcolor{black}{#1}}
\newcommand{\changed}{}
\newcommand{\stopchanged}{}

\newcommand{\twopartdef}[4]
{
	\left\{
	\begin{array}{lll}
		#1 & #2 \\
		#3 &  #4 
	\end{array}
	\right.
}

%% file: coveringME.bbl
\begin{thebibliography}{58}
\providecommand{\natexlab}[1]{#1}
\providecommand{\url}[1]{\texttt{#1}}
\expandafter\ifx\csname urlstyle\endcsname\relax
  \providecommand{\doi}[1]{doi: #1}\else
  \providecommand{\doi}{doi: \begingroup \urlstyle{rm}\Url}\fi

\bibitem[Babaioff et~al.(2014)Babaioff, Lucier, Nisan, and Leme]{BabaioffLNL14}
Moshe Babaioff, Brendan Lucier, Noam Nisan, and Renato~Paes Leme.
\newblock On the efficiency of the {W}alrasian mechanism.
\newblock In \emph{{ACM} Conference on Economics and Computation, {EC}}, pages
  783--800, 2014.

\bibitem[Badanidiyuru et~al.(2012)Badanidiyuru, Kleinberg, and
  Singer]{badanidiyuru2012learning}
Ashwinkumar Badanidiyuru, Robert Kleinberg, and Yaron Singer.
\newblock Learning on a budget: posted price mechanisms for online procurement.
\newblock In \emph{Proceedings of the 13th ACM Conference on Electronic
  Commerce}, pages 128--145. ACM, 2012.

\bibitem[Bein and Brucker(1985)]{SeriesParallel}
Wolfgang~W. Bein and Peter Brucker.
\newblock Minimum cost flow algorithms for series-parallel networks.
\newblock \emph{Discrete Applied Mathematics}, 10\penalty0 (2):\penalty0 117 --
  124, 1985.

\bibitem[Bhattacharya et~al.(2010)Bhattacharya, Conitzer, Munagala, and
  Xia]{bhattacharya2010incentive}
Sayan Bhattacharya, Vincent Conitzer, Kamesh Munagala, and Lirong Xia.
\newblock Incentive compatible budget elicitation in multi-unit auctions.
\newblock In \emph{Proceedings of the twenty-first annual ACM-SIAM symposium on
  Discrete Algorithms}, pages 554--572. SIAM, 2010.

\bibitem[Borgs et~al.(2007)Borgs, Chayes, Immorlica, Jain, Etesami, and
  Mahdian]{borgs2007dynamics}
Christian Borgs, Jennifer Chayes, Nicole Immorlica, Kamal Jain, Omid Etesami,
  and Mohammad Mahdian.
\newblock Dynamics of bid optimization in online advertisement auctions.
\newblock In \emph{Proceedings of the 16th international conference on World
  Wide Web}, pages 531--540. ACM, 2007.

\bibitem[Brainard and Scarf(2000)]{BSAD}
W.~C. Brainard and H.~E. Scarf.
\newblock How to compute equilibrium prices in 1891.
\newblock \emph{Cowles Foundation Discussion Paper}, 1270, 2000.

\bibitem[Budish(2011)]{BudishJPE11}
Eric Budish.
\newblock The combinatorial assignment problem: Approximate competitive
  equilibrium from equal incomes.
\newblock \emph{Journal of Political Economy}, 119\penalty0 (6):\penalty0 1061
  -- 1103, 2011.

\bibitem[Budish and Kessler(2014)]{Budish2014changing}
Eric~B Budish and Judd~B Kessler.
\newblock Changing the course allocation mechanism at wharton.
\newblock \emph{Chicago Booth Research Paper}, \penalty0 (15-08), 2014.

\bibitem[Che and Gale(2000)]{che2000optimal}
Yeon-Koo Che and Ian Gale.
\newblock The optimal mechanism for selling to a budget-constrained buyer.
\newblock \emph{Journal of Economic theory}, 92\penalty0 (2):\penalty0
  198--233, 2000.

\bibitem[Chen et~al.(2009)Chen, Dai, Du, and Teng]{Chen.plc}
X.~Chen, D.~Dai, Y.~Du, and S.-H. Teng.
\newblock Settling the complexity of {A}rrow-{D}ebreu equilibria in markets
  with additively separable utilities.
\newblock In \emph{FOCS}, 2009.

\bibitem[Cheung et~al.(2013)Cheung, Cole, and Devanur]{cheung2013tatonnement}
Yun~Kuen Cheung, Richard Cole, and Nikhil Devanur.
\newblock Tatonnement beyond gross substitutes?: gradient descent to the
  rescue.
\newblock In \emph{Proceedings of the forty-fifth annual ACM symposium on
  Theory of computing}, pages 191--200. ACM, 2013.

\bibitem[Codenotti et~al.(2006)Codenotti, Saberi, Varadarajan, and Ye]{CSVY}
B.~Codenotti, A.~Saberi, K.~Varadarajan, and Y.~Ye.
\newblock Leontief economies encode two-player zero-sum games.
\newblock In \emph{SODA}, 2006.

\bibitem[Codenotti et~al.(2005)Codenotti, McCune, and Varadarajan]{CMV}
Bruno Codenotti, Benton McCune, and Kasturi Varadarajan.
\newblock Market equilibrium via the excess demand function.
\newblock In \emph{{ACM} Symposium on the Theory of Computing}, 2005.

\bibitem[Cole and Fleischer(2008)]{cole2008fast}
Richard Cole and Lisa Fleischer.
\newblock Fast-converging tatonnement algorithms for one-time and ongoing
  market problems.
\newblock In \emph{Proceedings of the fortieth annual ACM symposium on Theory
  of computing}, pages 315--324, 2008.

\bibitem[Cole et~al.(2013)Cole, Gkatzelis, and Goel]{ColeGG13}
Richard Cole, Vasilis Gkatzelis, and Gagan Goel.
\newblock Mechanism design for fair division: allocating divisible items
  without payments.
\newblock In \emph{{ACM} Conference on Electronic Commerce, {EC} '13}, pages
  251--268, 2013.

\bibitem[Deng et~al.(2002)Deng, Papadimitriou, and Safra]{DPS}
Xiaotie Deng, Christos Papadimitriou, and Shmuel Safra.
\newblock On the complexity of equilibria.
\newblock In \emph{Proceedings of ACM Symposium on Theory of Computing}, 2002.

\bibitem[Devanur and Kannan(2008)]{DK08}
N.~Devanur and R.~Kannan.
\newblock Market equilibria in polynomial time for fixed number of goods or
  agents.
\newblock In \emph{FOCS}, pages 45--53, 2008.

\bibitem[Devanur and Vazirani(2004)]{DV}
N.~Devanur and V.~V. Vazirani.
\newblock The spending constraint model for market equilibrium: Algorithmic,
  existence and uniqueness results.
\newblock In \emph{Proceedings of 36th STOC}, 2004.

\bibitem[Devanur et~al.(2008)Devanur, Papadimitriou, Saberi, and
  Vazirani]{DPSV}
N.~Devanur, C.~H. Papadimitriou, A.~Saberi, and V.~V. Vazirani.
\newblock Market equilibrium via a primal-dual algorithm for a convex program.
\newblock \emph{JACM}, 55\penalty0 (5), 2008.

\bibitem[Devanur et~al.(2015)Devanur, Garg, Mehta, Vazirani, and
  Yazdanbod]{devanur2015market}
Nikhil Devanur, Jugal Garg, Ruta Mehta, Vijay~V Vazirani, and Sadra Yazdanbod.
\newblock A market for scheduling, with applications to cloud computing.
\newblock \emph{arXiv preprint arXiv:1511.08748}, 2015.

\bibitem[Dobzinski et~al.(2012)Dobzinski, Lavi, and Nisan]{dobzinski2012multi}
Shahar Dobzinski, Ron Lavi, and Noam Nisan.
\newblock Multi-unit auctions with budget limits.
\newblock \emph{Games and Economic Behavior}, 74\penalty0 (2):\penalty0
  486--503, 2012.

\bibitem[Etessami and Yannakakis(2010)]{EY07}
K.~Etessami and M.~Yannakakis.
\newblock On the complexity of {N}ash equilibria and other fixed points.
\newblock \emph{{SIAM} Journal on Computing}, 39(6):\penalty0 2531--2597, 2010.

\bibitem[Feldman et~al.(2007)Feldman, Muthukrishnan, Pal, and
  Stein]{feldman2007budget}
Jon Feldman, S~Muthukrishnan, Martin Pal, and Cliff Stein.
\newblock Budget optimization in search-based advertising auctions.
\newblock In \emph{Proceedings of the 8th ACM conference on Electronic
  commerce}, pages 40--49. ACM, 2007.

\bibitem[Fiat et~al.(2011)Fiat, Leonardi, Saia, and Sankowski]{fiat2011single}
Amos Fiat, Stefano Leonardi, Jared Saia, and Piotr Sankowski.
\newblock Single valued combinatorial auctions with budgets.
\newblock In \emph{Proceedings of the 12th ACM conference on Electronic
  commerce}, pages 223--232. ACM, 2011.

\bibitem[Garg et~al.(2012)Garg, Mehta, Sohoni, and Vazirani]{GMSV}
Jugal Garg, Ruta Mehta, Milind Sohoni, and Vijay~V. Vazirani.
\newblock A complementary pivot algorithm for market equilibrium under
  separable piecewise-linear concave utilities.
\newblock In \emph{{ACM} Symposium on the Theory of Computing}, 2012.

\bibitem[Garg et~al.(2014)Garg, Mehta, and Vazirani]{GMV}
Jugal Garg, Ruta Mehta, and Vijay~V. Vazirani.
\newblock Dichotomies in equilibrium computation, and complementary pivot
  algorithms for a new class of non-separable utility functions.
\newblock In \emph{STOC}, 2014.

\bibitem[Garg et~al.(2017)Garg, Mehta, Vazirani, and Yazdabod]{GargFIXP}
Jugal Garg, Ruta Mehta, Vijay~V. Vazirani, and Sadra Yazdabod.
\newblock Settling the complexity of {L}eontief and {PLC} exchange markets
  under exact and approximate equilibria.
\newblock In \emph{{ACM} Symposium on the Theory of Computing}, 2017.
\newblock To appear.

\bibitem[Garg and Kapoor(2004)]{GK}
R.~Garg and S.~Kapoor.
\newblock Auction algorithms for market equilibrium.
\newblock In \emph{Proceedings of 36th STOC}, 2004.

\bibitem[Ghodsi et~al.(2011)Ghodsi, Zaharia, Hindman, Konwinski, Shenker, and
  Stoica]{Ghodsi2011}
Ali Ghodsi, Matei Zaharia, Benjamin Hindman, Andy Konwinski, Scott Shenker, and
  Ion Stoica.
\newblock Dominant resource fairness: Fair allocation of multiple resource
  types.
\newblock In \emph{Proceedings of the 8th USENIX Conference on Networked
  Systems Design and Implementation}, NSDI'11, pages 323--336, 2011.

\bibitem[Goel et~al.(2015)Goel, Mirrokni, and Leme]{goel2015polyhedral}
Gagan Goel, Vahab Mirrokni, and Renato~Paes Leme.
\newblock Polyhedral clinching auctions and the adwords polytope.
\newblock \emph{Journal of the ACM (JACM)}, 62\penalty0 (3):\penalty0 18, 2015.

\bibitem[Hurwicz(1972)]{Hurwicz1972}
Leonid Hurwicz.
\newblock On informationally decentralized systems.
\newblock In C.~B. McGuire and Roy Radner, editors, \emph{Decision and
  Organization: A Volume in Honor of Jacob Marschak}. North-Holland, Amsterdam,
  1972.

\bibitem[Jain(2007)]{JainAD}
K.~Jain.
\newblock A polynomial time algorithm for computing the {A}rrow-{D}ebreu market
  equilibrium for linear utilities.
\newblock \emph{{SIAM} Journal on Computing}, 37\penalty0 (1):\penalty0
  306--318, 2007.

\bibitem[Jain and Vazirani(2010)]{jain2010eisenberg}
Kamal Jain and Vijay~V Vazirani.
\newblock Eisenberg--{G}ale markets: Algorithms and game-theoretic properties.
\newblock \emph{Games and Economic Behavior}, 70\penalty0 (1):\penalty0
  84--106, 2010.

\bibitem[Kakutani(1941)]{kakutani}
Shizuo Kakutani.
\newblock A generalization of {B}rouwer's fixed point theorem.
\newblock \emph{Duke Mathematical Journal}, 8\penalty0 (3):\penalty0 457--459,
  1941.

\bibitem[Kelly and Vazirani(2002)]{KV}
F.~P. Kelly and V.~V. Vazirani.
\newblock Rate control as a market equilibrium.
\newblock Unpublished manuscript., 2002.
\newblock URL \url{http://www.cc.gatech.edu/~vazirani/KV.pdf}.

\bibitem[Laffont and Robert(1996)]{laffont1996optimal}
Jean-Jacques Laffont and Jacques Robert.
\newblock Optimal auction with financially constrained buyers.
\newblock \emph{Economics Letters}, 52\penalty0 (2):\penalty0 181--186, 1996.

\bibitem[Lavi and Swamy(2009)]{LaviSwamy}
Ron Lavi and Chaitanya Swamy.
\newblock Truthful mechanism design for multidimensional scheduling via cycle
  monotonicity.
\newblock \emph{Games and Economic Behavior}, 67\penalty0 (1), 2009.

\bibitem[Mas-Colell et~al.(1995)Mas-Colell, Whinston, and
  Green]{MasColell-book}
Andreu Mas-Colell, Michael~Dennis Whinston, and Jerry~R Green.
\newblock \emph{Microeconomic theory}, volume~1.
\newblock Oxford university press New York, 1995.

\bibitem[Megiddo and Papadimitriou(1991)]{megiddo1991total}
Nimrod Megiddo and Christos~H Papadimitriou.
\newblock On total functions, existence theorems and computational complexity.
\newblock \emph{Theoretical Computer Science}, 81\penalty0 (2):\penalty0
  317--324, 1991.

\bibitem[Mehta et~al.(2005)Mehta, Saberi, Vazirani, and Vazirani]{MSVV.design}
A.~Mehta, A.~Saberi, U.~Vazirani, and V.~Vazirani.
\newblock Adwords and generalized on-line matching.
\newblock In \emph{FOCS}, 2005.

\bibitem[Moulin(2004)]{Moulin-book}
Herv{\'e} Moulin.
\newblock \emph{Fair division and collective welfare}.
\newblock MIT press, 2004.

\bibitem[Muthukrishnan et~al.(2010)Muthukrishnan, P{\'a}l, and
  Svitkina]{muthukrishnan2010stochastic}
S~Muthukrishnan, Martin P{\'a}l, and Zoya Svitkina.
\newblock Stochastic models for budget optimization in search-based
  advertising.
\newblock \emph{Algorithmica}, 58\penalty0 (4):\penalty0 1022--1044, 2010.

\bibitem[Myerson(1981)]{Myerson}
Roger~B. Myerson.
\newblock Optimal auction design.
\newblock \emph{Mathematics of Operations Research}, 6\penalty0 (1):\penalty0
  58--73, 1981.

\bibitem[Nisan and Ronen(2001)]{NisanRonen}
N.~Nisan and A.~Ronen.
\newblock Algorithmic mechanism design.
\newblock \emph{Games and Economic Behavior}, 35\penalty0 (1--2):\penalty0
  166--196, 2001.

\bibitem[Nisan et~al.(2009)Nisan, Bayer, Chandra, Franji, Gardner, Matias,
  Rhodes, Seltzer, Tom, Varian, and Zigmond]{nisan2009google}
Noam Nisan, Jason Bayer, Deepak Chandra, Tal Franji, Robert Gardner, Yossi
  Matias, Neil Rhodes, Misha Seltzer, Danny Tom, Hal Varian, and Dan Zigmond.
\newblock Google’s auction for {TV} ads.
\newblock \emph{Automata, Languages and Programming}, pages 309--327, 2009.

\bibitem[Orlin(2010)]{orlin2010improved}
James~B Orlin.
\newblock Improved algorithms for computing {F}isher's market clearing prices:
  computing {F}isher's market clearing prices.
\newblock In \emph{Proceedings of the forty-second ACM symposium on Theory of
  computing}, pages 291--300. ACM, 2010.

\bibitem[Rubinstein(2016)]{Rubinstein}
A.~Rubinstein.
\newblock Personal communication, 2016.

\bibitem[Rustichini et~al.(1994)Rustichini, Satterthwaite, and
  Williams]{Rustichini1994convergence}
Aldo Rustichini, Mark~A Satterthwaite, and Steven~R Williams.
\newblock Convergence to efficiency in a simple market with incomplete
  information.
\newblock \emph{Econometrica: Journal of the Econometric Society}, pages
  1041--1063, 1994.

\bibitem[Scarf(1973)]{scarf.book}
H.~Scarf.
\newblock \emph{The Computation of Economic Equilibria}.
\newblock Yale University Press, 1973.

\bibitem[Schrijver(1986)]{Sch-book}
A.~Schrijver.
\newblock \emph{Theory of Linear and Integer Programming}.
\newblock John Wiley \& Sons, New York, NY, 1986.

\bibitem[Schrijver(2003)]{schrijver.book}
A.~Schrijver.
\newblock \emph{Combinatorial Optimization}.
\newblock Springer-Verlag, 2003.

\bibitem[Shoven and Whalley(1992)]{Shoven}
J.~B. Shoven and J.~Whalley.
\newblock \emph{Applying general equilibrium}.
\newblock Cambridge University Press, 1992.

\bibitem[Singer(2010)]{singer2010budget}
Yaron Singer.
\newblock Budget feasible mechanisms.
\newblock In \emph{Foundations of Computer Science (FOCS), 2010 51st Annual
  IEEE Symposium on}, pages 765--774. IEEE, 2010.

\bibitem[Smale(1976)]{Smale}
S.~Smale.
\newblock A convergent process of price adjustment and global {N}ewton methods.
\newblock \emph{Journal of Mathematical Economics}, 3(2):\penalty0 107--120,
  1976.

\bibitem[Vazirani and Yannakakis(2011)]{VY}
Vijay~V. Vazirani and Mihalis Yannakakis.
\newblock Market equilibrium under separable, piecewise-linear, concave
  utilities.
\newblock \emph{Journal of ACM}, 58\penalty0 (3):\penalty0 10:1--10:25, 2011.

\bibitem[V{\'e}gh(2012{\natexlab{a}})]{vegh2012concave}
L{\'a}szl{\'o}~A V{\'e}gh.
\newblock Concave generalized flows with applications to market equilibria.
\newblock In \emph{Foundations of Computer Science (FOCS), 2012 IEEE 53rd
  Annual Symposium on}, pages 150--159. IEEE, 2012{\natexlab{a}}.

\bibitem[V{\'e}gh(2012{\natexlab{b}})]{vegh2012strongly}
L{\'a}szl{\'o}~A V{\'e}gh.
\newblock Strongly polynomial algorithm for a class of minimum-cost flow
  problems with separable convex objectives.
\newblock In \emph{Proceedings of the forty-fourth annual ACM symposium on
  Theory of computing}, pages 27--40. ACM, 2012{\natexlab{b}}.

\bibitem[Wu and Zhang(2007)]{wu2007proportional}
Fang Wu and Li~Zhang.
\newblock Proportional response dynamics leads to market equilibrium.
\newblock In \emph{Proceedings of the thirty-ninth annual ACM symposium on
  Theory of computing}, pages 354--363. ACM, 2007.

\end{thebibliography}
